\documentclass[nofootinbib,onecolumn ,aps,amsmath,amssymb,longbibliography,a4paper,superscriptaddress,tightenlines,notitlepage,12pt]{revtex4-2}

\usepackage{xcolor}
\usepackage[T1]{fontenc}
\usepackage{hyperref}
\usepackage[english]{babel}
\usepackage{amsthm}
\usepackage{amsmath}
\usepackage{times}
\usepackage{amsfonts}
\usepackage{amssymb}
\usepackage{thmtools}
\usepackage{thm-restate}
\usepackage{enumitem}
\usepackage{amscd}
\usepackage{graphicx}
\graphicspath{{tikz/}}
\usepackage[all]{xy}
\usepackage{tikz}
 \usetikzlibrary{matrix}
 \usetikzlibrary{arrows}
 \usetikzlibrary{decorations.pathreplacing}

 \usepackage{comment}

\usepackage{ytableau}

\usepackage[english]{babel}
\usepackage[font=small,format=plain,labelfont=bf,up]{caption}
\usepackage{bbm}

\usepackage{environ}
\NewEnviron{myalign}{
\begin{equation}\begin{split}
  \BODY
\end{split}\end{equation}
}

\newtheorem{theorem}{Theorem}
\newtheorem{corollary}{Corollary}
\newtheorem{lemma}{Lemma}
\newtheorem{proposition}{Proposition}
\newtheorem{definition}{Definition}
\newtheorem{remark}{Remark}

\newcommand{\myleft}{\mathopen{}\mathclose\bgroup\left}
\newcommand{\myright}{\aftergroup\egroup\right}

\renewcommand{\i}{\ensuremath\mathrm{i}}
\newcommand{\id}{\ensuremath\mathrm{id}}

\DeclareMathOperator{\Tr}{Tr}

\DeclareMathOperator{\rank}{rank}
\DeclareMathOperator{\vecmap}{vec}
\DeclareMathOperator{\matmap}{mat}

\DeclareMathOperator{\Comm}{Comm}
\DeclareMathOperator{\End}{End}

\DeclareMathOperator{\Hom}{Hom}

\def\FF{\mathbb{F}}

\DeclareMathOperator{\sign}{sign}

\DeclareMathOperator{\Ad}{Ad}% Adjoint rep.: Ad_U[X] = U X U^\dagger

\newcommand{\MO}{\mathrm{M}} % moment operator macro

\DeclareMathOperator{\U}{U}%unitary group U(n)
\DeclareMathOperator{\Cl}{Cl}% Clifford group Cl(n) \subset U(n)
\DeclareMathOperator{\SU}{SU}%special unitary group SU(n)

\DeclareMathOperator\tr{Tr}

\DeclareMathOperator\GL{GL}

\newcommand{\norm}[1]{\left\Vert #1 \right\Vert} %norm with variable height
\newcommand{\normb}[1]{\bigl\Vert #1 \bigr\Vert} %norm with big height
\newcommand{\snorm}[1]{\norm{#1}_\infty} %spectral norm  =  (2->2)-norm

\newcommand{\tnorm}[1]{\norm{#1}_{1}} %trace norm

\newcommand{\dnorm}[1]{\norm{#1}_\diamond} %diamond norm

\newcommand{\dnormb}[1]{\normb{#1}_\diamond}

\newcommand{\TwoNorm}[1]{\norm{#1}_{2}} % 2 norm

\newcommand{\ket}[1]{\left.\left|{#1}\right.\right\rangle}

\newcommand{\bra}[1]{\left.\left\langle{#1}\right.\right|}

\newcommand{\ketbra}[2]{\ket{#1} \!\! \bra{#2}}

\newcommand{\sandwich}[3]
  {\left\langle  #1 \right| #2 \left| #3 \right\rangle}

\newcommand{\oket}[1]{\left.\left|{#1}\right.\right)}
\newcommand{\obra}[1]{\left.\left({#1}\right.\right|}
\newcommand{\obraket}[2]{\left( #1 \middle| #2 \right)}
\newcommand{\oketbra}[2]{\oket{#1} \!\! \obra{#2}}
\newcommand{\osandwich}[3]
  {\left(  #1 \right| #2 \left| #3 \right)}

\newcommand{\Haar}{\mathrm{H}}

\newcommand{\R}{\mathbb{R}}
\newcommand{\N}{\mathbb{N}}
\newcommand{\C}{\mathbb{C}}
\newcommand{\Z}{\mathbb{Z}}
\newcommand{\Q}{\mathbb{Q}}
\newcommand{\F}{\mathbb{F}}

\newcommand{\ie}{i.\,e.}

\newcommand{\eg}{e.\,g.}
 
\newcommand{\one}{\mathbbm{1}}

\DeclareMathOperator{\spann}{span}
\newcommand{\haar}{\mu_\mathrm{H}}

\def\be                 {\begin{equation}}
\def\ee                 {\end{equation}}

\def \d            {\mathrm{d}}

\definecolor{jonas}{rgb}{0,.4,1}

\definecolor{ingo}{rgb}{1,.2,.4}

\definecolor{markus}{HTML}{006600}

\definecolor{felipe}{rgb}{0.2,.5,0.5}

\definecolor{jens}{rgb}{0.2,.7,0.9}

\newcommand{\newtext}[1]{{
#1}}

\begin{document}

\title{Efficient unitary designs with a system-size independent\\ number of non-Clifford gates}

\author{J. Haferkamp}
\email{jhaferkamp42@gmail.com}
\affiliation{%
  Dahlem Center for Complex Quantum Systems, Freie Universit{\"a}t Berlin, Germany
}%

\author{F. Montealegre-Mora}
\author{M. Heinrich}
\affiliation{%
  Institute for Theoretical Physics, University of Cologne, Germany
}%
\author{J. Eisert}
\affiliation{%
	Dahlem Center for Complex Quantum Systems, Freie Universit{\"a}t Berlin, Germany
}%
\author{D. Gross}
\affiliation{%
	Institute for Theoretical Physics, University of Cologne, Germany
}
\author{I. Roth}
\affiliation{%
  Dahlem Center for Complex Quantum Systems, Freie Universit{\"a}t Berlin, Germany
}%

\begin{abstract}\noindent
	Many quantum information protocols require the implementation of random unitaries. Because it takes exponential resources to produce Haar-random unitaries drawn from the full $n$-qubit group, one often resorts to  $t$-designs. 
	Unitary $t$-designs mimic the Haar-measure up to $t$-th moments. It is known that  Clifford operations can implement at most $3$-designs. In this work, we quantify the non-Clifford resources required to break this barrier.
	We find that it suffices to inject $O(t^{4}\log^{2}(t)\log(1/\varepsilon))$ many non-Clifford gates into a polynomial-depth random Clifford circuit to obtain an $\varepsilon$-approximate $t$-design.
	Strikingly, the number of non-Clifford gates required is independent of the system size -- asymptotically, the density of non-Clifford gates is allowed to tend to zero.
	We also derive novel bounds on the convergence time of random Clifford circuits to the $t$-th moment of the uniform distribution on the Clifford group. Our proofs exploit a recently developed variant of Schur-Weyl duality for the Clifford group, as well as bounds on restricted spectral gaps of averaging operators.
\end{abstract}

\maketitle

Random vectors and unitaries are ubiquitous in
protocols and arguments of quantum information and many-body physics.
In quantum information, a paradigmatic example is the \emph{randomized benchmarking protocol}
\cite{ZyczkowskiRB,MagGamEmer,KnillBenchmarking}, which aims to characterize the error rate of quantum gates.
There, random unitaries are used to average potentially complex errors into a single, easy to measure error rate. In many-body physics, random unitaries are used e.g.\ to model the dynamics that are thought to describe the mixing process that quantum information undergoes when absorbed into, and evaporated from, a black hole \cite{HaydenBlackHoles}.
In these and related cases, one is faced with the issue that unitaries drawn uniformly from the full many-body group are \emph{unphysical} in the sense that, with overwhelming probability, they cannot be implemented efficiently.
The notion of a \emph{unitary $t$-design} captures an efficiently realizable version of uniform randomness \cite{dankert_exact_2009,DankertThesis,GroAudEis}.
More specifically, a probability measure on the unitary group is a $t$-design if it matches the uniform Haar measure up to $t$-th moments.

Applications abound.
The randomness provided by designs is used to foil attackers in quantum cryptography protocols \cite{AmbainisEtAl:2009,
divinzeno_data_2001,
0810.2327}. 
It guards against worst case behavior in various quantum
\cite{sen_random_2006,
PhysRevA.72.032325, 
ScottDesigns,
PhysRevA.84.022327, 
AverageGateFidelities,
0810.2327, 
KueZhuGro16b}
and classical
\cite{GroKraKue15_partial}
estimation problems.
Designs allow for an efficient implementation of \emph{decoupling} procedures, a primitive in quantum Shannon theory \cite{szehr2013decoupling}.
In quantum complexity, unitary designs are used as models for generic instances of time evolution that display a quantum computational speed-up 
\cite{Generic,ComplexityGrowth}.
Unitary designs are now standard tools for the quantitative study of toy models in high energy physics, quantum gravity, and quantum thermodynamics \cite{HaydenBlackHoles,RobertsYoshida,PhysRevE.87.032137,onorati_mixing_2017}.

The multitude of applications motivates the search for efficient constructions of unitary $t$-designs~\cite{brandao_local_2016,
	brandao_efficient_2016,
	cleve2015near, harrow_random_2009,hunter2019unitary}.
In particular, Brandao, Harrow and Horodecki~\cite{brandao_local_2016} show that local random circuits on $n$ qubits with $O(n^2t^{10})$ many gates give rise to an approximate $t$-design.
In practice, it is often desirable to find more structured implementations.
Designs consisting of \emph{Clifford operations} would be particular attractive from various points of view:
(i)
Because the Clifford unitaries form a finite group, elements can be represented exactly using a small number ($O(n^2)$) of bits.
(ii)
The Gottesman-Knill Theorem ensures that there are efficient classical algorithms for simulating Clifford circuits.
(iii)
Most importantly, in \emph{fault-tolerant architectures} \cite{QEC2,EarlFaultTolerant}, Clifford unitaries tend to have comparatively simple realizations, while the robust implementation of general gates (e.g.\ via \emph{magic-state distillation}) carries a significant overhead.
The difference is so stark that in this context, Clifford operations are often considered to be a free resource, and the complexity of a circuit is measured solely in terms of the number of non-Clifford gates \cite{ResourceTheory,PhysRevLett.118.090501}.

The Clifford group is known to form a unitary
$t$-design for $t=2$ \cite{divinzeno_data_2001} %,dankert_exact_2009} 
and $t=3$
\cite{Webb3Design,PhysRevA.96.062336,Kueng3Design},
but fails to have this property for $t>3$
\cite{Webb3Design,PhysRevA.96.062336,Kueng3Design,zhu_clifford_2016,HelsenClifford}.
\newtext{
In fact, the Clifford group is singled out among the finite subgroups of the unitary group by being a 3-design \cite{bannai_unitary_2020}.
Moreover, Refs.~\cite{bannai_unitary_2020,sawicki_universal_2017} together imply that \emph{any} local gate set that generates an exact unitary design of order $t>3$ must necessarily be universal, c.f.~the discussion in Sec.~\ref{section:isomorphism}.
Hence, any efficient design construction for $t>3$ can only be approximate, and the Clifford group seems to be a distinguished starting point.
}

This leads us to the central question underlying this work:
\textit{How many non-Clifford gates are required to generate an approximate unitary $t$-design?}
A direct application of the random circuit model of Ref.~\cite{brandao_local_2016} yields an estimate of $O(n^2t^{10})$ non-Clifford operations. 
In this paper we show that a polynomial-sized random Clifford circuit, together with a \emph{system size-independent} number of $O(t^{4}\log^2(t))$ non-Clifford gates -- a ``homeopathic dose'' -- is already sufficient.

\begin{figure}[h]
	\center
	\includegraphics[scale=0.9]{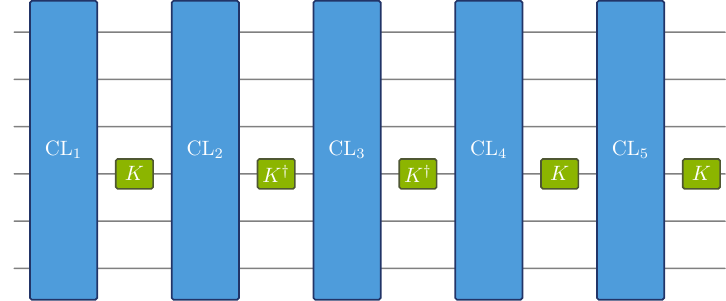}
	\caption{\emph{$K$-interleaved Clifford circuits}: We consider a model where random Clifford operations are alternated with a non-Clifford gate $K$ or its inverse $K^\dagger$.
		\label{figure}
	}
\end{figure}

We establish this main result for two different circuit models (Fig.~\ref{figure}).
In Section~\ref{sec:main result}, we consider alternating unitaries drawn uniformly from the Clifford group with a non-Clifford gate.
This gives rise to an efficient quantum circuit, as there are classical algorithms for sampling uniformly from the Clifford group, and for producing an efficient gate decomposition of the resulting operation \cite{koenig_how_2014}.
A somewhat simpler model is analyzed in Section~\ref{sec:local circuits}. 
There, we assume that the Clifford layers are circuits consisting of gates drawn form a local Clifford gate set.
These circuits will only approximate the uniform measure on the Clifford group.
Theorem~\ref{theorem:clifford-design}, which might be of independent interest, gives novel bounds on the convergence rate.

The key to this scaling lies in the structure of the commutant of the $t$-th tensor power of the Clifford group, described by a variant of Schur-Weyl duality developed in a sequence of recent works \cite{zhu_clifford_2016,nezami2016multipartite,gross2017schur,FelipeGross}.
There, it has been shown that the dimension of this commutant 
--
which measures the failure of the Clifford group to be a $t$-design from a representation theoretical perspective
--
is independent of the system size.
Refs.~\cite{zhu_clifford_2016,gross2017schur} have used this insight to provide a construction for exact \emph{spherical} $t$-designs that consist of a system size-independent number of Clifford orbits.
It has been left as an open problem whether these ideas can be generalized from spherical designs to the more complex notion of unitary designs, and whether the construction can be made efficient~\cite{gross2017schur}.
The present work resolves this question in the affirmative.

Finally, we note that in Ref.~\cite{zhou_entanglement_2019}, it has been observed numerically that adding a single $T$ gate to a random Clifford circuit has dramatic effects on the entanglement spectrum. 
A relation to $t$-designs was suspected.
Our result provides a rigorous understanding of this observation.

\section{Results}

\subsection{Approximate $t$-designs with few non-Clifford gates}
\label{sec:main result}

To state our results precisely, we need to formalize the relevant notion of approximation, as well as the circuit model used.
Let $\nu$ be a probability measure on the unitary group $U(d)$.
The measure $\nu$ gives rise to a quantum channel
\begin{equation}
  \MO_t(\nu)(\rho) :=\int_{\U(d)}U^{\otimes t}\rho \left(U^{\dagger}\right)^{\otimes t}\mathrm{d}\nu(U),
\end{equation}
which applies $U^{\otimes t}$, with $U$ chosen according to $\nu$.
We will refer to $\MO_t(\nu)$ as the \emph{$t$-th moment operator} associated with $\nu$.
Following Ref.~\cite{harrow_random_2009},
we quantify the degree to which a measure approximates a $t$-design by the diamond norm distance of its moment operator to the moment operator of the Haar measure $\mu_{\rm H}$ on $U(d)$.

\begin{definition}[Approximate unitary design]
	Let $\nu$ be a distribution on $\U(d)$.
	Then $\nu$ is an (additive) $\varepsilon$-approximate $t$-design if
	\begin{equation}
	\|\MO_t(\nu)-\MO_t(\haar)\|_{\diamond}\leq \varepsilon.
	\end{equation}
\end{definition}

Denote the uniform measure on the multiqubit Clifford group $\mathrm{Cl}(2^n)$ by $\mu_{\rm Cl}$, and let $K$ be some fixed single-qubit non-Clifford gate.
The circuit model we are considering 
(Figure~\ref{figure})
interleaves Clifford unitaries drawn from $\mu_{\rm Cl}$, with random gates from
$\{K, K^\dagger,\one\}$ acting on an arbitrary qubit\footnote{
  We use the set $\{K, K^\dagger, \one\}$ instead of just $\{K\}$ for technical reasons:
  Making the set closed under the adjoint causes the moment operator to be Hermitian.
  The identity is included to ensure that the concatenation of two random elements has a non-vanishing probability of producing a non-Clifford gate---a property that will slightly simplify the proof.
  Of course, in a physical realization, identity gates and the following Clifford operation are redundant and need not be implemented. 
}.
Note that the concatenation of two unitaries drawn from measures $\nu_1$ and $\nu_2$ is described by the convolution $\nu_1*\nu_2$ of the respective measures.
We thus arrive at this formal definition of the circuit model:
\begin{definition}[$K$-interleaved Clifford circuits]\label{def:KinterleavedClifford}
	Let $K\in U(2)$.
	Consider the probability measure $\xi_K$ that draws uniformly from the set $\{K\otimes \one_{2^{n-1}}, K^{\dagger}\otimes \one_{2^{n-1}}, \one_{2^n}\}$.
	A \emph{$K$-interleaved Clifford circuit of depth $k$} is the random circuit acting on $n$ qubits described by the probability distribution
	\begin{equation}
\sigma_{k}:=	\underbrace{\mu_{\rm Cl}*\xi_K*\dots *\mu_{\rm Cl}*\xi_K}_{k\; \text{times}}.
	\end{equation}
\end{definition}
For convenience, we work with the logarithm of base $2$: $\log(x):=\log_2(x)$.
We are now equipped to state the main result of this work in the form of a theorem:

\begin{restatable}[Unitary designs with few non-Clifford gates]{theorem}{Tcountdesign}\label{thm:Tcountdesign}
	Let $K \in U(2)$ be a non-Clifford unitary.
	There are constants $C_1(K), C_2(K)$ 
	such that 
	for any $k\geq C_1(K)\log^2(t)(t^{4}+t\log(1/\varepsilon))$,
	a $K$-interleaved Clifford circuit with depth $k$ acting on $n$ qubits is an additive $\varepsilon$-approximate $t$-design for all $n\geq C_2(K)t^2$.
\end{restatable}

We give the proofs of this theorem in Section~\ref{section:t-gates}.
In Theorem~\ref{thm:Tcountdesign}, we consider uniformly drawn multiqubit Clifford unitaries.
This can be achieved with $O(n^3)$ classical random bits~\cite{koenig_how_2014} and then implemented with $O(n^2/\log(n))$ gates~\cite{aaronson_stabilizer_2004}.
Combined with these results, Theorem~\ref{thm:Tcountdesign} implies an overall gate count of $O(n^2/\log(n)t^4\log^2(t))$ improving the scaling compared to Ref.~\cite{brandao_local_2016} in the dependence on both $t$ and $n$.
In this sense, our construction can be seen as a classical-quantum hybrid construction of unitary designs: The scaling is significantly improved by outsourcing as many tasks as possible to a classical computer.
A construction in which all parts of the random unitary are local random circuits is considered in Corollary~\ref{cor:combined}.

For designs generated from general random local circuits, numerical results suggest that convergence is much faster in practice than indicated by the proven bounds~\cite{cwiklinski_local_2013}.
We expect that a similar effect occurs here, and that in fact very shallow $K$-interleaved Clifford circuits are sufficient to approximate $t$-designs.
This intuition is supported by the numerical results of Ref.~\cite{zhou_entanglement_2019}, which show that even a single $T$-gate has dramatic effects on the entanglement spectrum of a quantum circuit.

It is moreover noteworthy that circuits with few $T$-gates \newtext{can be efficiently simulated 
\cite{bravyi_stabilizer_2019,
PhysRevLett.115.070501, 
heinrich2019robustness,PhysRevLett.116.250501,
seddon2020quantifying}.}
The scaling of these algorithms is polynomial in the depth of the circuit, but exponential in the number of $T$-gates.
Combined with our result, this implies that for fixed additive errors $\varepsilon$, there are families of $\varepsilon$-approximate unitary $O(\log(n))$-designs simulable in
quasi-polynomial time.
For the general random quantum circuit model, it is conjectured that a depth of order $O(nt)$ suffices to approximate $t$-designs~\cite{brandao_local_2016,brandao_complexity_2019}.
If such a linear scaling is sufficient in our model, the quasi-polynomial time estimate for classical simulations would improve to polynomial.

For the proof of Theorem~\ref{thm:Tcountdesign} we need to analyse the connection between the $t$-th moment operator of the Haar measure and the commutant of the diagonal action of the Clifford group.
The latter was proven to be spanned by representations of so-called \emph{stochastic Lagrangian sub-spaces} in Ref.~\cite{gross2017schur}.
In particular, we prove almost tight bounds on the overlap of the Haar operator with these basis vectors in Lemma~\ref{lemma:haarsymmetrization} that might be of independent interest.
This will allow us to invoke a powerful theorem by Varj{\'u}~\cite{varju_walks_2013} on restricted spectral gaps of probability distributions on compact Lie groups to show that non-Clifford unitaries have a strong impact on representations of Lagrangian sub-spaces that are not also permutations.
We combine this insight with a careful combinatorial argument about the Gram-Schmidt orthogonalization of the basis corresponding to stochastic Lagrangian sub-spaces to bound the difference to a unitary $t$-design in diamond norm.

Moreover, the bound for Theorem~\ref{thm:Tcountdesign} allows us to prove a corollary about the stronger notion of \textit{relative approximate designs}:
\begin{definition}[Relative 
$\varepsilon$-approximate $t$-design]
	We call a probability $\nu$ a relative $\varepsilon$-approximate $t$-design if 
	\begin{equation}
	(1-\varepsilon)\MO_t(\nu)\preccurlyeq\MO_t(\mu_{\rm H})\preccurlyeq (1+\varepsilon)\MO_t(\nu),
	\end{equation}
	where $A\preccurlyeq B$ if and only if $B-A$ is completely positive.
\end{definition}

\begin{corollary}[$K$-interleaved Clifford circuits as relative $\varepsilon$-approximate $t$-designs]\label{cor:dropindependence}
	There are constants $C_1'(K), C_2'(K)$ such that a $K$-interleaved Clifford circuit is a relative $\varepsilon$-approximate $t$-design in depth $k\geq C_1'(K)\log^{2}(t)(2nt+\log(1/\varepsilon))$ for all $n\geq C_2'(K)t^2$.
\end{corollary}
Hence, if we drop the system-size independence, we can achieve a scaling of $O(nt)$ at least until $t~\sim \sqrt{n}$.

While we believe the setting of $K$-interleaved Clifford circuits to be the more relevant case, the same method of proof works for \textit{Haar}-interleaved Clifford circuits.
Here, we draw not from the gate set $\{K_i,K^{\dagger}_i,\one\}$, but instead Haar-randomly from $U(2)$.
The advantage is that we obtain explicit constants for the depth, while the depth in the $K$-interleaved setting has to depend on a constant (as $K$ might be arbitrarily close to the identity).
\begin{proposition}[Haar-interleaved Clifford circuits as additive $\varepsilon$-approximate $t$-designs]\label{prop:Haarinverleavedversion}
	For $k\geq 36(33t^4+3t\log(1/\varepsilon))$, Haar-interleaved Clifford circuits with depth $k$ form an additive $\varepsilon$-approximate $t$-design for all $n\geq 32t^2+7$.
\end{proposition}
Similarly, variants of Corollary~\ref{cor:dropindependence} for Haar-interleaved Clifford circuits can be obtained, here also without the $\log^{2}(t)$ dependence.
\newtext{Finally, we discuss an application to higher R\~{e}nyi entropies in Appendix~\ref{section:renyi}.}

\subsection{Local random Clifford circuits for Clifford and unitary designs}
\label{sec:local circuits}

The circuits considered in the previous section require one to find the gate decomposition of a random Clifford operation.
In this section, we analyze the case where the Clifford layers are circuits consisting of gates drawn from a local set of generators.

As a first step, we establish that a $2$-local random Clifford circuit on $n$ qubits of depth $O(n^2t^{9}\log^{-2}(t)\log(1/\varepsilon))$ constitutes a relative $\varepsilon$-approximate Clifford $t$-design, i.e., reproduces the moment operator of the Clifford group up to the $t$-th order with a relative error of $\varepsilon$.
We consider local random Clifford circuits that consist of $2$-local quantum gates from a finite set $G$ with is closed under taking the inverse and generates $\Cl(4)$. We refer to such a set as a \emph{closed, generating set}.
A canonical example for such a closed, generating set is $\{H \otimes \one, S \otimes \one, S^3 \otimes \one ,\mathrm{CX} \}$ where $H$ is the Hadamard gate, $S$ is the phase gate and $\mathrm{CX}$ is the cNOT-gate \cite{NielsenChuang}.
Such a set $G$ induces a set of multi-qubit Clifford unitaries $\hat G \subset \Cl(n)$ by acting on any pair of adjacent qubits \newtext{on a line}, where we adopt periodic boundary conditions. We then define the corresponding random Clifford circuits.
\begin{definition}[Local random Clifford circuit]
	\label{def:random-cliff-circuit}
	Let $G \subset \Cl(4)$ be a closed, generating set \newtext{containing the identity}.
	Define the probability measure $\sigma_G$ as the measure having uniform support on $\hat G \subset \Cl(n)$ acting on $n$ qubits.
	A \emph{local random Clifford circuit} of depth $m$ is the random circuits described by the probability measure  $\sigma_G^{\ast m}$.
\end{definition}

\newtext{
For technical reasons, we again assume that the identity is part of the generating set.
This assumption can be avoided but simplifies the argumentation in the following.
As for the Definition~\ref{def:KinterleavedClifford} of $K$-interleaved Clifford circuits before, any upper bound on the depth of local random Clifford circuits with identity is a bound for those without. 
}

Our result on local random Clifford circuits even holds for a stronger notion for approximations of designs, namely relative approximate designs.
Write $A\preccurlyeq B$ if  $B-A$ is positive semi-definite.

\begin{definition}[Relative approximate Clifford $t$-designs]
	Let $\nu$ be a probability measure on $\mathrm{Cl}(2^n)$.
	 Then, $\nu$ is a \emph{relative $\varepsilon$-approximate Clifford $t$-design} if
	\begin{equation}\label{eq:relativeclifforddesign}
	(1-\varepsilon)\MO_t(\mu_{\rm Cl})\preccurlyeq \MO_t(\nu) \preccurlyeq (1+\varepsilon)\MO_t(\mu_{\rm Cl}).
	\end{equation}
\end{definition}

With this definition, our result reads as follows.

\begin{restatable}[Local random Clifford designs]{theorem}{clifforddesign}
\label{theorem:clifford-design}
Let $n\geq 12t$, then a local random Clifford circuit of depth $O(n\log^{-2}(t)t^{8}(2nt+\log(1/\varepsilon)))$ constitutes a relative $\varepsilon$-approximate Clifford $t$-design.
\end{restatable}
The proof of the theorem is given in Section~\ref{section:clifford-design}.
This result is a significant improvement over the scaling of $O(n^8)$, which is implicit in Ref.~\cite{divinzeno_data_2001}.

We can combine this result with the bounds obtained in Section~\ref{section:t-gates}.
To this end, consider a random circuit that $k$-times alternatingly applies a local random Clifford circuit of depth $m$, and a unitary drawn from the probability measure $\xi_K$.
The corresponding probability measure is 
\begin{equation}\label{eq:combined_model}
\sigma_{k,m}:= \underbrace{\sigma_G^{*m}*\xi_K*\dots *\sigma_G^{*m}*\xi_K}_{k\;\text{times}}.
\end{equation}
For these local random circuits we establish the following result:
\begin{restatable}[Local random unitary design]{corollary}{combined}
	\label{cor:combined} Let $K \in U(2)$ be a non-Clifford gate and let $G \subset \Cl(4)$ be a closed, generating set.
	There are constants $C''_1(K,G), C''_2(K), C''_3(K)$ such that whenever
	\begin{equation*}
	m \geq C''_1(K,G)n \log^{-2}(t) t^8\left( 2nt + \log (1/\varepsilon) \right)
	\text{\quad and\quad}
	k\geq C''_2(K)\log^2(t)(t^{4}+t\log(1/\varepsilon)),
	\end{equation*}
	the local random circuit $\sigma_{k,m}$, defined in \eqref{eq:combined_model}, is an $\varepsilon$-approximate unitary $t$-design for all $n\geq C''_3(K)t^2$.
\end{restatable}
The complete argument for the corollary is given at the end of  Section~\ref{section:clifford-design}.
After introducing technical preliminaries in  Section~\ref{section:preliminaries}, the remainder of the paper, Section~\ref{section:t-gates} and Section~\ref{section:clifford-design}, is devoted to the proofs of Theorem~\ref{thm:Tcountdesign}, Theorem~\ref{theorem:clifford-design} and the Corollary~\ref{cor:combined}.
\newtext{
Finally, in Section~\ref{section:isomorphism} we elaborate on and formalize as
Proposition~\ref{theorem:singlingoutclifford} the observation that there exists no non-universal gate set generating exact $4$-designs for arbitrary system size.
This observation is an immediate consequence of the classification of finite unitary $t$-groups and a criterion for the universality of finite gate sets \cite{guralnick_larsen_2005,bannai_unitary_2020,sawicki_universal_2017}.
}

\section{Technical preliminaries}
\label{section:preliminaries}

\subsection{Operators and superoperators}
\label{section:norms}

Given a (finite-dimensional) Hilbert space $\mathcal{H}$, we denote with $L(\mathcal{H})$ the space of linear operators on $\mathcal H$ with involution $\dagger$ mapping an operator to its adjoint with respect to the inner product on $\mathcal H$. $L(\mathcal{H})$ naturally inherits a Hermitian inner product, the \emph{Hilbert-Schmidt inner product}
\begin{equation}
 \obraket{A}{B} := \tr(A^\dagger B), \qquad \forall A,B\in L(\mathcal H).
\end{equation}
As this definition already suggests, we will use ``operator kets and bras'' whenever we think it simplifies the notation. Concretely, we write $\oket{B}=B$ and denote with $\obra{A}$ the linear form on $L(\mathcal H)$ given by 
\begin{equation}
 \obra{A}:\; B \longmapsto \obraket{A}{B}.
\end{equation}
Following common terminology in quantum information theory, we call linear maps $\phi:\,L(\mathcal H)\rightarrow L(\mathcal H)$ on operators ``superoperators''. 
We use $\phi^\dagger$ to denote the adjoint map with respect to the Hilbert-Schmidt inner product. 
Note that with the above notation, $\phi = \oketbra{A}{B}$ defines a rank one superoperator with $\phi^\dagger = \oketbra{B}{A}$.
Moreover, we will denote by the superoperator $\Ad_A := A\cdot A^{-1}$ the \emph{adjoint action} of an invertible operator $A\in \GL(\mathcal H)$ on $L(\mathcal H)$. For notational reasons, we sometimes write $\Ad(A)$ instead of $\Ad_A$.

We consistently reserve the notation $\norm{\cdot}_p$ for the Schatten $p$-norms
\begin{equation}
 \norm{A}_p := \tr(|A|^p)^{1/p} = \norm{\sigma(A)}_{\ell_p},
\end{equation}
where $\sigma(A)$ is the vector of singular values of $A$. In particular, we use the \emph{trace norm} $p=1$, the \emph{Frobenius} or \emph{Hilbert-Schmidt norm} $p=2$ and the \emph{spectral norm} $p=\infty$. Clearly, this norms can be defined for both operators and superoperators and we will use the same symbol in both cases. For the latter, however, there is also a family of induced operator norms
\begin{equation}
 \norm{\phi}_{p\rightarrow q} := \sup_{\norm{X}_p \leq 1} \norm{\phi(X)}_q.
\end{equation}
Note that $\norm{\cdot}_{2\rightarrow 2} \equiv \norm{\cdot}_{\infty}$. 
Finally, we are interested in ``stabilized'' versions of these induced norms, in particular the \emph{diamond norm} 
\begin{myalign}
 \dnorm{\phi} &:=\sup_{d\in\N} \norm{\phi\otimes\id_{L(\C^d)}}_{1\rightarrow 1} = \norm{\phi\otimes\id_{L(\mathcal H)}}_{1\rightarrow 1}.%, \\
\end{myalign}
The following norm inequality will be useful \cite{low_pseudo-randomness_2010}
\begin{equation}
\label{eq:dnorm-bound}
 \dnorm{\phi} \leq (\dim \mathcal H)^2 \snorm{\phi}, \qquad \snorm{\phi} \leq \sqrt{\dim \mathcal H} \dnorm{\phi}.
\end{equation}

\subsection{Commutant of the diagonal representation of the Clifford group}
\label{section:representation}

In this section, we review some of the machinery developed in Ref.~\cite{gross2017schur}.
Recall that the $n$-qubit \emph{Clifford group} $\Cl(n)$ is defined as the unitary normalizer of the Pauli group $\mathcal{P}_n$
\newtext{as}
\begin{equation}
\label{eq:clifford-group}
 \Cl(n) = \left\{ U \in U(2^n, \Q[i]) \; \big| \; U\mathcal{P}_n U^\dagger \subset \mathcal{P}_n \right\}.
\end{equation}
Here, we followed the convention to restrict the matrix entries to rational complex numbers.
 This avoids the unnecessary complications from an infinite center $U(1)$ yielding a finite group with minimal center $Z(\Cl(n))=Z(\mathcal{P}_n)\simeq \Z_4$.
The Clifford group can equivalently be defined in a less conceptual but more constructive manner:
It is the subgroup of $\U(2^n)$ generated by $\mathrm{CX}$, the controlled not gate, the Hadamard gate $H$ and the phase gate $S$.

For this work, the $t$-th diagonal representation of the Clifford group, defined as
\begin{equation}
 \tau^{(t)}:\; \Cl(n) \longrightarrow \U(2^{nt}), \quad U \longmapsto U^{\otimes t},
\end{equation}
will be of major importance. It acts naturally on the Hilbert space $((\C^2)^{\otimes n})^{\otimes t}$ which can be seen as $t$ copies of an $n$-qubit system. However, it will turn out that the operators commuting with this representation naturally factorize with respect to a different tensor structure on this Hilbert space, namely $((\C^2)^{\otimes t})^{\otimes n}\simeq ((\C^2)^{\otimes n})^{\otimes t}$. Because of the different exponents, it should be clear from the context which tensor structure is meant.
We will make ubiquitous use of the description of the commutant of the diagonal representation in terms of \emph{stochastic Lagrangian sub-spaces} \cite{gross2017schur}:
\begin{definition}[Stochastic Lagrangian sub-spaces]\label{def:stochastic lagrangians}
	Consider the quadratic form $\mathfrak q:\mathbb{Z}^{2t}_2\to \mathbb{Z}_{4}$ defined as $\mathfrak q(x,y):= x\cdot x-y\cdot y \mod{4}$.
	The set $\Sigma_{t,t}$ denotes the set of all sub-spaces $T\subseteq \mathbb{Z}^{2t}_2$ being subject to the following properties:
	\begin{enumerate}
	  \item T is totally $\mathfrak q$-isotropic: $x\cdot x=y\cdot y \mod{4}$ for all $(x,y)\in T$. \label{def:q isotropicity}
		\item T has dimension $t$ (the maximum dimension compatible with total isotropicity).
		\item T is \emph{stochastic}: $(1,\dots, 1)\in T$.
	\end{enumerate}
\end{definition}
We call elements in $\Sigma_{t,t}$ \textit{stochastic Lagrangian sub-spaces}.
We have
\begin{equation}\label{eq:boundnumberlagrangian}
|\Sigma_{t,t}|=\prod_{k=0}^{t-2}(2^k+1)\leq 2^{\frac12 (t^2+5t)}.
\end{equation}
With this notion, we can now state the following key theorem from Ref.~\cite{gross2017schur}.
\begin{theorem}[\cite{gross2017schur}]
\label{theorem:commutant}
	If $n\geq t-1$, then the commutant $\tau^{(t)}(\Cl(n))'$ of the $t$-th diagonal representation of the Clifford group is spanned by the linearly independent operators $r(T)^{\otimes n}$, where $T\in \Sigma_{t,t}$ and
	\begin{equation}
	r(T):= \sum_{(x,y)\in T} |x\rangle\langle y|.
	\end{equation}
\end{theorem}
Since the representation in question is fixed throughout this paper, we will simplify the notation from now on and write $\Cl(n)'\equiv \tau^{(t)}(\Cl(n))'$.
To make use of a more sophisticated characterization of the elements $r(T)$ developed in Ref.~\cite[Section~4]{gross2017schur}, we need the following definitions.

\newtext{
\begin{definition}[Stochastic orthogonal group]
Consider the quadratic form $q:\mathbb{Z}^t_2\to \mathbb{Z}_4$ defined as $q(x):=x\cdot x \mod{4}$.
The \emph{stochastic orthogonal group} $O_t$ is defined as the group of  $t\times t$ matrices $O$ with entries in $\mathbb{Z}_2$ such that $q(Ox)=q(x)$ for all $x\in \mathbb{Z}^t_2$.
% \begin{enumerate}
% 	\item $q(Ox)=q(x)$ for all $x\in \mathbb{Z}^t_2$ and
% 	\item $O(1,\dots, 1)^{T}=(1,\dots, 1)^T~~~\mathrm{mod}~2$.
% \end{enumerate}
\end{definition}
The subspace $T_O:=\{(Ox,x),x\in \mathbb{Z}_2^t\}$ is a stochastic Lagrangian subspace.
Moreover, the operator $r(O):=r(T_O)$ is unitary.
We will therefore canonically embed the orthogonal stochastic group $O_t\subset\Sigma_{t,t}$.
Notice that the permutation group on $t$ objects, referred to as $S_t$, may be embedded into $O_t$ by acting on the standard basis of $\mathbb{Z}_2^t$.
Together with $O_t$, the following definition can be used to fully characterize the set of 
stochastic Langrangian sub-spaces, $\Sigma_{t,t}$.

% With this notion and the next one, all stochastic Lagrangian sub-spaces can be characterized.

\begin{definition}[Defect sub-spaces]
A defect subspace is a subspace $N\subseteq\mathbb{Z}_2^t$ which is isotropic with respect to $q$, that is, that $q(x)=0$ for all $x\in N$.
\end{definition}

The quadratic form $q$ is what is known as a \emph{generalized quadratic refinement} of the bi-linear form defined by the inner product $(x,y)\mapsto x\cdot y \mod 2$ (see, e.g., Ref.~\cite[App.~A]{klausthesis} for a self-contained discussion).
In the following, the ortho-complement $N^\perp$ of a subspace $N\subseteq\mathbb{Z}_2^t$ is taken with respect to the inner product modulo 2,
\begin{align*}
  N^\perp = \{v\in\mathbb{Z}_2^t \;|\; v\cdot u = 0\mod 2,\ \forall\ u\in N\}.
\end{align*}

Notice that $q(x)=0$ implies that $x\cdot \mathbf{1}_t=0\mod 2$, where $\mathbf{1}_t:=(1,\dots ,1)^T$ is the all-ones vector.
Thus, we do not need a separate clause requiring $\mathbf{1}_t\in N^\perp$ in the definition of defect sub-spaces (compare Ref.~\cite[Def.~4.16]{gross2017schur}).
Moreover, one may verify that $2q(x)=2x\cdot \mathbf{1}_t\mod 4$.
This implies, similarly, that if $O$ preserves $q$, then $O\mathbf{1}_t=\mathbf{1}_t$.
Borrowing the language of~\cite{gross2017schur}, all $q$-isometries are stochastic (compare the definition of the orthogonal stochastic group in  that reference,~\cite[Def.~4.11]{gross2017schur}). 
The reason for these simplifications is that here we focus on the qubit case exclusively, while Ref.~\cite{gross2017schur} works simultaneously for qubits and odd qudits. We use the names \emph{stochastic orthogonal group} and \emph{defect subspace} (rather than simply \emph{$q$-isometry group} and \emph{isotropic subspace}) to keep with the notation of that reference.

For any defect subspace $N$, it holds that $N\subseteq N^{\perp}$ (and thus $\dim N\leq t/2$).
Because of this, defect sub-spaces $N\subseteq\mathbb{Z}_2^t$ } %ending \newtext 
define \emph{Calderbank-Shor-Sloane (CSS)} codes
\begin{equation}
\mathrm{CSS}(N):=\left\{Z(p)X(q) \; | \; q,p\in N\right\},
\end{equation}
where the action of the multi-qubit Pauli operators is $Z(p)\ket{x}:=(-1)^{p\cdot x}\ket{x}$ and $X(q)\ket{x}:=\ket{x+q}$ for $x\in\Z_2^t$.
The corresponding projector is given by
\begin{equation}
P_N:=P_{\mathrm{CSS}(N)}=\frac{1}{|N|^2}\sum_{q,p\in N}Z(p)X(q).
\end{equation}
Since the order of the stabilizer group is $2^{2\dim N}$, $P_N$ projects onto a $2^{t-2\dim N}$-dimensional subspace of $(\C^2)^{\otimes t}$.
For $N=\{0\}$ we set $P_{\mathrm{CSS}(N)}:=\one$.
We summarize the findings of \newtext{Ref.}~\cite[Section~4]{gross2017schur} in Thm.~\ref{theorem:r decomposition}.
\newtext{We give a short proof to give an explicit relation between this theorem and the results of that work.
\begin{theorem}[\cite{gross2017schur}]
  \label{theorem:r decomposition}
 Consider $T\in \Sigma_{t,t}$, then
 \begin{equation}
 r(T)=2^{\dim N}r(O)P_{\mathrm{CSS}(N)}=2^{\dim N'}P_{\mathrm{CSS}(N')}r(O')
 \end{equation}
 for $O,O'\in O_t$ and $N,N'$ are unique defect sub-spaces with $\dim N=\dim N'$.
\end{theorem}
\begin{proof}%[(Proof sketch)]
  Recall from Ref.~\cite{gross2017schur} that the code space $\mathrm{range}\ P_{\mathrm{CSS}(N)}$ has an orthonormal basis of coset state vectors given by
  \begin{align*}
    \left\{ 
    \ket{N,[x]} := \frac{1}{\sqrt N}\sum_{y\in N} \ket{x+y}
    \;\Big|\; x\in N^\perp,\ [x]\in N^\perp/N
    \right\}.
  \end{align*} 
  One may compute that $r(O)\ket{N, [x]} = \ket{ON, [Ox]}$.
  This way, 
  \begin{align*}
    r(O)P_{\mathrm{CSS}(N)} = \sum_{[x]\in N^\perp/N} \ket{ON,[Ox]}\bra{N,[x]}.
  \end{align*} 
  Comparing this equation to~\cite[Lem.~4.23]{gross2017schur} we see that the set $\{2^{\dim N}r(O)P_{\mathrm{CSS}(N)}\}_O$ is equal to the set of $r(T)$ operators with right defect subspace given by $N$, i.e., with $T_{RD}=N$ in the notation of that reference.
  This way, varying over $N$ we obtain the full set $\Sigma_{t,t}$.
  The existence of a decomposition $2^{\dim N}P_{\mathrm{CSS}(N')}r(O')$ follows from the above by noting that $r(O)P_{\mathrm{CSS}(N)}r(O)^\dagger = P_{\mathrm{CSS}(ON)}$.
\end{proof}
}%ending \newtext

\begin{lemma}[Norms of $r(T)$]
\label{lemma:rT-norms}
 Suppose $r(T)=2^{\dim N}r(O)P_{N}$ as in Theorem~\ref{theorem:r decomposition}. Then it holds:
 \begin{align}
  \tnorm{r(T)} &= 2^{t-\dim N}, & \TwoNorm{r(T)} &= 2^{t/2}, & \snorm{r(T)} &= 2^{\dim N}.
 \end{align}
\end{lemma}
\begin{proof}
 Since any Schatten $p$-norm is unitarily invariant, we have $\norm{r(T)}_p = 2^{\dim N} \norm{P_N}_p$. The statements follow from $\rank P_N = 2^{t-2\dim N}$.
\end{proof}

In the following, we will often work with a normalized version of the $r(T)$ operators which we define as
\begin{equation}
\label{eqn:psiT}
Q_T := \frac{r(T)}{\TwoNorm{r(T)}} =  2^{-t/2} r(T).
\end{equation}

\section{
	Approximate unitary $t$-designs
}\label{section:t-gates}
In this section, we give a bound on the number of non-Clifford gates needed to leverage the Clifford group to an approximate unitary $t$-design. 
This is made precise by the following two theorems which rely on two distinct proof strategies and come with different trade-offs.

\Tcountdesign*

Recall from Def.~\ref{def:KinterleavedClifford} that a $K$-interleaved Clifford circuit has an associated probability measure $\sigma_K:=(\mu_{\Cl} * \xi_K)^{*k}$ where $\xi_K$ is the measure which draws uniformly from $\{K,K^\dagger,\one\}$ on the first qubit. Let us introduce the notation
\begin{equation}
\mathrm{R}(K):=
\int_{U(2^n)} \Ad_U^{\otimes t}\mathrm{d}\xi_{k}(U)=
\frac13\left(\Ad_K^{\otimes t}+\Ad_{K^{\dagger}}^{\otimes t}+\id\right)\otimes \id_{n-1}.
\end{equation}
Then, our goal is to bound the deviation of the moment operator
\begin{equation}
\MO_t(\sigma_{k}) = \int_{U(2^n)} \Ad_U^{\otimes t} \mathrm{d}\sigma_{k}(U)=\underbrace{ \MO_t(\mu_{\rm Cl})\mathrm{R}(K)\dots \MO_t(\mu_{\rm Cl})\mathrm{R}(K)}_{\text{$k$ times}},
\end{equation}
from the Haar projector $P_\Haar\equiv \MO_t(\haar)$ in diamond norm.
Using that $P_\Haar$ is invariant under left and right multiplication with unitaries, we have the identity
\begin{equation}
\label{eq:Aidentity}
 A^k - P_\Haar = ( A - P_\Haar )^k,
\end{equation}
for any mixed unitary channel $A$.
Thus, we can rewrite the difference of moment operators as
\begin{equation}
\label{eq:differencehaar}
\MO_t(\sigma_{k})-P_\Haar = [P_{\Cl}\mathrm{R}(K)]^k-P_{\Haar} = \left[ \left(P_{\Cl}-P_{\Haar}\right)\mathrm{R}(K)\right]^k,
\end{equation}
where we have introduced the shorthand notation $P_{\Cl}:= \MO_t(\mu_{\Cl})$.

\begin{remark}[Non-vanishing probability of applying the identity]
	We apply $K$, $K^{\dagger}$ with equal probability in Theorem~\ref{thm:Tcountdesign} such that $R(K)$ is Hermitian.
	The non-vanishing probability of applying $\one$, i.e., of doing nothing, is necessary in the proof of Lemma~\ref{lemma:overlapnonclifford}, because we require the probability distribution $\xi_K*\xi_K$ to have non-vanishing support on a non-Clifford gate.
	If $\xi_K$ is the uniform measure on $K$ and $K^{\dagger}$, then $\xi_K*\xi_K$ has support on $K^2$, $(K^{\dagger})^{2}$ and $\one$.
	We can hence drop this assumption for gates that do not square to a Clifford gate.
	This is not the case for e.g.~the $T$-gate.
\end{remark}

Our proof strategy for Theorem \ref{thm:Tcountdesign} makes use of the following two lemmas which are proven in Section~\ref{sec:overlaps} and \ref{section:haarsymmetrization}. The first lemma is key to the derivations in this section.
It is based on a bound (Lemma~\ref{lemma:haarsymmetrization}) on the overlap of stochastic Lagrangian sub-spaces with the Haar projector and Theorem~\ref{theorem:varju}, a special case of a
theorem about restricted spectral gaps of random walks on compact Lie groups due to Varj{\'u}~\cite{varju_walks_2013}.
\begin{restatable}[Overlap bound]{lemma}{overlapnonclifford}
  \label{lemma:overlapnonclifford}
	Let $K$ be a single qubit gate which is not contained in the Clifford group. Then, there is a constant $c(K)>0$ such that
	\begin{equation}
	\eta_{K,t}:=
	\max_{\substack{T\in\Sigma_{t,t}-S_t\\T'\in \Sigma_{t,t}}}
	\frac13 \left| \osandwich{Q_T}{\Ad_K^{\otimes t}+\Ad_{K^{\dagger}}^{\otimes t}+\id}{Q_{T'}}\right|\leq 1-c(K)\log^{-2}(t).
	\end{equation}
\end{restatable}

The second lemma is of a more technical nature.

\begin{restatable}[Diamond norm bound]
	{lemma}
	{diamondbound}
	\label{lemma:diamondbounds}
	Consider $T_1,T_2\in \Sigma_{t,t}$ and denote with $N_1,N_2$ their respective defect spaces. Then, it holds that
	\begin{align}
	\dnorm{\oketbra{Q_{T_1}}{Q_{T_2}}} &\leq 2^{\dim N_2-\dim N_1},\label{eq:diamondbound}\\
	|\obraket{Q_{T_1}}{Q_{T_2}}| &\leq 2^{-|\dim N_1-\dim N_2|}\label{eq:overlapsbound}.
	\end{align}
\end{restatable}

The difficulty of using these results to bound the difference
\begin{equation}
  \label{eqn:the difference}
  \MO_t(\sigma_{k})-P_\Haar = \big[ \left(P_{\Cl}-P_{\Haar}\right)\mathrm{R}(K)\big]^k,
\end{equation}
stems from the following reason: The range of the projector $P_{\Cl}-P_{\mathrm{H}}$ is the ortho-complement of the space spanned by permutations $Q_{\pi}^{\otimes n}$ for $\pi\in S_t$ within the commutant of the Clifford group spanned by the operators $Q_T^{\otimes n}$.  Although this is a conveniently factorizing and well-studied basis, it is \emph{non-orthogonal}. Thus, the projectors do not possess a natural expansion in this basis and we can not directly use the above bounds. 
However, we can write it explicitly in a suitable orthonormal basis of the commutant obtained by the Gram-Schmidt procedure from the basis $\{Q_T^{\otimes n} \, | \,T \in\Sigma_{t,t}\}$.
We summarize the properties of this basis in the following lemma.

\begin{restatable}[Properties of the constructed basis]{lemma}{gramschmidt}
	\label{lemma:gram-schmidt}
	Let $\{T_j\}_{j=1}^{|\Sigma_{t,t}|}$ be an enumeration of the elements of $\Sigma_{t,t}$ such that the first $t!$ spaces $T_j$ correspond to the elements of $S_t$.
	Then, the $\{E_j\}$ constitutes an orthogonal (but not normalized) basis, where
	\begin{equation}\label{eq:gramschmidt}
	E_j :=\sum_{i=1}^jA_{i,j}\,Q_{T_i}^{\otimes n}:=
	\sum_{i=1}^j\left[\sum_{\substack{\Pi\in S_j\\\Pi(j)=i}}\sign(\Pi)\prod_{l=1}^{j-1} \obraket{Q_{T_l}}{Q_{T_{\Pi(l)}}}^n\right]\,Q_{T_i}^{\otimes n}.
	\end{equation}
	Denote by $N_i$ the defect space of $T_i$. For $n\geq \frac12(t^2+5t)$, we have
	\begin{align}
	|A_{i,j}|&\leq 2^{t^3+4t^2+6t-n|\dim N_i-\dim N_j|}, \qquad \forall i,j, \label{eq:gramschmidtfirst}\\
	|A_{i,j}|&\leq 2^{2t^2+10t-n}, \qquad \forall i\neq j. \label{eq:gramschmidtsecond}
	\end{align}
	Moreover, it holds that
	\begin{equation}\label{eq:gramschmidtthird}
	1-2^{t^2+7t-n}\leq A_{j,j}\leq 1+2^{t^2+7t-n}.
	\end{equation}
\end{restatable}

We believe that the explicit bounds in Lemma~\ref{lemma:gram-schmidt} might be of independent interest in applications of the Schur-Weyl duality of the Clifford group.
For the sake of readibility, and as Theorem~\ref{thm:Tcountdesign} holds up to an inexplicit constant, we will bound all polynomials in $t$ by their leading order term in the following.
\newtext{Specifically, the bounds in Lemma~\ref{lemma:gram-schmidt} will be simplified by using
the inequalities
\begin{align}
  t^3+4t^2+6t &\leq 11t^3,\\
  2t^2+10t \leq 12t^2 &\leq 12 t^3,\\
  t^2+7t \leq 8t^2 &\leq 8t^3
\end{align}
which hold for all positive integers $t$.
}
\begin{proof}[Proof of Theorem~\ref{thm:Tcountdesign}]

	Notice that from~\eqref{eq:differencehaar}, we have the expression
	\begin{align}
	\|[P_{\Cl}&\mathrm{R}(K)]^k-P_{\mathrm{H}}\|_{\diamond}\\
	&=
  \left\|\left[\left(\sum_{j=t!+1}^{|\Sigma_{t,t}|}\frac{1}{\obraket{E_j}{E_j}}\oketbra{E_j}{E_j}\right)\mathrm{R}(K)\right]^k\right\|_{\diamond}\label{eq:fullexpression}\\
	&=
  \left\|\sum_{j_1,\dots, j_m=t!+1}^{|\Sigma_{t,t}|}\prod_{l=1}^k
  \frac{1}{\obraket{ E_{j_l}}{E_{j_l}}}
  \oket{E_{j_1}}\obra{E_{j_1}}R(K)\oket{E_{j_2}} \dots
  \obra{ E_{j_k}}R(K)\right\|_{\diamond}\\
	&\leq 
  \sum_{j_1,\dots, j_k=t!+1}^{|\Sigma_{t,t}|}\prod_{l=1}^k
  \frac{1}{\obraket{ E_{j_l}}{E_{j_l}}}
  \prod_{r=1}^{k-1} |\obra{E_{j_r}}R(K)\oket{E_{j_{r+1}}}| \cdot
  \Big\|\oketbra{E_{j_1}}{ E_{j_k}}\Big\|_{\diamond}.\label{eq:bigsum}
	\end{align}
  We now bound each of the factors in each term above.
	First, we compute the squared norm of $\oket{E_j}$,
	\begin{myalign}
		\obraket{E_j}{E_j} = \sum_{r,l=1}^j A_{r,j} A_{l,j} \obraket{Q_{T_r}}{Q_{T_l}}^n
		= 
    A_{j,j}^2 + \sum_{k,l< j} A_{r,j} A_{l,j} \obraket{Q_{T_k}}{Q_{T_l}}^n.
	\end{myalign}
	Using Eqs.~\eqref{eq:gramschmidtsecond} and~\eqref{eq:gramschmidtthird}, we thus bound
	\begin{myalign}
		\obraket{E_j}{E_j} 
    &\leq 
    \left(1 + 2^{t^2+7t-n}\right)^2 + (j^2-1) 4^{2t^2+10t-n} \\
    &\leq \left(1 + 2^{t^2+7t-n}\right)^2 + |\Sigma_{t,t}|^2 4^{2t^2+10t-n}\\
		&\leq 1 + 2^{31t^2 - 2n },
	\end{myalign}
	and in the same way
	\begin{equation}
	   \obraket{E_j}{E_j} \geq 1 - 2^{31t^2 - 2n }.
	\end{equation}
	Now we use that $n \geq 16t^2$.
	Letting $x:=2^{31t^2 - 2n }\in[0,\frac12]$, the inequalities $1/(1-x)\leq 1+2x$ 
	and $1-2x\leq 1/(1+x)$ hold.
	This leads to
	\begin{equation}\label{eq:boundinversenorm}
	   \frac{1}{\obraket{E_j}{E_j}}= 1+a_j\qquad\text{with}\qquad |a_j|\leq 2^{32t^2-2n}.
	\end{equation}

  We now focus on the second factor,
  \begin{align}
	   |\obra{E_i}R(K)\oket{E_j}|\leq \sum_{r=1}^i\sum_{l=1}^j
     |A_{r,i}A_{l,j}|\cdot \left|\obra{Q_{T_r}^{\otimes n}} R(K)\oket{Q_{T_l}^{\otimes n}}\right|.
  \end{align}
	If for $\obra{Q_{T_{r}}}R(K)\oket{Q_{T_{l}}}$ one of the stochastic Lagrangian sub-spaces 
  does not correspond to a permutation, Lemma~\ref{lemma:overlapnonclifford} introduces a 
  factor of $\eta_{K,t}$.
  \newtext{
	If both correspond to a permutation, we redefine the factors in a way that leads to simpler expressions in the calculations used below.
  Namely, in this case we redefine $A_{r,i}$ and $A_{l,j}$ by multiplying it with $2$.
  This is compensated by introducing a 
  factor of $\frac14$ and letting 
  \begin{equation}
	   \bar{\eta}_{K,t}:=\max\left\{\frac14,\eta_{K,t}\right\}.
	\end{equation}
We can do this as $i$ and $j$ do not correspond to permutations and hence $A_{r,j}$ and $A_{lj}$ are exponentially suppressed, which remains true after rescaling by $2$.}   In this case, moreover, $r<t!+1\leq i $ and $l< t!+1\leq j$, so the factor
  $|A_{r,i}A_{l,j}|$ will be exponentially suppressed according to~\eqref{eq:gramschmidtsecond} 
	and so this redefinition will not affect the asymptotic scaling in $n$.

  We provide two bounds for $|\obra{E_i}R(K)\oket{E_j}|$ that will be used later on.
  \newtext{We will use repeatedly that the diamond norm is multiplicative under the tensor product of superoperators~\cite[Thm.~3.49]{watrous2018theory}.}
	First, using~\eqref{eq:gramschmidtfirst},~\eqref{eq:gramschmidtthird}
  and~\eqref{eq:overlapsbound}, we obtain
  \begin{align}
	  & |\obra{E_i}R(K)\oket{E_j}|\leq \sum_{r=1}^i\sum_{l=1}^j
    |A_{r,i}A_{l,j}|\cdot \left|\obra{Q_{T_r}^{\otimes n}} R(K)\oket{Q_{T_l}^{\otimes n}}\right|\\
	  &\leq
    \bar{\eta}_{K,t}(1+2^{8t^2-n}) \sum_{r=1}^i\sum_{l=1}^j 
    2^{24t^3-n|\dim N_r-\dim N_i|-n|\dim N_l-\dim N_j|-(n-1)|\dim N_l-\dim N_r|}\label{eq:felipe q1}\\
	  &\leq 
    \bar{\eta}_{K,t}(1+2^{8t^2-n})|\Sigma_{t,t}|^{2} 2^{25t^3-n|\dim N_j-\dim N_i|}\label{eq:telescope}\\
	  &\leq
    \bar{\eta}_{K,t} (1+2^{8t^2-n}) 2^{31t^3-n|\dim N_j-\dim N_i|}
    \label{eq:boundonproductA},
	\end{align}
  where we have used $2^{|\dim N_l-\dim N_r|}\leq 2^t$\newtext{$\leq 2^{t^3}$, and the fact that for the rescaled $A_{r,i}$, the inequality~\eqref{eq:gramschmidtfirst} implies
  \begin{align*}
    A_{r,i}\leq 2^{11t^3 - |\dim  N_r-\dim N_j|+1}\leq2^{12t^3 - |\dim  N_r-\dim N_j|}
  \end{align*}
  for all $r,i$.}
  \newtext{
  Moreover, we have used the triangle inequality, 
  \begin{align}
    &|\dim N_r-\dim N_i|+|-\dim N_l+\dim N_j|+|\dim N_l-\dim N_r|\\
    \geq 
    &|\dim N_r-\dim N_i-\dim N_l+\dim N_j+\dim N_l-\dim N_r|\nonumber\\
    =
    &|\dim N_j - \dim N_i|,\nonumber
  \end{align}
  in the inequality~\eqref{eq:telescope}. 
  }
  The second bound follows from equations~\eqref{eq:gramschmidtsecond} 
  and~\eqref{eq:gramschmidtthird}, and we consider two cases. If $i\neq j$, then
	\begin{align}
	|\obra{E_i}R(K)\oket{ E_j}|&\leq \sum_{r=1}^i\sum_{l=1}^j |A_{r,i}A_{l,j}|\cdot|\obra{Q_{T_r}^{\otimes n}}R(K)\oket{Q_{T_l}^{\otimes n}}|\nonumber\\
	&\leq 
  \bar{\eta}_{K,t}(1+2^{8t^2-n})|\Sigma_{t,t}|^2 2^{19t^2-n}\nonumber\\
	&\leq \bar{\eta}_{K,t}(1+2^{8t^2-n})2^{25t^2-n}.\label{eq:boundonproductB}
	\end{align}
	Otherwise,
	\begin{align}
|\obra{E_i}R(K)\oket{ E_i}|\leq& \sum_{r=1}^i\sum_{l=1}^i |A_{r,i}A_{l,i}|\cdot|\obra{Q_{T_r}^{\otimes n}}R(K)\oket{Q_{T_l}^{\otimes n}}|\\
	\leq& 
  \bar{\eta}_{K,t}\left(
  |A_{i,i}|^2+(i^2-1) 2^{12t^2-n}
  \right)\label{eq:separate and bound}\\
	\leq& 
  \bar{\eta}_{K,t}\left(
  (1+2^{8t^2-n})^2+(1+2^{8t^2-n})2^{16t^2-n}
  \right)\label{eq:bound i}\\
	\leq& \bar{\eta}_{K,t}(1+2^{16t^2-n})^3.\label{eq:+ to x}
	\end{align}
  \newtext{
  In inequality~\eqref{eq:separate and bound}, we have bounded the term $r=l=i$ using~\eqref{eq:gramschmidtthird}, and each of the other terms using~\eqref{eq:gramschmidtsecond}. 
  Moreover, in the inequalities~\eqref{eq:bound i} and~\eqref{eq:+ to x} we use that $i\leq|\Sigma_{t,t}|$, and
  \begin{align*}
    1+2^{8t^2-n}\leq(1+2^{8t^2-n})^2\leq(1+2^{16t^2-n})^2.
  \end{align*}
  % $a+b\leq 2ab$ for $a,b\geq1$ and that $n\geq 16t^2+1$.
  }%end \newtext
	Lastly, we obtain from~\eqref{eq:gramschmidtfirst} and~\eqref{eq:diamondbound}
	\begin{align}
	   \|\oketbra{E_i}{E_j}\|_{\diamond}
     &\leq 
     \sum_{r=1}^i\sum_{l=1}^j|A_{r,i}A_{l,j}|\cdot
     \left\|\oketbra{Q_{T_r}^{\otimes n}}{ Q_{T_l}^{\otimes n}}\right\|_{\diamond}\\
	   &\leq 
     |\Sigma_{t,t}|^2 2^{24t^3-n|\dim N_r-\dim N_i|-n|\dim N_l-\dim N_j|+n(\dim N_l-\dim N_r)}\\
     &\leq 
     2^{30t^3+n(\dim N_j-\dim N_i)}.\label{eq:recapdiamondbound}
	\end{align}
	
  We now start piecing these expressions together to bound~\eqref{eq:bigsum}.
  Equations~\eqref{eq:recapdiamondbound} and~\eqref{eq:boundinversenorm} give
	\begin{multline}
	\|[P_{\Cl}\mathrm{R}(K)]^k-P_{\mathrm{H}}\|_{\diamond}\leq\\ \left(1+2^{32t^2-2n}\right)^k\sum_{j_1,\dots, j_k=t!+1}^{|\Sigma_{t,t}|}2^{30t^3+n(\dim N_{j_k}-\dim N_{j_1})}\prod_{r=1}^{k-1}|\obra{ E_{j_r}}R(K)\oket{E_{j_{r+1}}}|.
  \label{eq:first step main bound}
	\end{multline}
  To bound~\eqref{eq:first step main bound}, we will bunch together the contribution of all
  terms whose sequence $\{j_1, \dots, j_k\}$ contains $l$ changes. Moreover, we will treat
  differently the cases $l\leq\lfloor t/2\rfloor$ and $l>\lfloor t/2\rfloor$.
	In the former case, we use~\eqref{eq:boundonproductA} to get
	\begin{equation}
	   \prod_{r=1}^{k-1}|\obra{E_{j_r}}R(K)\oket{E_{j_{r+1}}}|%
     \leq% 
     \bar{\eta}_{K,t}^{k-1}(1+2^{16t^2-n})^{3(k-1)}2^{ l 31t^3-n|\dim N_{j_k}-\dim N_{j_1}|}.
	\end{equation}
  In this case, the factor of $2^{n(\dim N_{j_k}-\dim N_{j_1})}$ coming from~\eqref{eq:recapdiamondbound}
  is cancelled by the last factor of $2^{-n|\dim N_{j_k}-\dim N_{j_1}|}$.

	In the latter case, we turn to~\eqref{eq:boundonproductB} instead to obtain
  \begin{align*}
    \prod_{r=1}^{k-1}|\obra{E_{j_r}}R(K)\oket{E_{j_{r+1}}}|%
    \leq% 
    \bar{\eta}_{K,t}^{k-1}(1+2^{16t^2-n})^{3(k-1)} 2^{l 25 t^2 - ln}.
  \end{align*}
  Here, the exponential factor coming from~\eqref{eq:recapdiamondbound} is cancelled by 
  $2^{-ln}$ since $\dim N_{j_k}-\dim N_{j_1} \leq \lfloor t/2\rfloor$.
	Counting the instances of sequences with $l$ changes, we may put these considerations
  together to bound
	\begin{align*}
	\|[P_{\Cl}\mathrm{R}(K)]^k-P_{\mathrm{H}}\|_{\diamond}
	\leq& 
  \left(1+2^{32t^2-2n}\right)^k\left(1+2^{16t^2-n}\right)^{3(k-1)}\bar{\eta}_{K,t}^{k-1}
  \Bigg{[}\sum_{l=0}^{\lfloor\frac{t}{2}\rfloor} {k\choose l}|\Sigma_{t,t}|^{l+1}2^{l 31t^3}\\
	&+\sum_{l=\lfloor \frac{t}{2}\rfloor+1}^{k}{k\choose l}|\Sigma_{t,t}|^{l+1}2^{(l-\lfloor \frac{t}{2}\rfloor)(25t^2-n)}2^{\lfloor\frac{t}{2}\rfloor 25t^2}\Bigg]\\
	\leq& 
  \left(1+2^{32t^2-2n}\right)^{4k}\bar{\eta}_{K,t}^{k-1}
  \Bigg[\frac{t}{2}{k\choose\lfloor\frac{t}{2}\rfloor}|\Sigma_{t,t}|^{\lfloor \frac{t}{2}\rfloor+1}2^{\lfloor \frac{t}{2}\rfloor 31t^3}
	\\
	&+\sum_{l=1}^{k-\lfloor \frac{t}{2}\rfloor}{k\choose l+\lfloor \frac{t}{2}\rfloor}|\Sigma_{t,t}|^{l+1+\lfloor \frac{t}{2}\rfloor}2^{l(25t^2-n)}2^{13t^3}\Bigg]\\
	\stackrel{\ddagger}{\leq}& 
  \left(1+2^{32t^2-2n}\right)^{4k}\bar{\eta}_{K,t}^{k-1}\Bigg[ 2^{32t^4+t\log(k)}\\
	&+k^{\lfloor \frac{t}{2}\rfloor}|\Sigma_{t,t}|^{1+\lfloor \frac{t}{2}\rfloor}2^{13t^3}\sum_{l=0}^{k}{k\choose l}|\Sigma_{t,t}|^{l}2^{l(25t^2-n)}\Bigg]\\
	\leq & 
  \left(1+2^{32t^2-2n}\right)^{4k}\bar{\eta}_{K,t}^{k-1}\Bigg[ 2^{32t^4+t\log(k)}+2^{18t^3+\log(k)t}\left(1+2^{28t^2-n}\right)^{k}\bigg]\\
  \newtext{ \leq } &
  \newtext{
  \left(1+2^{32t^2-2n}\right)^{4k}\left(1+2^{28t^2-n}\right)^{k} 2^{t\log(k)}\bar{\eta}_{K,t}^{k-1}\Bigg[ 2^{32t^4}+2^{18t^3}\bigg],
  }
	\end{align*}
	where we have used in $\ddagger$ that
	\begin{align*}
	{k\choose l+\lfloor \frac{t}{2}\rfloor}&=\frac{(k)!}{(k-l-\lfloor \frac{t}{2}\rfloor)!(l+\lfloor \frac{t}{2}\rfloor)!}\\
	&\leq (k-l-\big\lfloor \frac{t}{2}\big\rfloor+1)\dots (k-l)\frac{k!}{(k-l)!l!}\\
	&\leq k^{\lfloor \frac{t}{2}\rfloor}{k\choose l}.
	\end{align*}
	\newtext{Finally, noting that $2^{32t^4}+2^{18t^3}\leq 2^{33t^4}$  for all positive integers $t$,} we obtain the bound
	\begin{equation}
	\| \MO_t(\sigma_{k}) - P_{\mathrm{H}} \|_{\diamond} \leq  
  2^{33t^4+t\log(k)}\left(1+2^{32t^2-n}\right)^{5k}\bar{\eta}_{K,t}^{k-1}\label{eq:finalcombinedbound},
	\end{equation}
	where $\bar{\eta}_{K,t}$ is bounded by Lemma~\ref{lemma:overlapnonclifford}.
	Taking the logarithm and using the inequality $\log(1+x)\leq x$ repeatedly, this implies Theorem~\ref{thm:Tcountdesign}.
\end{proof}
With the above bound, we can also prove Corollary~\ref{cor:dropindependence}.
\begin{proof}[Proof of Corollary~\ref{cor:dropindependence}]
 Consider the self-adjoint superoperator $A:=P_{\Cl}R(K)P_{\Cl}$.
 As $P_{\Cl}$ is a projector, we have with Eq.~\eqref{eq:Aidentity}
 \begin{equation}
  (A-P_\Haar)^k = A^k - P_\Haar = \left[P_{\Cl} R(K)\right]^k - P_\Haar = \MO_t(\sigma_{k}) - P_{\mathrm{H}}.
 \end{equation}
 Using norm inequality between operator and diamond norm Eq.~\eqref{eq:dnorm-bound} and the previous result Eq.~\eqref{eq:finalcombinedbound}, we find
 \begin{multline}
 ||A-P_H||^{k}_{\infty}=||(A-P_H)^k||_{\infty}\leq 2^{nt/2}\| \MO_t(\sigma_{k}) - P_{\mathrm{H}} \|_{\diamond}\\
  \leq  
 2^{33t^4+t\log(k)+nt/2}\left(1+2^{32t^2-n}\right)^{5k}\bar{\eta}_{K,t}^{k-1}.
 \end{multline}
 \newtext{
 Taking the $k$-th square root of the expresion above, we obtain a sequence of infinitely many bounds for $||A-P_H||_{\infty}$ which converges as $k\to\infty$.
 That limit gives
 }%end \newtext
 \begin{equation}\label{eq:finalspectralbound}
 ||A-P_H||_{\infty}\leq \left(1+2^{32t^2-n}\right)^{5}\bar{\eta}_{K,t}.
 \end{equation}
 Combined with Ref.~\cite[Lem.~4]{brandao_local_2016}, Eq.~\eqref{eq:finalspectralbound} implies the result.
\end{proof}
The bound in Eq.~\eqref{eq:finalcombinedbound} also suffices to prove Proposition~\ref{prop:Haarinverleavedversion}:
\begin{proof}[Proof of Proposition~\ref{prop:Haarinverleavedversion}]
 The proof follows exactly as the proof of Theorem~\ref{thm:Tcountdesign}, but with the factor $7/8$ instead of $\bar{\eta}_{K,t}$ (compare Lemma~\ref{lemma:haarsymmetrization}).
Using $\log_2(7/8)\leq -0.19$ the result can be checked.
\end{proof}

\section{Convergence to higher moments of the Clifford group}\label{section:clifford-design}

In this section, we aim to prove:

\clifforddesign*

The proof of Theorem~\ref{theorem:clifford-design} follows a well-established strategy~\cite{brandao_local_2016,brown_convergence_2010} in a sequence of lemmas. For the sake of readibility, the proofs of these lemmas have been moved to Section~\ref{section:clifford-lemmas}. Given a measure $\nu$ on the Clifford group $\Cl(n)$, recall that its $t$-th moment operator was defined as
\begin{equation*}
 \MO_t(\nu) := \int_{\rm Cl(2^n)} \Ad_U^{\otimes t}\mathrm{d}\nu(U).
\end{equation*}
The idea of the proof is that if $\MO_t(\nu)$ is close to the moment operator $\MO_t(\mu_{\Cl})\equiv P_{\Cl}$ of the uniform (Haar) measure $\mu_{\rm Cl}$ on the Clifford group, $\nu$ is an approximate Clifford design. However, we have seen that there are different notions of closeness. We define its deviation in (superoperator) \emph{spectral norm} as
\begin{equation*}
g_{\rm Cl}(\nu,t) := \left\| \MO_t(\nu) - \MO_t(\mu_{\Cl}) \right\|_{\infty}.
\end{equation*}
Then, we prove the following lemma in Section~\ref{section:clifford-lemmas}.

\begin{restatable}[Relative $\varepsilon 2^{2tn}$-approximate Clifford $t$-designs]{lemma}{reductionspectral}
 \label{lemma:reduction-spectral}
 Suppose that $0 \leq \varepsilon < 1$ is such that $g_{\rm Cl}(\nu,t)\leq\varepsilon$. Then, $\nu$ is a relative $\varepsilon 2^{2tn}$-approximate Clifford $t$-design.
\end{restatable}

Recall that we have defined the measure $\sigma_G$ on the Clifford group $\Cl(n)$ in Def.~\ref{def:random-cliff-circuit} by randomly drawing from a 2-local Clifford gate set $G$ and applying it to a random qubit $i$, or to a pair of adjacent qubits $(i,i+1)$, respectively. For this measure, we show that it fulfills the assumptions of Lemma~\ref{lemma:reduction-spectral}:

\begin{proposition}[Clifford expander bound]
 \label{prop:clifford-expander-bound}
 Let $\sigma_G$ be as in Def.~\ref{def:random-cliff-circuit} and $n \geq 12t$. Then, $g_{\rm Cl}(\sigma_G,t) \leq 1 - c(G) n^{-1} \log^2(t) t^{-8}$ for some constant $c(G)>0$.
\end{proposition}
We will prove Proposition~\ref{prop:clifford-expander-bound} in the end of this section. From this, Theorem~\ref{theorem:clifford-design} follows as a direct consequence:

\begin{proof}[Proof of Theorem~\ref{theorem:clifford-design}]
First, note that $g_{\Cl}(\nu^{*k},t)=g_{\Cl}(\nu,t)^k$ for all probability measures $\nu$ on the Clifford group.
This can be easily verified using the observation
\begin{equation}
\MO_t(\mu_{\Cl}) \MO_t(\nu) = \MO_t(\nu) \MO_t(\mu_{\Cl}) = \MO_t(\mu_{\Cl}).
\end{equation}
Hence, combining the bound given by Proposition~\ref{prop:clifford-expander-bound} and Lemma~\ref{lemma:reduction-spectral}, we find that the $k$-step random walk $\sigma_G^{*k}$ is a $\varepsilon$-approximate Clifford $t$-design, if we choose $k = O\left(n \log^{-2}(t) t^8\left( 2nt + \log (1/\varepsilon) \right) \right)$.
\end{proof}

For the sake of readibility, let us from now on drop the dependence on $G$ and write $\sigma\equiv \sigma_G$. In order to prove Proposition~\ref{prop:clifford-expander-bound}, we use a reformulation of $g(\sigma,t)$ based on the following observation. Since $G$ is closed under taking inverses, the moment operator $\MO_t(\sigma)$ is self-adjoint with respect to the Hilbert-Schmidt inner product. Due to $\sigma$ being a probability measure, its largest eigenvalue is 1 with eigenspace corresponding to the operator subspace which is fixed by the adjoint action $\Ad(g^{\otimes t})$ of all generators \cite{brown_convergence_2010}. Equivalently, this is the subspace of operators which commute with any generator $g^{\otimes t}$. However, any operator commuting with all generators also commutes with every element in the Clifford group $\Cl(n)$ and vice versa. Hence, this subspace is nothing but the Clifford commutant $\Cl(n)'$ with projector $P_{\Cl}
%\equiv
\newtext{:=}
\MO_t(\mu_{\Cl})$. 
Thus, the spectral decomposition is
\begin{equation}
\label{eq:random-walk-decomp}
 \MO_t(\sigma) = P_{\Cl} + \sum_{r \geq 2} \lambda_r(\MO_t(\sigma)) \Pi_r,
\end{equation}
where $\lambda_r(X)$ denotes the $r$-th largest eigenvalue of a normal operator $X$.
\newtext{
Hence, we find
\begin{equation}
	g(\sigma,t)= \left\| \MO_t(\sigma) - P_{\Cl} \right\|_{\infty} = \lambda_*\left(\MO_t(\sigma)\right) := \max\left\{ \lambda_2\left(\MO_t(\sigma)\right), |\lambda_\mathrm{min}\left(\MO_t(\sigma)\right)| \right\},
\end{equation}
where $\lambda_\mathrm{min}\left(\MO_t(\sigma)\right)$	is the smallest eigenvalues of $\MO_t(\sigma)$.
We continue by arguing that it sufficient to consider the case when $\lambda_*\left(\MO_t(\sigma)\right) = \lambda_2\left(\MO_t(\sigma)\right) > 0$.

To this end, consider the linear operator $T_{\sigma}:L^2(\Cl(n))\to L^2(\Cl(n))$ given as
\begin{equation}
\label{eq:hecke-operator}
T_{\sigma}f(g):=\int f(h^{-1}g)\mathrm{d}\sigma(h).
\end{equation}
This is the (Hermitian) averaging operator with respect to $\sigma$ on the group algebra $L^2(\Cl(n))$.
The largest eigenvalue of
$T_\sigma$ is $\lambda_1(T_\sigma)=1$ and its eigenspace corresponds to the trivial representation.
By Ref.~\cite[Lem.~1]{diaconis_random_1993}, its smallest eigenvalue is lower bounded by
\begin{equation}
 \lambda_\mathrm{min}(T_\sigma) \geq -1 + 2 \sigma(\one) = -1 + \frac{2}{|G|},
\end{equation}
where $\sigma(\one) \equiv \sigma(\{\one\})=1/|G|$ is the probability of drawing the identity.
According to the Peter-Weyl theorem, the spectrum of $\MO_t(\sigma)$ is exactly the spectrum of the restriction of $T_\sigma$
to the irreducible representations  that appear in the representation $U\mapsto \Ad_U^{\otimes t}$.
In particular, we find $\lambda_\mathrm{min}(\MO_t(\sigma)) \geq -1 + \frac{2}{|G|}$.
Let us assume that $\lambda_*\left(\MO_t(\sigma)\right) = |\lambda_\mathrm{min}\left(\MO_t(\sigma)\right)|$.
Then, $g(\sigma,t) \leq 1 - 2/|G| < 1$ and hence we can argue as in the proof of Thm.~\ref{theorem:clifford-design} to show that local random Clifford circuits form relative $\varepsilon$-approximate Clifford $t$-designs in depth $O(2nt + \log(1/\varepsilon))$.

Therefore, we consider the more relevant case when $\lambda_*\left(\MO_t(\sigma)\right) = \lambda_2\left(\MO_t(\sigma)\right) > 0$ in the following, this is
\begin{equation}
\label{eq:reduction-eigenvalue}
	g(\sigma,t)= \left\| \MO_t(\sigma) - P_{\Cl} \right\|_{\infty} = \lambda_2\left(\MO_t(\sigma)\right).
\end{equation}
Since $\MO_t(\sigma)$ is self-adjoint, we can interpret it as an Hamiltonian on the Hilbert space $L( (\C^2)^{\otimes nt})$.}
In this light, it will turn out to be useful to recast Eq.~\eqref{eq:reduction-eigenvalue} as the spectral gap of a suitable family of \emph{local Hamiltonians} with vanishing ground state energy:
\begin{equation}
\label{eq:reduction-Hamiltonian}
H_{n,t}:=n \left( \id - \MO_t(\sigma) \right) = \sum_{i=1}^n h_{i,i+1}, \quad \text{with}\quad h_{i,i+1}:=\frac{1}{|G|}\sum_{g\in G}\left(\id- \Ad(g_{i,i+1}^{\otimes t})\right).
\end{equation}
Let us summarize these findings in the following lemmas.

\begin{lemma}[Spectral gap]
\label{lemma:reduction-gap}
Let $\sigma$ be as in Def.~\ref{def:random-cliff-circuit} and $H_{n,t}$ the Hamiltonian from Eq.~\eqref{eq:reduction-Hamiltonian}. It holds that
	\begin{equation}
	g(\sigma,t)=1-\frac{\Delta(H_{n,t})}{n}.
	\end{equation}
\end{lemma}

\begin{lemma}[Ground spaces]\label{lemma:ground-space}
	The Hamiltonians $H_{n,t}$ are positive operators with ground state energy $0$.
	The ground space is given by the Clifford commutant
	\begin{equation}
	\Cl(n)'=\spann\left\{ r(T)^{\otimes n} \; \big| \; T \in \Sigma_{t,t} \right\},
	\end{equation}
	where $\Sigma_{t,t}$ is the set of stochastic Lagrangian sub-spaces of $ \mathbb{Z}_2^t\oplus \mathbb{Z}_2^t$.
\end{lemma}

In the remainder of this section, we will prove the existence of a uniform lower bound on the spectral gap of $H_{n,t}$.
In combination with Lemma~\ref{lemma:reduction-gap} and Lemma~\ref{lemma:reduction-spectral} this will imply Theorem~\ref{theorem:clifford-design}.
While it is highly non-trivial to show spectral gaps in the thermodynamic limits, we can use the fact that $H_{n,t}$ is \emph{frustration-free} (compare Lemma~\ref{lemma:ground-space}).
This allows us to apply the powerful \emph{martingale method} pioneered by Nachtergaele~\cite{nachtergaele_gap_1994}.

\begin{restatable}[Lower bound to spectral gap]{lemma}{spectralgapbound}

\label{lemma:spectralgapbound}
Let the Hamiltonian $H_{n,t}$ be as in Eq.~\eqref{eq:reduction-Hamiltonian} and assume that  $n\geq 12t$. Then, $H_{n,t}$ has a spectral gap
satisfying
	\begin{equation}\label{eq:lowerboundgap}
	   \Delta(H_{n,t})\geq \frac{\Delta(H_{12t,t})}{48t}.
	\end{equation}
\end{restatable}

\begin{proof}[Proof of Proposition~\ref{prop:clifford-expander-bound}]
We can now combine the bound in~\eqref{eq:lowerboundgap} with any lower bound on the spectral gap independent of $t$.
\newtext{To this end, we make again use of the averaging operator $T_{\sigma}:L^2(\Cl(n))\to L^2(\Cl(n))$ introduced in Eq.~\eqref{eq:hecke-operator} before.
By Ref.~\cite[Cor.~1]{diaconis_random_1993}
we have that
\begin{equation}\label{eq:diaconis}
\lambda_2(T_\sigma)\leq 1-\frac{\eta}{d^2},
\end{equation}
where $\eta$ is the probability of the least probable generator (here $1/|G|n$) and $d$ is the diameter of
the associated Cayley graph (given in Ref.~\cite{aaronson_clifford_2004} as $d=O(n^3/\log(n))$.

Since the representation $U\mapsto \Ad_U^{\otimes t}$ contains a trivial component, the second largest eigenvalue of $\MO_t(\sigma)$ can be at most $\lambda_2(T_\sigma)$.
Thus, $H_{n,t}$ has a gap of at least $\eta/d^2$.
}
Finally, by
Lemma~\ref{lemma:spectralgapbound} it follows that
\begin{myalign}
  \Delta(H_{n,t})\geq \frac{\Delta(H_{12t,t})}{48t} \geq c(G) t^{-8}\log(t)^2,
\end{myalign}
for a constant $c(G)$.
We note that the applicability of Ref.~\cite[Cor.~1]{diaconis_random_1993} to random walks on the Clifford group has also been observed in Ref.~\cite{divinzeno_data_2001}.

\end{proof}

We can combine Theorem~\ref{theorem:clifford-design} and Theorem~\ref{thm:Tcountdesign} to obtain the following corollary:
\combined*
\begin{proof}
	Consider the superoperator
	\begin{equation}
	\MO_t(\sigma_{k,m}) = \int_{U(2^n)} \Ad(U^{\otimes t})\,\mathrm{d}\sigma_{k,m}(U)=\underbrace{ \MO_t(\sigma^{*m})\mathrm{R}(K)\dots \MO_t(\sigma^{*m})\mathrm{R}(K)}_{\text{$k$ times}},
	\end{equation}
	where $\sigma^{*m}$ denotes the probability measure of a depth $m$ local random walk on the Clifford group (cp.~Def.~\ref{def:random-cliff-circuit}).
	We would like to bound the difference between the Haar random $t$-th moment operator $\MO_t(\haar)=:P_\mathrm{H}$ and $\MO_t(\sigma_{k,m})$.
    Notice the following standard properties of $P_\mathrm{H}$:
	\begin{equation}\label{eq:relationPHaar}
	P_{\mathrm{H}}\MO_t(\nu)=\MO_t(\nu)P_{\mathrm{H}}=P_{\mathrm{H}},\qquad\text{and} \qquad P_{\mathrm{H}}^{\dagger}=P_{\mathrm{H}},
	\end{equation}
	for any probability measure $\nu$ on $U(2^n)$.
	In particular, we have that $P_{\mathrm{H}}$ is an orthogonal projector.
	As in the last section, we make use of the spectral decomposition in Eq.~\eqref{eq:random-walk-decomp} to decompose $\MO_t(\sigma^{*k})$ as follows:
	\begin{myalign}
		\MO_t(\sigma_{k,m}) - P_{\mathrm{H}}
		&= \left[ \MO_t(\sigma^{*m})\mathrm{R}(K) \right]^k - P_{\mathrm{H}}\\
		&= \left[ \bigg(P_{\Cl} + \sum_{i\geq 2}\lambda^m_i \Pi_i \bigg)\mathrm{R}(K) \right]^k - P_{\mathrm{H}}.
	\end{myalign}
Recall the shorthand notation $P_{\Cl}:= \MO_t(\mu_{\Cl})$.
	Using the triangle inequality and the inequality \eqref{eq:dnorm-bound}, this implies
	\begin{myalign}
		\dnorm{\MO_t(\sigma_{k,m})-P_{\mathrm{H}}}
		&\leq \dnormb{ [P_{\Cl}\mathrm{R}(K)]^k -P_{\mathrm{H}}} + 2^{2tn}\sum_{l=1}^{k}{k\choose l}\lambda_{2}^{lm}\\
		&\leq  \dnormb{ [P_{\Cl}\mathrm{R}(K)]^k -P_{\mathrm{H}}} + k2^{2tn+1}\lambda_2^m.
	\end{myalign}
	Note that we bounded the second largest eigenvalue $\lambda_2$ of $\MO_t(\sigma)$ in Proposition~\ref{prop:clifford-expander-bound}.
We can now combine Proposition~\ref{prop:clifford-expander-bound} with~\eqref{eq:finalcombinedbound} to obtain:
	\begin{equation}
\| \MO_t(\sigma_{k,m}) - P_{\mathrm{H}} \|_{\diamond} \leq k2^{2tn+1}\lambda_2^m+ 2^{33t^4+t\log(k)}\left(1+2^{32t^2-n}\right)^{5k}\bar{\eta}_{K,t}^k.
\end{equation}
\end{proof}

\section{Singling out the Clifford group}\label{section:isomorphism}

There are a number of ways to motivate the construction of approximate unitary $t$-designs from random Clifford circuits.
\newtext{From a practical} point of view, Clifford gates are often comparatively easy to implement, in particular in fault-tolerant architectures.
In this section, we point out that Refs.~\cite{bannai_unitary_2020, sawicki_universal_2017} together imply that the Clifford groups are also mathematically distinguished.
\newtext{
We formulate this observation as Proposition~\ref{theorem:singlingoutclifford}: 
The finite case follows from the recently obtained classification of finite unitary subgroups forming $t$-designs, so-called \emph{unitary $t$-groups}, by \textcite{bannai_unitary_2020} building on earlier results by \textcite{guralnick_larsen_2005}.
The infinite case is a corollary of a theorem about universality of finitely generated subgroups by \textcite{sawicki_universal_2017}.
}

\newtext{
This section is independent from the rest of the paper and has the sole purpose of highlighting the results in Refs.~\cite{guralnick_larsen_2005,bannai_unitary_2020,sawicki_universal_2017} and explicitly formulate their combined implications for the generation of unitary $t$-designs.
Moreover, it might serve as an intuitive justification for the usefulness and omnipresence of Clifford unitaries in random circuit constructions.
}

For any subgroup $G\subseteq \mathrm{U}(d)$, we let
\begin{align*}
  \overline{G}:=\{\det(U^\dagger)U\,|\, U\in G\}\subseteq \SU(d).
\end{align*}
Notice that $\overline{G}$ is a unitary $t$-design if and only if $G$ is.

Proposition~\ref{theorem:singlingoutclifford} refers to $t$-designs generated by \emph{finite
gate sets}, which we define now. The starting point is a Hilbert space $(\C^q)^{\otimes r}$
for some $r$. A finite gate set is a finite subset
\begin{align*}
  \mathcal{G}\subset \SU\big((\C^q)^{\otimes r}\big).
\end{align*}
We will denote by $\mathcal{G}_n$ the subgroup of $\SU\big((\C^q)^{\otimes n}\big)$
generated by elements of $\mathcal{G}$ acting on any $r$ tensor factors (here $r\leq n$).
The number $q$ is called the \emph{local dimension} of $\mathcal{G}$.

\begin{restatable}[Singling out the Clifford group~\cite{guralnick_larsen_2005,bannai_unitary_2020,sawicki_universal_2017}]{proposition}{clifforduniqueness}
	\label{theorem:singlingoutclifford}

  Let $t\geq2$, and let $\mathcal{G}$ be a finite gate set with local dimension $q\geq 2$.
  Assume that (1) either all $\mathcal{G}_n$ are finite or they are all infinite, and (2) there is an $n_0$ such that for
  all $n\geq n_0$, $\mathcal{G}_n$ is a unitary $t$-design.

	Then, one of the following cases apply:	\begin{enumerate}[label=(\roman*)]
		\item If $t=2$, we have either $q$ prime and $\mathcal{G}_n$ is isomorphic to a subgroup of the Clifford group $\overline{\Cl}(q^n)$, \emph{or} $\mathcal{G}_n$ is dense in $\mathrm{SU}(q^n)$,
		\item If $t=3$, we have \emph{either} $q=2$ and $\mathcal{G}_n$ is isomorphic to the full Clifford group $\overline{\Cl}(2^n)$ \emph{or} $\mathcal{G}_n$ is dense in $\mathrm{SU}(q^n)$,
		\item If $t \geq 4$ then $\mathcal{G}_n$ is dense in $\mathrm{SU}(q^n)$.
	\end{enumerate}
\end{restatable}

	Note that a finitely generated infinite subgroup of $\SU(d)$ is always dense in some compact Lie subgroup (cp.~\cite[Fact 2.6]{sawicki_universal_2017}). In particular, it inherits a Haar measure from this Lie subgroup which allows for a definition of unitary $t$-design.

\paragraph{Finite case.} In the classification in Ref.~\cite{bannai_unitary_2020}, the non-existence of finite unitary $t$-groups was shown for $t \geq 4$ (and dimension $d>2$). Already the case $t=3$ is very restrictive, since the authors arrive at the following result:

\begin{lemma}[{Ref.~\cite[Thm.~4]{bannai_unitary_2020}}]
	Suppose $d\geq 5$ and \newtext{consider a finite subgroup $H < \SU(d)$ which is a unitary 3-design.} Then, $H$ is either one of finitely many exceptional cases or $d=2^n$ and $H$ is isomorphic to the Clifford group $\overline{\Cl}(2^n)$.
\end{lemma}
This establishes the finite version of $(ii)$, the $t=3$ case.

The classification of unitary 2-designs is however more involved, it includes certain irreducible
representations of finite unitary and symplectic groups (compare~\cite[Thm.~3~Lie-type case]{bannai_unitary_2020}),
and a finite set of exceptions. The exceptions can be ruled out in the same way as above.

The former, the Lie-type cases, happen in dimensions $(3^n\pm1)/2$ and $(2^n+(-1)^n)/3$. There
is no $q$ for which there exists an $n_0$ such that for all $n\geq n_0$ there
exists an $m\in \mathbb{N}$ satisfying either
\begin{align*}
  q^n = (3^m\pm1)/2 \qquad\text{or}\qquad q^n = (2^m+(-1)^m)/3.
\end{align*}
Thus, the assumptions of Prop.~\ref{theorem:singlingoutclifford} rule these out. This establishes
the finite version of $(i)$.

\paragraph{Infinite case.}
Define the commutant for a set $S\subset \SU(d)$ of the adjoint action as
\begin{equation}
\Comm(\mathrm{Ad}_S) \nonumber:=\left\{ L\in \End\left(\mathbb{C}^{d\times d}\right) \,\big| \, [\Ad_g,L]=0\;\;\forall g\in S\right\}.
\end{equation}
We show that the second case can be reduced to Cor.~3.5 from Ref.~\cite{sawicki_universal_2017} applied to the simple Lie group $\SU(d)$.

\begin{lemma}[{\cite[Cor.~3.5]{sawicki_universal_2017}}]
	\label{lemma:adjointcommutant}
	Given a finite set $G\subset \SU(d)$ such that $\mathcal G = \langle G \rangle$ is infinite.
	Then, the group $\mathcal G $ is dense in $\SU(d)$ if and only if
	\begin{equation}
	\Comm(\Ad_{\mathcal G })\cap \End(\mathfrak{su}(d))=\{\lambda \, \id_{\mathfrak{su}(d)} \,|\,\lambda\in\mathbb{R}\}.
	\end{equation}
\end{lemma}

Recall that a subgroup $\mathcal{G}\subseteq U(d)$ is a unitary $2$-group if and only if  $\Comm(U\otimes U|U\in\mathcal{G})=\Comm(U\otimes U|U\in\U(d))=\spann(\one,\mathbb{F})$, where $\mathbb{F}$ denotes the flip of two tensor copies (see also App.~\ref{app:tdesigns}
).
Let us denote the partial transpose on the second system of a linear operator $A\in L(\mathbb{C}^{d}\otimes\mathbb{C}^{d})$ by $A^{\Gamma}$.
Then, one can easily verify that $\Gamma$ induces a vector space isomorphism between $\Comm(U\otimes U|U\in\mathcal{G})$ and $\Comm(U\otimes \overline U|U\in\mathcal{G})$.
The image of the basis $\{\one,\mathbb{F}\}$ is readily computed as
\begin{equation}
\one^{\Gamma} = \one, \quad\quad\quad \mathbb F^{\Gamma} = d \ketbra{\Omega}{\Omega},
\end{equation}
where $\ket\Omega = d^{-1/2}\sum_{i=1}^d \ket{ii}$ is the maximally entangled state vector.
Next, we use that $U\otimes\overline{U}=\matmap(\Ad_U)$ is the matrix representation of $\Ad_U=U\cdot U^\dagger$ with respect to the basis $E_{i,j}=\ketbra{i}{j}$ of $L(\C^d)$.
Thus, we have $\Comm(\Ad_{\mathcal G}) \simeq \Comm(U\otimes \overline{U}|U\in\mathcal G)$ as algebras.
Pulling the above basis of $\Comm(U\otimes \overline{U}|U\in\mathcal G)$ back to $\Comm(\Ad_{\mathcal G})$, we then find:
\begin{equation}
\matmap^{-1}(\one) = \id_{L(\C^d)}, \quad \matmap^{-1}(\ketbra{\Omega}{\Omega}) = \tr(\bullet) \id_{L(\C^d)}.
\end{equation}
Hence, we have shown that any element in $\Comm(\Ad_{\mathcal G})$ is a linear combination of these two maps. However, by restricting to $\mathfrak{su}(d)$, the second map becomes identically zero, thus we have
\begin{equation}
\Comm(\Ad_{\mathcal G })\cap \End(\mathfrak{su}(d))=\{\lambda\, \id_{\mathfrak{su}(d)} \,|\,\lambda\in\mathbb{R}\}.
\end{equation}
By Lemma \ref{lemma:adjointcommutant}, this shows that any finitely generated infinite unitary 2-group $\mathcal{G}\leq\SU(d)$ is dense in $\SU(d)$.
Since any unitary $t$-group is in particular a 2-group, this is also true for any $t>2$.

\section{Proofs}

\subsection{Proof of overlap lemmas}
\label{sec:overlaps}

In this section, we prove three technical lemmas which are needed throughout this paper.
\newtext{These lemmas give bounds on the overlaps of the operators $Q_T^{\otimes n}$ and hence quantify how far this basis is from an orthonormal basis of the commutant of the Clifford tensor power representation, i.e., for $\mathrm{range}\ P_{\Cl}$.}

\diamondbound*
\begin{proof}
	First, recall that $Q_T:= 2^{-t/2}r(T)$. Then, we make use of the following elementary bound on the diamond norm of rank one superoperator $\oketbra{A}{B}$:
	\begin{myalign}
	 \dnorm{\oketbra{A}{B}} &= \sup_{\tnorm{X}=1} \tnorm{A\otimes \tr_1\left(B\otimes\one X\right)} \\
	 &\stackrel{\dagger}{\leq} \tnorm{A} \sup_{\tnorm{X}=1} \tnorm{B\otimes\one X} \\
	 &\stackrel{\ddagger}{=} \tnorm{A} \snorm{B\otimes\one} \\
	 &= \tnorm{A} \snorm{B}.
	\end{myalign}
	Here, we have used in $\dagger$ that the partial trace is a contraction w.r.t.~$\tnorm{\cdot}$ and in $\ddagger$ a version of the duality between trace and spectral norm \cite{bhatia_book}.  
	Given stochastic Lagrangians $T_1$ and $T_2$ with defect spaces $N_1$ and $N_2$, we thus find using Lem.~\ref{lemma:rT-norms}:
	\begin{myalign}
	 \dnorm{\oketbra{Q_{T_1}}{Q_{T_2}}} \leq 2^{-t} \tnorm{r(T_1)}\snorm{r(T_2)} = 2^{\dim N_2 - \dim N_1}.
	\end{myalign}

	To prove $2.$, we use Ref.~\cite[Eq.~(4.25)]{gross2017schur} and that the transpose does not change the dimension of the corresponding defect subspace.
	Moreover, we assume w.l.o.g.~that $\dim N_2\geq \dim N_1$.
	We have
	\begin{equation}
	|\obraket{Q_{T_1}}{Q_{T_2}}|=
	2^{-t}|\Tr[r(T_1)r(T_2)^T]|=2^{-t+\dim( N_1\cap N_2)}|\Tr[r(T)]|
	\end{equation}
	where $r(T)$ is described by a stochastic orthogonal and a defect space $N^{\perp}_1\cap N_2+N_1$.
	Hence, we obtain (together with H{\"o}lder's inequality):
	\begin{equation}
	|\obraket{Q_{T_1}}{Q_{T_2}}|\leq
	2^{-t+\dim( N_1\cap N_2)}2^{t-\dim(N_1^{\perp}\cap N_2+N_1)}.
	\end{equation}
	Using $N\subseteq N^{\perp}$ for all defect spaces and the general identity $\dim(V+W)=\dim V+\dim W-\dim(V\cap W)$, this yields
	\begin{equation}
	|\obraket{Q_{T_1}}{Q_{T_2}}|
	\leq 2^{\dim(N_1\cap N_2)-\dim N_1}\leq 2^{\dim N_2-\dim N_1}.
	\end{equation}
\end{proof}

\newtext{Next, we define a \emph{frame operator} associated to the basis $Q_T^{\otimes n}$.
If the basis was orthogonal, this frame operator would simply be the projector $P_{\Cl}$ onto the Clifford commutant.

\begin{definition}[Clifford frame operator]
 We define the Clifford frame operator of the basis $Q_T^{\otimes n}$ as
 \begin{equation}
  S_{\Cl} := \sum_{T\in\Sigma_{t,t}} \oketbra{Q_T}{Q_T}^{\otimes n}.
 \end{equation}
\end{definition}

Hence, a quantifier for the orthogonality of the $Q_T^{\otimes n}$ basis is the distance of $S_{\Cl}$ to the projector $P_{\Cl}$.
As we prove in Lem.~\ref{lemma:Snormbound}, we have $P_{\Cl} \approx S_{\Cl}$ in spectral norm and we will use this result later in the proof of Lem.~\ref{lemma:spectralgapbound}.
In order to show this, we first derive a result on the \emph{sum of overlaps} in Lem.~\ref{lemma:overlaps}.

Interestingly, $S_{\Cl}$ is \emph{not} close to $P_{\Cl}$ in diamond norm (see.~Ch.~15 in Ref.~\cite{heinrich_2021}). 
To derive our main result, we instead construct an orthogonalized basis from the $Q_T^{\otimes n}$.
Some properties of the orthogonalized basis are proven in Lem.~\ref{lemma:gram-schmidt}, which also makes use of Lem.~\ref{lemma:overlaps}.
}

\begin{lemma}[Overlap of stochastic Lagrangian sub-spaces]\label{lemma:overlaps}
We have $\obraket{Q_T}{Q_{T'}}\geq 0$ for all $T,T'\in \Sigma_{t,t}$. Moreover, for all $T\in \Sigma_{t,t}$ the sum of overlaps is
\begin{equation}\label{eq:superuseful}
\sum_{T'\in \Sigma_{t,t}} \obraket{Q_T}{Q_{T'}}^n = (-2^{-n}; 2)_{t-1}  \leq 1+ t2^{t-n},
\end{equation}
where $(-2^{-n}; 2)_{t-1} = \prod_{r=0}^{t-2}(1+2^{r-n})$ and the last inequality holds for $n + 2 \geq t + \log_2(t)$.
\end{lemma}

\begin{proof}
	Denote by $\mathrm{Stab}(n)$ the set of stabilizer states on $n$ qubits.
	Since the operators $r(T)$ are entry-wise non-negative, we have $\obraket{Q_T}{Q_{T'}}=2^{-t}\Tr(r(T)^\dagger r(T'))\geq 0$.
	Note that $r(T)^\dagger = r(\tilde T)$ for a suitable $\tilde T\in\Sigma_{t,t}$ (cp.~Thm.~\ref{theorem:r decomposition}).
	We obtain
	\begin{myalign}
	\sum_{T'\in \Sigma_{t,t}} \obraket{Q_T}{Q_{T'}}^n &=\frac{1}{2^{tn}}\sum_{T'\in \Sigma_{t,t}}\Tr\left[r(\tilde T)^{\otimes n} r(T')^{\otimes n}\right]\\
	&\stackrel{\dagger}{=}\frac{2^n\prod_{r=0}^{t-2}(2^r+2^n)}{2^{tn}}\Tr\left[r(\tilde T)^{\otimes n}\mathbb{E}_{s\in \mathrm{Stab(n)}}(\ketbra{s}{s}^{\otimes t})\right]\\
	 &=\frac{2^n\prod_{r=0}^{t-2}(2^r+2^n)}{2^{tn}}\mathbb{E}_{s\in \mathrm{Stab(n)}}\sandwich{s^{\otimes t}}{r(\tilde T)^{\otimes n}}{s^{\otimes t}}\\
	&\stackrel{\ddagger}{=}\frac{2^n\prod_{r=0}^{t-2}(2^r+2^n)}{2^{tn}}\\
    &=\prod_{r=0}^{t-2}(1+2^{r-n}) \\
    &\leq \left( 1+2^{t-2-n}\right)^{t-1} \\
    &\stackrel{*}{\leq} \exp\left((t-1) 2^{t-n-2}\right),
	\end{myalign}
	where we have again used \cite[Thm.~5.3]{gross2017schur} in $\dagger$ and in $\ddagger$ that $\sandwich{s^{\otimes t}}{r(T)^{\otimes n}}{s^{\otimes t}}=1$ for all $T\in \Sigma_{t,t}$ and all $s\in\mathrm{Stab}(n)$ (compare Ref.\  \cite[Eq.~(4.10)]{gross2017schur}). Finally, in $*$ we have used the ``inverse Bernoulli inequality'' $(1+x)^r \leq e^{rx}$ which holds for all $x\in\R$ and $r\geq 0$. By assumption, the following holds
	\begin{myalign}
	 0 \geq t + \log_2(t) - n - 2 \quad \Rightarrow \quad  1 \geq t 2^{t-n-2} \geq (t-1)2^{t-n-2}.
	\end{myalign}
	Thus, we can use the inequality $e^x \leq 1+2x$ for $0\leq x\leq 1$ to obtain
	\begin{myalign}
	 \sum_{T'\in \Sigma_{t,t}} \obraket{Q_T}{Q_{T'}}^n & \leq 1 + (t-1) 2^{t-n-1} \\
	 & \leq 1 + t 2^{t-n}.
	\end{myalign}

\end{proof}

\begin{lemma}
 \label{lemma:Snormbound}
 Let $S_{\Cl}$ be the Clifford frame operator and $\Gamma$ the corresponding Gram matrix, \ie~$\Gamma_{T,T'}=\obraket{Q_T}{Q_T'}^n$. Then the following holds
 \begin{equation}
  \snorm{S_{\Cl}-P_{\Cl}} = \snorm{\Gamma-\one} \leq (-2^{-n};2)_{t-1} - 1 \leq t 2^{t-n},
 \end{equation}
 \newtext{where $(-2^{-n}; 2)_{t-1} = \prod_{r=0}^{t-2}(1+2^{r-n})$ and the last inequality holds for $n + 2 \geq t + \log_2(t)$.}
\end{lemma}

\begin{proof}
 Define the \emph{synthesis operator} of the frame as the map 
 \begin{equation}
  V:\,\C^{|\Sigma_{t,t}|}\rightarrow \Cl(n)', \quad V=\sum_{T\in\Sigma_{t,t}} \oket{Q_T^{\otimes n}} \!\! \bra{e_T},
 \end{equation}
 where $e_T$ is the standard basis of the domain. Then, we have clearly $\Gamma=V^\dagger V$ and $S_{\Cl}|_{\Cl(n)'}=VV^\dagger$. Since $S_{\Cl}$ and $P_{\Cl}$ are both identically zero on $\left(\Cl(n)'\right)^{\perp}$, this part does not contribute to the spectral norm. From this it is clear that  
 \begin{equation}
  \snorm{S_{\Cl}-P_{\Cl}} = \snorm{\Gamma-\one}.
 \end{equation}
 Moreover, we can compute
 \begin{myalign}
  \snorm{\Gamma-\one} &= \snorm{\sum_T \sum_{T,T'} \obraket{Q_T}{Q_{T'}}^n  \ketbra{e_T}{e_{T'}}} \\
  &\leq \max_{T}\sum_{T'\neq T} \obraket{Q_T}{Q_{T'}}^n \\
  &= (-2^{-n};2)_{t-1} - 1,
 \end{myalign}
 where we have used that the spectral norm of Hermitian operators is bounded by the max-column norm and inserted the exact result of Lemma~\ref{lemma:overlaps} in the last step. Finally, said lemma provides the desired bound for $n+2\geq t+\log_2 t$.
\end{proof}

\subsection{Proof of Lemmas for Theorem \ref{thm:Tcountdesign}}
\label{section:haarsymmetrization}
\overlapnonclifford*
The proof of Lemma~\ref{lemma:overlapnonclifford} is based on two results.
The first states that the basis elements $r(T)$ of the commutant of tensor powers of the Clifford group either belong to the commutant of the powers of the unitary group, or else are far away from it.
\begin{restatable}[Haar symmetrization]{lemma}{haarsymmetrization}
\label{lemma:haarsymmetrization}
	For all $t$ and for all $T\in\Sigma_{t,t}\setminus S_t$, it holds that
	\begin{equation}
    \label{eqn:hilbert space bound}
    \osandwich{Q_T}{P_{\Haar}}{Q_T}=2^{-t} \TwoNorm{P_\mathrm{H}[r(T)]}^2 \,\leq\, \frac78,
	\end{equation}
  where $Q_T$ is as in Eq.~(\ref{eqn:psiT}) and $P_{\Haar}=\MO_t(\haar)$ is the
  $t$-th moment operator of the single-qubit unitary group $\U(2)$.
\end{restatable}
The proof is given in Section~\ref{sec:r(T) bounds}.
In Appendix~\ref{sec:converse bounds}, we show that the constant $7/8$ cannot be improved below $7/10$, by exhibiting a $T$ that attains this bound.

The second ingredient to Lemma~\ref{lemma:overlapnonclifford} is a powerful theorem by Varj{\'u}~\cite{varju_walks_2013}.
Here, we specialize this theorem to the unitary group:
\begin{theorem}[{\cite[Thm.~6]{varju_walks_2013}}]
\label{theorem:varju}
	Let $\nu$ be a probability measure on $\U(d)$.
	Consider the averaging operator $T_{v}(\nu)$ on a irreducible representation $\pi_{v}:\,\U(d)\rightarrow \mathrm{End}(W_v)$
  parameterized by highest weight $v\in\Z^d$:
	\begin{equation}
	T_v(\nu) :=\int_{\U(d)} \pi_v(U) \, \d\nu(U).
	\end{equation}
	Then there are numbers $C(d)>0$ and $r_0>0$ such that
	\begin{equation}
	\Delta_r(\nu):=1-\max_{0<|v|\leq r}\snorm{T_v(\nu)} \geq C(d)\Delta_{r_0}(\nu)\log^{-2}(r),
	\end{equation}
	where $|v|^2=\sum_{i}v_i^2$.
\end{theorem}

\begin{proof}[Proof of Lemma~\ref{lemma:overlapnonclifford}]
	Consider the probability measure $\xi_K$ that draws uniformly from the set $\{K, K^{\dagger}, \one\}$.
	Moreover, define $\nu_K$ on $\U(2)$ as the average of the uniform measure on $\{H, S, S^3\}$ 	and $\xi_K*\xi_K$.
	Hence, the according moment operator is
	\begin{myalign}
	\MO_t(\nu_K):=&\frac16 (\Ad_H^{\otimes t}+\Ad_S^{\otimes t}+(\Ad_S^3)^{\otimes t})+\frac12 \MO_{t}(\xi_K*\xi_K)\\
	=& \frac16 (\Ad_H^{\otimes t}+\Ad_S^{\otimes t}+(\Ad_S^3)^{\otimes t})+\frac12 \MO_t(\xi_K)^2.
	\end{myalign}
	As the Clifford group augmented with any non-Clifford gate is universal~\cite[Thm.~6.5]{nebe_clifford_2001}, so is the probability measure $\nu_K$.

	It follows from the representation theory of the unitary group (see App.~\ref{section:irrepunitarygroup})
  that the representation $U\mapsto \Ad_U^{\otimes t}$ does not contain irreducible representations $W_v$ with
  highest weight of length $|v|>\sqrt{2}t$.
	Thus, we can decompose into these irreducible representations as follows:
	\begin{myalign}
	 \snorm{\MO_t(\nu_K) - P_{\Haar}} &= \snorm{ \bigoplus_{|v|\leq \sqrt{2}t} \left( T_v(\nu_K) - T_v(\haar)\right) \otimes \id_{m_v} } \\
	 &\leq \snorm{ \bigoplus_{0<|v|\leq \sqrt{2}t}  T_v(\nu_K) } \\
	 &= \max_{0<|v|\leq \sqrt{2}t} \snorm{T_v(\nu_K)} \\
	 &= 1 - \Delta_{\sqrt{2}t}(\nu_K).
	\end{myalign}
    Here, $m_v$ denotes the multiplicity of the irreducible representation $W_v$ (possibly zero).
    In the second step we have used that $P_{\Haar}$ has only support on the trivial irreducible representation  $v=0$, where both $P_{\Haar}$ and $\MO_t(\nu_K)$ act as identity and thus cancel.
    Hence, only non-trivial irreducible representations are contributing.
    To bound $\Delta_{\sqrt{2}t}(\nu_K)$, we can invoke Theorem~\ref{theorem:varju} combined with the
    fact that for any universal probability measure the restricted gap is non-zero:
    $\Delta_{r}(\nu_K)>0$ for all $r\geq 1$ (compare e.g. Ref.~\cite{harrow_random_2009}).
	Hence, we obtain
	\begin{equation}
	\Delta_{\sqrt{2}t}(\nu_K)\geq C(2)\Delta_{r_0}(\nu_K)\log^{-2}(\sqrt{2}t)\geq \frac14C(2)\Delta_{r_0}(\nu_K)\log^{-2}(t)=:c'(K)\log^{-2}(t) > 0,
	\end{equation}
	where $c(K)>0$.
	Therefore, we have
	\begin{equation}\label{eq:varjuapplied}
	\snorm{\MO_t(\nu_K) - P_{\Haar}}  \leq 1-\Delta_{\sqrt{2}t}(\nu_K)\leq 1-c'(K)\log^{-2}(t)=:\kappa_{t,K},
	\end{equation}
	Furthermore, consider the operator
	\begin{equation}
    X_T:=\frac{(\id-P_{\Haar})Q_T}{\norm{(\id-P_{\Haar})Q_T}_2}.
	\end{equation}
	We obtain
	\begin{myalign}
	\snorm{\MO_t(\nu_K) - P_{\Haar}} &= \max_{\|X\|_2=1}\left| \osandwich{X}{\MO_t(\nu_K) -P_{\Haar}}{X} \right| \\
	&\geq \frac{\left| \osandwich{X_T}{\MO_t(\nu_K) -P_{\Haar}}{X_T} \right|}{\|X_T\|^2_2}\\
	&=\frac{\left| \osandwich{Q_T}{(\id-P_{\Haar})\MO_t(\nu_K) (\id-P_{\Haar})}{Q_T} \right|}{\osandwich{Q_T}{(\id-P_{\Haar})^2}{Q_T}}\\
	&=\frac{|\osandwich{Q_T}{\MO_t(\nu_K)}{Q_T} - \osandwich{Q_T}{P_{\Haar}}{Q_T}|}{1-\osandwich{Q_T}{P_{\Haar}}{Q_T}} \\
	&\geq \frac{\osandwich{Q_T}{\MO_t(\nu_K)}{Q_T} - \osandwich{Q_T}{P_{\Haar}}{Q_T}}{1-\osandwich{Q_T}{P_{\Haar}}{Q_T}}.
	\end{myalign}
	In the fourth step, we again used the properties of the Haar projector as in Eq.~\eqref{eq:relationPHaar}.
	Combining this with~\eqref{eq:varjuapplied} and Lemma~\ref{lemma:haarsymmetrization} we obtain
	\begin{equation}
	\osandwich{Q_T}{\MO_t(\nu_K)}{Q_T} \leq \kappa_{t,K}+(1-\kappa_{t,K})\osandwich{Q_T}{P_{\Haar}}{Q_T} \leq 1-\frac{1}{8}c'(K)\log^{-2}(t).
	\end{equation}
	We can use that $\osandwich{Q_T}{\Ad_S^{\otimes t}}{Q_T}=\osandwich{Q_T}{\Ad_{S^3}^{\otimes t}}{Q_T}=\osandwich{Q_T}{\Ad_H^{\otimes t}}{Q_T}=1$ for all $T\in \Sigma_{t,t}$ because $Q_T=2^{-t/2}r(T)$ commutes with the $t$-th diagonal action of the single-qubit Clifford group (compare~\cite[Lem.~4.5]{gross2017schur}).
	We immediately obtain
	\begin{equation}
	\osandwich{Q_T}{\MO_t(\xi_K)^2}{Q_T}\leq 1-\frac{1}{4}c'(K)\log^{-2}(t).
	\end{equation}
	From the Cauchy-Schwarz inequality, we now get
	\begin{myalign}
	\left|\osandwich{Q_T}{\MO_t(\xi_K)}{Q_{T'}}\right|
	&\leq \sqrt{\osandwich{Q_T}{\MO_t(\xi_K)^2}{Q_T}}\\
	&\leq \sqrt{1-\frac{1}{4}c'(K)\log^{-2}(t)}\\
	&\leq 1-\frac18 c'(K)\log^{-2}(t)\\
    &=: 1- c(K)\log^{-2}(t),
	\end{myalign}
	where we have used that $c'(K)\log^{-2}(t) \leq \Delta_{\sqrt{2}t}(\nu_K)\leq 1$ such that we can use the inequality $\sqrt{1-x}\leq 1-x/2$ for $x\leq 1$. This shows the claimed statement.
\end{proof}

\begin{remark}[Quantum gates with algebraic entries]
	If we restrict to gates $K$ that have only algebraic entries, we can apply the result from Ref.~\cite{bourgain_spectral_2011} and save the additional overhead of $\log^2(t)$ in the scaling.
	This applies to the $T$-gate and for essentially all gates that might be used in practical implementations.
	Here, we have chosen the more general approach.
\end{remark}
\begin{remark}[Implications for quantum information processing]
	Theorem~\ref{theorem:varju} has miscellaneous implications for quantum information processing.
	E.g. we can immediately combine this bound with the local-to-global lemma in Ref.~\cite[Lem.~16]{onorati_mixing_2017} to extend Ref.~\cite[Cor.~7]{brandao_local_2016} to gate sets with non-algebraic entries at the cost of an additional overhead of $\log^2(t)$ in the scaling.
	The bottleneck to loosen the invertibility assumption as well is the local-to-global lemma which only works for Hermitian moment operators (symmetric distributions).
	Work to lessen the assumption of invertibility has been done in Ref.~\cite{mezher2019}.
	Extending this would be an interesting application which we, however, do not pursue in this work.
\end{remark}

\gramschmidt*
\begin{proof}
	The form of~\eqref{eq:gramschmidt} is up to a constant the determinant formulation of the Gram-Schmidt procedure.
	First, note that the number of permutations of $n$ elements with no fixed points is known from Ref.~\cite{montmort_permutation_1713} to be
	\begin{equation}
	D(n)=n!\sum_{r=0}^n\frac{(-1)^r}{r!}\leq 2\frac{n!}{e}
	\end{equation}
	for $n\geq 1$.
	Here, $D$ stands for ``derangement'' as permutations without fixed points are sometimes called.
	Then, the number of permutations having exactly $k$ fixed points is ${n\choose k}$ many choices of $k$ points times the number $D(n-k)$ of deranged permutations on the remaining $n-k$ objects:
	\begin{equation}
	p(n,k):={n\choose k}D(n-k)\leq 2e^{-1}\frac{n!}{k!}.
	\end{equation}
	The following estimate for certain sums involving $p(n,k)$ will shortly become useful.
	 Note that we have for any $M,L\in\N$ and $m\in\R$ such that $2^m>M-L$ and $M\geq L\geq 1$:
	\begin{multline}
	\label{eq:fixed-point-approx}
	 \sum_{k=0}^{M-L} p(M,k) 2^{-m(M-k)} \leq \frac 2 e \sum_{k=0}^{M-L} 2^{-mM} M! \frac{2^{mk}}{k!} \\
	 \leq \frac 2 e  2^{-mM} (M-L+1) M!\frac{2^{m(M-L)}}{(M-L)!} \leq M^{L+1} 2^{-mL}.
	\end{multline}
	Here, we have used in the second inequality that $2^{mk}/k!$ is monotonically increasing for $k\leq M-L < 2^m$ and a standard bound on binomial coefficients in the last step.

	We start by bounding the diagonal coefficients $A_{j,j}$. The idea is to divide the set of permutations into sets of permutations with exactly $k$ fixed points.
	For any such permutation, the product of overlaps collapses to only $j-1-k$ non-trivial inner products. By assumption $n\geq\frac12(t^2+5t)\geq t+\log_2 t$, thus we can be bound any of those using Lemma~\ref{lemma:overlaps} as
    \begin{equation}\label{eq:recalloverlap}
		\obraket{Q_T}{Q_{T'}}^n \leq t2^{t-n}, \quad \text{ for all } T\neq T'.
    \end{equation}
    Note that the trivial permutation (corresponding to $k=j-1$ fixed points) contributes by exactly 1 to the sum. Thus, we find the following bound using Eq.~\eqref{eq:fixed-point-approx} with $M=j-1$, $L=1$ and $m=n-t-\log_2 t$:
    
	\begin{myalign}
	\label{eq:gram-diagonal}
	 A_{j,j}=|A_{j,j}| &\leq \sum_{\pi\in S_{j-1}} \prod_{l=1}^{j-1} \obraket{Q_l}{Q_{\pi(l)}}^n \\
	 &\leq 1 + \sum_{k=0}^{j-2} p(j-1,k) 2^{-(n-t-\log_2 t)(j-1-k)} \\
	 &\leq 1 + (j-1)^2 \, 2^{-n+t+\log_2 t} \\
	 &< 1 + 2^{t^2 + 7t - n},
	\end{myalign}
	where we have used Eq.~\eqref{eq:boundnumberlagrangian} in the last step as $j-1 < j \leq |\Sigma_{t,t}|\leq 2^{\frac 1 2 (t^2+5t)}$. Using the reverse triangle inequality, we get a lower bound in the same way:
	
	\begin{equation}
	 A_{j,j}=|A_{j,j}| 
	 \geq 1 - \left| \sum_{\pi\in S_{j-1}\setminus\id}\sign(\pi) \prod_{l=1}^{j-1} \obraket{Q_l}{Q_{\pi(l)}}^n\right| 
	 \geq 1 -  2^{t^2 + 7t - n}.
	\end{equation}

    Next, we will bound the off-diagonal terms $A_{i,j}$. It is well known that every permutation $\Pi\in S_j$ can be written as a product of disjoint cycles. Given a $\Pi \in S_j$ with $\Pi(j)=i$, consider the cycle $j\mapsto i\mapsto i_1\mapsto i_2\mapsto \dots i_r\mapsto j$ in $\Pi$.
	Then, we have the bound
	\begin{myalign}
	\label{eq:cycle-bound}
	\prod_{l=1}^{j-1}\obraket{Q_{T_l}}{Q_{T_{\Pi(l)}}}^n
	&\leq \obraket{Q_{T_i}}{Q_{T_{i_1}}}^n \dots \obraket{Q_{T_{i_r}}}{Q_{T_j}}^n\\
	&\leq 2^{-n(|\dim N_i-\dim N_{i_1}|+\dots |\dim N_{i_r}-\dim N_{j}|)}\\
	&\leq 2^{-n|\dim N_i-\dim N_{j}|},
	\end{myalign}
	where we have used Lemma~\ref{lemma:diamondbounds}, the triangle inequality and a telescope sum.
    We set $L:=|\dim N_i-\dim N_{j}|$ and split the sum over permutations into those with more than or equal to $j-L$ many fixed points and those with less.
	In the first case, we use Eq.~\eqref{eq:cycle-bound} to bound the overlaps, in the second case we use Eq.~\eqref{eq:fixed-point-approx} as before. This yields the following bound
	\begin{myalign}
     |A_{i,j}|& \leq \sum_{\substack{\Pi\in S_j\\\Pi(j)=i}}\prod_{l=1}^{j-1}\obraket{Q_{T_l}}{Q_{T_{\Pi(l)}}}^n\\
	 &\leq \sum_{k=j-L}^{j-1} p(j,k) 2^{-nL}  +  \sum_{k=0}^{j-L-1} p(j,k) 2^{-(n-t-\log_2t)(j-1-k)}\\
	 &\leq \frac 2 e \sum_{k=j-L}^{j-1}\frac{j!}{k!} 2^{-nL} +  2^{n-t-\log_2t} j^{L+2} \, 2^{-(n-t-\log_2t)(L+1)} \\
	 &\leq L \frac{j!}{(j-L)!} 2^{-nL} + j^{L+2} \, 2^{-(n-t-\log_2t)L} \\
	 &\leq L j^L 2^{-nL} + j^{L+2} \, 2^{-(n-t-\log_2t)L} \\
	 &\leq L |\Sigma_{t,t}|^{L+2} \, 2^{-(n-t-\log_2t)L}  \\
	 &\leq 2^{\log_2 L }2^{\frac 1 2 (t^2+5t)(L+2)} \, 2^{(t +\log_2t-n)L} \\
	 &= 2^{t^2+5t} 2^{(\frac 1 2 t^2 + \frac 5 2 t + t +\log_2t-n)L} \\
	 &\leq 2^{\frac 1 4 t^3 + \frac{11}{4} t^2 + 5t + (\frac t 2 + 1)\log_2 t - nL} \\
	 &\leq 2^{t^3 + 4 t^2 + 6t - n|\dim N_i-\dim N_j|},
	\end{myalign}
	where we have used again $j\leq |\Sigma_{t,t}|$ and $L\leq t/2$.

	Note that we can alternatively bound $A_{i,j}$ for $i \neq j$ using that the identity is not an allowed permutation, \ie~only permutations with less than $j-2$ fixed points can appear. With Eq.~\eqref{eq:fixed-point-approx} and~\eqref{eq:recalloverlap},  we get the following inequality
	\begin{myalign}
	\label{eq:gram-offdiagonal}
      |A_{i,j}|&\leq \sum_{k=0}^{j-2} p(j,k) 2^{-(n-t-\log_2t)(j-1-k)}\\
	 &\leq  j^3 2^{-(n-t-\log_2t)} \\
	 &\leq 2^{\frac 3 2 t^2 + \frac{15}{2} t + t + \log_2 t - n} \\
	 &\leq 2^{2t^2 + 10 t - n}.
	\end{myalign}
\end{proof}

\subsection{Proof of Haar symmetrization Lemma~\ref{lemma:haarsymmetrization}}
\label{sec:r(T) bounds}

\haarsymmetrization*

For an analysis of the tightness of the bound, see Appendix~\ref{sec:converse bounds}.
Recall that
\begin{equation}
 P_\mathrm{H}[A] := \int_{U(2)} U^{\otimes t} A (U^\dagger)^{\otimes t} \newtext{\mathrm{d}}\haar(U).
\end{equation}
Let $P_D$ be the Haar averaging operator, restricted to the diagonal unitaries.
As it averages over a subgroup, $P_D$ is a projection with range a super-set of $P_\mathrm{H}$.
By applying $P_D$ to $r(T)$, we can turn the statement (\ref{eqn:hilbert space bound}) from one involving \emph{Hilbert space} geometry to one about the \emph{discrete} geometry of stochastic Lagrangians.
Indeed,
\begin{myalign}
  2^{-t} \TwoNorm{P_\mathrm{H}[r(T)]}^2
  &=
  2^{-t} \TwoNorm{P_\mathrm{H}[P_D[r(T)]]}^2 \nonumber \\
  &\leq
  2^{-t} \TwoNorm{P_D[r(T)]}^2 \label{eqn:diagonal haar bound} \\
  &=
  2^{-t}  \Big(r(T), P_D[r(T)]\Big) \nonumber \\
  &=
  2^{-t}
  \sum_{(x,y)\in T}
  \sum_{(x',y')\in T}
  \big(
    \ketbra{x}{y},
	P_D[\ketbra{x'}{y'}]
   \big) \nonumber \\
  &=
  2^{-t}
  \sum_{(x,y)\in T}
  \sum_{(x',y')\in T}
  \big(
    \ketbra{x}{y},
	\int_0^{2\pi}
	e^{i 2\phi (h(x') - h(y'))}
	\ketbra{x'}{y'}
	\mathrm{d}\,\phi
   \big) \nonumber \\
   &=
   2^{-t} |\{ (x,y) \in T \,|\, h(x)=h(y) \}|  \nonumber \\
   &=\operatorname{Pr}_{(x,y)}[  h(x) = h(y) ], \nonumber
\end{myalign}
i.e., the overlap is upper-bounded by the probability that a uniformly sampled element $(x,y)$ of $T$ has components of equal Hamming weight.

We will bound the probability in slightly different ways for spaces $T$ with trivial \newtext{(i.e., zero-dimensional)} and non-trivial defect spaces.

\paragraph{Case I: trivial defect sub-spaces}

In this case, $T=\{ (Oy, y) \,|\, y\in\FF_2^t\}$ for some orthogonal stochastic matrix $O$.
The next proposition treats a slightly more general situation.

\begin{proposition}[Hamming bound]\label{prop:hamming bound}
  Let $O\in \mathrm{GL}(\FF_2^t)$.
  Assume $O$ has a column of Hamming weight $r$.
  Then the probability that $O$ preserves the Hamming weight of a vector $y$ chosen
  uniformly at random from $\FF_2^t$ satisfies the bound
  \begin{myalign}\label{eqn:hamming bound}
	\mathrm{Pr}_y[ h(Oy) = h(y) ]
	\leq
	\frac12
	+
	\left\{
	  \begin{array}{ll}
    2^{-(r+1)}
		{{r+1}\choose{(r+1)/2}}\quad
		&r \text{ odd} \\
		0&r\text{ even}.
	  \end{array}
	  \right.
  \end{myalign}
\end{proposition}

The bound in Eq.~(\ref{eqn:hamming bound}) decreases monotonically in $r$.
Orthogonal stochastic matrices $O$ satisfy $r=1\mod 4$, so the smallest non-trivial $r$
that can appear is $r=5$, for which the bound gives $.81$.

The proof idea is as follows:
For each $y\in\FF_2^t$, the two vectors $y, y+e_1$ differ in Hamming weight by $\pm 1$.
But, if $h(e_1)\neq 1$, then $h(Oy)-h(O(y+e_1))$ tends not to be $\pm 1$.
In such cases, $O$ does not preserve weights for \emph{both} $y$ and $y+e_1$.
Applying this observation to randomly chosen vectors, we can show the existence of many vectors for which $O$ changes the Hamming weight.

\begin{proof}[Proof (of Proposition~\ref{prop:hamming bound})]
  Assume without loss of generality that the first $r$ entries of $O e_1$ are $1$, and the remaing $t-r$ entries are $0$.

  Let $y$ be a uniformly distributed random vector on $\FF_2^t$, notice that also $Oy$,
  and $O(y+e_1)$ are uniformly distributed. Using the union bound, we find that
  \begin{myalign}
	\mathrm{Pr}[h(Oy)=h(y)]
	&=
	1- \mathrm{Pr}[h(Oy)\neq h(y)] \nonumber \\
	&=
	1- \frac12\big(
	  \mathrm{Pr}[h(Oy)\neq h(y)]
	  +
	  \mathrm{Pr}[h(Oy+Oe_1)\neq h(y+e_1) ]
	  \big)
	\nonumber \\
	&\leq
	1- \frac12 \mathrm{Pr}[h(Oy)\neq h(y) \,\vee\, h(Oy+Oe_1)\neq h(y+e_1) ]\nonumber  \\
	&=
	\frac12 + \frac12 \mathrm{Pr}[h(Oy)= h(y) \,\wedge\, h(Oy+Oe_1)= h(y+e_1) ]\nonumber  \\
	&\leq
	\frac12 + \frac12 \mathrm{Pr}[ h(Oy)-h(Oy+Oe_1) = \pm 1]. \label{eqn:union bound for hamming weight}
  \end{myalign}
  We would like to compute $\mathrm{Pr}[ h(Oy)-h(O(y+e_1)) = \pm 1]$. The vector $O(y+e_1)=O(y)+O(e_1)$
  arises from $O(y)$ by flipping the first $r$ components.
  This operation changes the Hamming weight by $\pm 1$ if and only if the number of ones
  in the first $r$ components of $O(y)$ equals $(r\pm 1)/2$. For even $r$, this condition
  cannot be met, and correspondingly $\mathrm{Pr}[ h(Oy)-h(O(y+e_1)) = \pm 1]=0$.

  In case of odd $r$, this probability becomes
  \begin{myalign}
	\mathrm{Pr}[ h(Oy)-h(O(y+e_1))
  &=
  \pm 1]= 2^{-r} {{r}\choose{(r-1)/2}}
	+
	2^{-r} {{r}\choose{(r+1)/2}}\\[5 pt]
	&=
	2^{-r}
	{{r+1}\choose{(r+1)/2}}.\label{eq:hamming-binomials}
\end{myalign}

\end{proof}

\paragraph{Case II: non-trivial defect sub-spaces}

We now turn to Lagrangians $T$ with a non-trivial defect subspace.

\begin{proposition}[Defect Hamming bound]\label{prop:defect hamming bound}
  Let $\{0\}\neq N\subset \FF_2^t$ be isotropic. 
  There exists an $n\in N$ such that
  if
  $x$ is
  chosen uniformly at random from $N^\perp$, then
  \begin{align*}
	\mathrm{Pr}_{x\in N^\perp}[ h(x) = h(x+n) ] \leq \frac34.
  \end{align*}
  What is more, let $T$ be a stochastic Lagrangian with non-trivial defect sub-spaces.
  Then, for an element $(x,y)$ drawn uniformly from $T$, we have
  \begin{align*}
   \operatorname{Pr}_{(x,y)\in T}[  h(x) = h(y) ] \leq \frac78.
  \end{align*}
\end{proposition}

\begin{proof}
  Let $d=\dim N$.
  Consider a $t \times d$ column-generator matrix $\Gamma$ for $N$.
  Permuting coordinates of $\FF_2^t$ and adopting a suitable basis, there is no loss of generality in assuming that $\Gamma$ is of the form
  \begin{align*}
	\Gamma=
	\begin{pmatrix}
	  G \\
	  \one_{d}
	\end{pmatrix},
	\qquad
	G \in \F_2^{(t-d)\times d}.
  \end{align*}
  Note that
  \begin{align*}
	\gamma=
	\begin{pmatrix}
	  \one_{t-d},&
	   G
	\end{pmatrix}
  \end{align*}
  is a row-generator matrix for $N^\perp$.
  Indeed, the row-span has dimenion $t-d$ and the matrices fulfill
  \begin{align*}
	\gamma \Gamma = G + G  = 0,
  \end{align*}
  i.e., the inner product between any column of $\Gamma$ and any row of $\gamma$ vanishes.
  It follows that elements $n\in N$, $x\in N^\perp$ are exactly the vectors of respective form
  \begin{align*}
	n=(\underbrace{G \tilde{n}}_{t-d},\, \underbrace{\tilde{n}}_d),\; \tilde{n}\in\FF_2^d;
	\qquad
	x=(\underbrace{\tilde{x}}_{t-d},\, \underbrace{G^T \tilde{x}}_d),\; \tilde{x}\in\FF_2^{t-d}.
  \end{align*}
  In particular, if $x$ is drawn uniformly from $N^\perp$, then the first $t-d$ components are uniformly distributed in $\FF_2^{t-d}$.
  For now, we restrict to the case where $G$ has a column, say the first, with $r\neq 1$ non-zero entries.
  We then choose $n=(Ge_1, e_1)$ and argue as in Eq.~(\ref{eq:hamming-binomials}) to obtain
  \begin{equation}
	\mathrm{Pr}_{x\in N^\perp}[ h(x) = h(x+n)] \leq
	\sup_{1\neq r\text{ odd}}
	2^{-r}
	{{r+1}\choose{(r+1)/2}}
	=\frac34
	\qquad\text{(attained for $r=3$)}.
  \end{equation}

  We are left with the case where all columns of $G$ have Hamming weight $1$.
  (If $N$ is a defect subspace, then Def.~\ref{def:stochastic lagrangians}.\ref{def:q isotropicity} implies that every column of $\Gamma$ has Hamming weight at least $4$.
  We treat the present case merely for completeness).
  As $N$ is isotropic, the columns of $\Gamma$ have mutual inner product equal to $0$:
  \begin{align*}
	\Gamma^T \Gamma  = 0
	\qquad\Leftrightarrow\qquad
	G^TG=-\one=\one \mod 2.
  \end{align*}
  It follows that all columns have to be mutually orthogonal standard basis vectors $e_i\in\F_2^{t-d}$.
  Thus, by permutating the first $t-d$ coordinates of $\FF_2^t$, we can assume that $G$ is of the form
  \begin{equation*}
   G = \begin{pmatrix} \one_d \\ 0 \end{pmatrix}, \quad \Rightarrow \quad
	N =\{ (\tilde{n}\oplus 0_{t-2d}, \tilde{n}) \,|\, \tilde{n}\in\FF_2^{d}\}, \quad
	N^\perp =\{ (\tilde{x}, \tilde{x}|_d) \,|\, \tilde{x}\in\FF_2^{t-d}\},
  \end{equation*}
  where $\tilde x|_d$ denotes the restriction of $\tilde x$ to the first $d$ components.
  Adding $n:=(e_1\oplus 0, e_1)$ to $x=(\tilde{x}, \tilde{x}|_d)$, the Hamming weight of the two parts change both by $\pm 1$, giving $h(x+n)=h(x)\pm 2$. Thus, we have $\mathrm{Pr}[ h(x) = h(x+n)]=0$.

  We have proven the first advertised claim.
  It implies the second one, as argued next.
  Let $N$ be the left defect subspace of $T$.
  By Ref.~\cite[Prop.~4.17]{gross2017schur}, we find 
  the following.
  \begin{itemize}
	\item
    The restriction $\{x \,|\, (x,y) \in T \text{ for some }y\}$ equals  $N^\perp$.
	\item
	The stochastic Lagrangian $T$ contains $N\oplus 0$.
  \end{itemize}
  Assume that $(x,y)$ is distributed uniformly in $T$.
  By the first cited fact, $x$ is distributed uniformly in $N^\perp$.
  By the second fact, $(x+n,y)$ follows the same distribution as $(x,y)$, for each $n\in N$.
  Thus, repeating the argument in the proof of Proposition~\ref{prop:hamming bound}, we find that for any fixed $n\in N$:
  \begin{align*}
	\mathrm{Pr}[ h(x) = h(y)]
	&=
	1- \mathrm{Pr}[h(x)\neq h(y)] \\
	&\leq
	1- \frac12 \mathrm{Pr}[h(x)\neq h(y) \,\vee\, h(x+n)\neq h(y) ] \\
	&\leq
	\frac12 + \frac12 \mathrm{Pr}[ h(x)=h(x+n) ]
	\leq \frac 78.
  \end{align*}
\end{proof}

\subsection{Proof of Lemmas for Theorem \ref{theorem:clifford-design}}
\label{section:clifford-lemmas}

\reductionspectral*
\begin{proof}
	This follows similar to Ref.~\cite[Lem.~4\& Lem.~30]{brandao_local_2016}.
	Denote by $|\Omega_{2^n}\rangle$ the maximally entangled state vector on $\C^{2^n}\otimes\C^{2^n}$.
	The condition in~\eqref{eq:relativeclifforddesign} is equivalent to
	\begin{equation}
	(1-\varepsilon)\rho_{\rm Cl}\leq \rho_{\nu}\leq (1+\varepsilon)\rho_{\rm Cl},
	\end{equation}
	as an operator inequality, where
	\begin{equation}
	\rho_{\nu}:=(\Delta_{\nu}\otimes \one)(|\Omega_{2^n}\rangle\langle\Omega_{2^n}|)^{\otimes t}\qquad\text{and}\qquad \rho_{\rm Cl}:=\rho_{\mu_{\rm Cl}}.
	\end{equation}
	We have a decomposition of $(\C^{2^n})^{\otimes t}$ into irreducible representations of the Clifford group:
	\begin{equation}
	(\C^{2^n})^{\otimes t}\cong \bigoplus_{\gamma}C_{\gamma}\otimes L_{\gamma},
	\end{equation}
	where $\{C_\gamma\}$ is the set of all equivalence classes of irreducible representations 
  of $\Cl(n)$ that appear in the $t$-th order diagonal representation, and $L_\gamma$ are the
  corresponding multiplicity spaces (which by the double commutant theorem are irreducible representations  of
  the commutant algebra --we have chosen $L$ for Lagrangian).
	This implies that
	\begin{equation}
	|\Omega_{2^n}\rangle^{\otimes t}\cong \sum_{\substack{\gamma}}\sqrt{\frac{\dim L_\gamma\dim C_{\gamma}}{2^{nt}}} |\gamma,\gamma\rangle\otimes |\Omega_{C_{\gamma}}\rangle\otimes |\Omega_{L_{\gamma}}\rangle,
	\end{equation}
	where $|\Omega_{L_{\gamma}}\rangle$ and $|\Omega_{C_{\gamma}}\rangle$ denote maximally entangling state vectors on two copies of $L_{\gamma}$ and $C_{\gamma}$, respectively.
	Indeed, observe that $\ket{\Omega_{2^n}}^{\otimes t}=2^{-nt/2}\vecmap(\one)$ and that the identity restricted to sub-spaces is just the identity on these sub-spaces.
	The prefactors then follow from normalizing the vectorized identity operators on the direct summands.

	Since $\Cl(n)$ acts via multiplication on the spaces $C_{\lambda}$, this implies that
	\begin{myalign}
	\rho_{\rm Cl}&=\int_{\Cl(n)}(U\otimes \one)^{\otimes t}(|\Omega_{2^n}\rangle\langle\Omega_{2^n}|)^{\otimes t}(U^{\dagger}\otimes \one)^{\otimes t}\newtext{\mathrm{d}\mu_{\Cl}(U)}\\
	& \cong \sum_{\gamma}\frac{\dim L_{\gamma}\dim C{\gamma}}{2^{nt}} (|\gamma\rangle\langle\gamma|)^{\otimes 2}\otimes \left(\frac{\one_{C{\gamma}}}{\dim C{\gamma}}\right)^{\otimes 2}\otimes |\Omega_{L_{\gamma}}\rangle\langle \Omega_{L_{\gamma}}|,
	\end{myalign}
 where the second line follows from Schur's lemma and the fact that $\int U^{\otimes t} \bullet (U^{\dagger})^{\otimes t}$ is trace preserving.
  The support of this operator is on the \emph{symmetric subspace} $\vee^t(\C^{2^n}\otimes\C^{2^n})$~\cite[Lem~30.1]{brandao_local_2016}.
	The minimal eigenvalue of this operator restricted to the symmetric subspace is
	\begin{equation}
	\min_{\gamma}\frac{\dim L_{\gamma}}{2^{nt}\dim C_{\gamma}},
	\end{equation}
  which we now lower bound. Let $\gamma^*$ denote the optimizer.
   By Schur-Weyl duality, the
  diagonal action of $\U(2^n)$ on $(\C^{2^n}\otimes\C^{2^n})^{\otimes t}$ decomposes as $\oplus_\lambda U_\lambda\otimes S_\lambda$
  where as usual $U_\lambda$ are Weyl modules and $S_\lambda$ are Specht modules. Restricting
  this action to the Clifford group, the $U_\lambda$ further decompose into irreducible representations 
  \begin{align*}
    U_\lambda \simeq \bigoplus_{\gamma \in I_\lambda} C_\gamma\otimes\C^{d_{\lambda,\gamma}},
  \end{align*}
  where $I_\lambda$ is the spectrum of $U_\lambda$ as a Clifford representation.
  Let $\Lambda_0$ be the set of all $\lambda$ such that $\gamma^*\in I_\lambda$,
  then as a Clifford representation
  \begin{myalign}
    (\C^{2^n}\otimes\C^{2^n})^{\otimes t} \simeq  C_{\gamma^*}\otimes\Big(\bigoplus_{\lambda \in \Lambda_0} S_\lambda\otimes\C^{d_{\lambda,\gamma^*}}\Big) \oplus (\text{other irreducible representations }).
  \end{myalign}
  Thus, as a vector space, we have
  \begin{myalign}
    L_{\gamma^*} = \bigoplus_{\lambda \in \Lambda_0} S_\lambda\otimes\C^{d_{\lambda,\gamma^*}}.
  \end{myalign}
  In particular, for any $\lambda \in \Lambda_0$ we have that $\dim C_{\gamma^*} \leq \dim U_\lambda$
  and $\dim L_{\gamma^*} \geq \dim S_\lambda$. Thus we get the following bound for the
  minimal eigenvalue:
  \begin{equation}
    \frac{\dim L_{\gamma^*}}{2^{nt}\dim C_{\gamma^*}}\geq \min_{\lambda\in\mathrm{Part}(t,2^n)}\frac{\dim S_{\lambda}}{2^{nt}\dim U_{\lambda}}\geq 2^{-2nt}.
  \end{equation}
  	The rest of the proof follows as in Ref.~\cite[Lem.~4]{brandao_local_2016}, mutatis mutandis.
\end{proof}

In order to prove Lemma~\ref{lemma:spectralgapbound} we make use of the following result by \textcite{nachtergaele_gap_1994} and 
Lemma~\ref{lemma:overlaps} bounding certain sums of overlaps of the operators $r(T)$.

\begin{lemma}[Nachtergaele {\cite[Thm.~3]{nachtergaele_gap_1994}}]
\label{lemma:nachtergaele}
	Let $H_{[p,q]}$ for $[p,q]\subset [n]=\{1,\dots,n\}$ be a family of positive semi-definite Hamiltonians with support on $(\mathbb{C}^2)^{\otimes (q-p+1)}\subset (\C^2)^{\otimes n}$. Assume there is a constant $l\in\N$, such that the following conditions hold:
	\begin{enumerate}
		\item There is a constant $d_l>0$ for which the Hamiltonians satisfy
		\begin{equation}\label{eq:firstcondition}
		0\leq \sum_{q=l}^n H_{[q-l+1,q]} \leq d_l H_{[1,n]}.
		\end{equation}
		\item There are $Q_l\in\N$ and $\gamma_l>0$ such that there is a local spectral gap:
		\begin{equation}
		\Delta\left(H_{[q-l+1,q]}\right)\geq \gamma_l, \quad \forall q \geq Q_l.
		\end{equation}
		\item Denote the ground state projector of $H_{[p,q]}$ by $G_{[p,q]}$. There exist $\varepsilon_l<1/\sqrt{l}$ such that
		\begin{equation}\label{eq:condition3}
		\snorm{G_{[q-l+2,q+1]}\left(G_{[1,q]}-G_{[1,q+1]}\right)} \leq \varepsilon_l, \quad \forall q \geq Q_l.
		\end{equation}
	\end{enumerate}
	Then, it holds that
	\begin{equation}
	\Delta\left(H_{[1,n]}\right)\geq \frac{\gamma_{l}}{d_{l}}\left(1-\varepsilon_l\sqrt{l}\right)^2.
	\end{equation}
\end{lemma}

While conditions 1) and 2) are merely translation-invariance with finit range of interactions and frustration-freeness in disguise, the third condition is highly non-trivial and involves knowledge of the ground-space structure.
Usually, finding the ground space in a basis can be just as hard as computing the spectral gap in the first place.
Fortunately, the ground space structure of the Hamiltonians $H_{n,t}$ is determined by the representation theory of the Clifford group.
With little additional work, we obtain the following lemma about
the ground space structure of our Hamiltonians.

\spectralgapbound*
\begin{proof}
    We make use of the Nachtergaele lemma.
	We have to verify the three conditions of Lemma~\ref{lemma:nachtergaele}. As already stated in Ref.~\cite{nachtergaele_gap_1994}, the first two conditions hold directly for translation-invariant local Hamiltonians as in our case.
	\begin{enumerate}
		\item The first condition immediately follows from the fact that we consider a translation-invariant $2$-local Hamiltonian. It is fulfilled for any choice of $l\geq 2$ and $d_l=l-1$.

		\item The second condition follows again for all $l\geq 2$ and the choice $Q_l=l$, since $H_{[q-l+1,q]}$ is a sum of positive semi-definite operators for all $q\geq l$ with spectrum that does not depend on $q$ due to translation-invariance. Thus, we can set
		\begin{equation}
		\gamma_l:= \Delta(H_{[q-l+1,q]}) > 0.
		\end{equation}
		\item The third condition requires a calculation and a non-trivial choice of $l$.
		We have to bound the quantity
		\begin{equation}
		R_{q,l} := \left\|G_{[q-l+2,q+1]}\left(G_{[1,q]}-G_{[1,q+1]}\right)\right\|_{\infty},
		\end{equation}
		for all $q\geq Q_l=l$. Here, $G_{[p,q]}$ denotes the orthogonal projector onto the ground space of $H_{[p,q]}$. Note that this ground space is simply a suitable translation of the Clifford commutant $\Cl(k)'$ for $k=q-p+1$ as shown in Lemma~\ref{lemma:ground-space}. Recall that it comes with a non-orthogonal basis $Q_T^{\otimes k}$, where
		\begin{equation}
		 Q_T := \frac{r(T)}{\|r(T)\|_2} = 2^{-t/2} r(T), \quad T \in \Sigma_{t,t}.
		\end{equation}
		Moreover, the projector $G_[p,q]$ is also simply a translation of the Clifford projector $P_{\Cl(k)}$ projecting onto $\Cl(k)'$. From the discussion in Section~\ref{sec:overlaps}, we know that the Clifford frame operator
		\begin{equation}
		S_{\Cl(k)}:= \sum_T \oketbra{Q_T}{Q_T}^{\otimes k},
		\end{equation}
		is a suitable approximation to $P_{\Cl(k)}$ when $k$ is large enough.
		Concretely, we have by Lem.~\ref{lemma:Snormbound}:
		\begin{equation}
		 \snorm{S_{\Cl(k)}-P_{\Cl(k)}} \leq (-2^{-k};2)_{t-1} - 1.
		\end{equation}
		Defining the shorthand notation $s_t(k)=(-2^{-k};2)_{t-1}$, we in particular get the bound
        \begin{equation}
            \snorm{S_{\Cl(k)}} \leq \snorm{S_{\Cl(k)}-P_{\Cl(k)}} + \snorm{S_{\Cl(k)}} \leq s_t(k),
        \end{equation}
        Let us introduce the shorthand notation $G_q:=G_{[1,q]}\equiv P_{\Cl(q)}$, $S_q = S_{[1,q]}\equiv S_{\Cl(q)}$, and $G_{q,l}:=G_{[q-l+2,q+1]}$, $S_{q,l}:=S_{[q-l+2,q+1]}$ for translations of the Clifford projector and frame operator, respectively.
        \newtext{Notice that $G_q-G_{q+1}$ is an orthogonal projector as the support of $G_{q+1}$ is by definition contained in that of $G_q$.
        Therefore, restricted to the support of $G_q$, the operator $G_q-G_{q+1}$ projects onto the orthogonal complement of the support of $G_{q+1}$.
        Combining this fact with the above inequalities}, we find
		\begin{myalign}
		R_{q,l} &= \snorm{G_{q,l}\left(G_{q}-G_{q+1}\right)}\\
		&\leq \snorm{(G_{q,l}-S_{q,l})(G_q-G_{q+1})} + \snorm{S_{q,l}(G_q-G_{q+1})} \\
		&\leq s_t(l) - 1 + \snorm{S_{q,l}(S_q-S_{q+1})} + \snorm{S_{q,l}(G_q-S_q)} + \snorm{S_{q,l}(G_{q+1}-S_{q+1})} \\
		&\leq \snorm{S_{q,l}(S_q-S_{q+1})} + s_t(l) - 1 + s_t(l)\left( s_t(q) + s_t(q+1) -2 \right) \\
		&\stackrel{q\geq l}{\leq} \snorm{S_{q,l}(S_q-S_{q+1})} + \left( s_t(l) - 1 \right) \left( 2 s_t(l) + 1 \right) \\
		&= \snorm {\sum_{T\in\Sigma_{t,t}} \oketbra{Q_T}{Q_T}^{\otimes (q-l+1)} \otimes Y_T} + \left( s_t(l) - 1 \right) \left( 2 s_t(l) + 1 \right),
		\end{myalign}
		where the operator $Y_T$ can be straightforwardly computed as
		\begin{myalign}
		Y_T := \sum_{T'\neq T}\left( \obraket{Q_{T'}}{Q_{T}}^{l-1} \oketbra{Q_{T'}}{Q_{T}}^{\otimes (l-1)}\right)\otimes \Big( \oketbra{Q_{T'}}{Q_{T'}} \big(\id-\oketbra{Q_{T}}{Q_{T}}\big) \Big).
		\end{myalign}
		Invoking the synthesis operators
		\begin{equation}
		 V_{k} = \sum_T \oket{Q_T^{\otimes k}} \!\! \bra{e_T}:\; \C^{|\Sigma_{t,t}|}\longrightarrow \Cl(k)',
		\end{equation}
		introduced in Lemma~\ref{lemma:Snormbound}, one can bound the above norm as
		\begin{myalign}
		  \snorm {\sum_{T} \oketbra{Q_T}{Q_T}^{\otimes (q-l+1)} \otimes Y_T} & =  \snorm {\sum_{T} V_{q-l+1}\ketbra{e_T}{e_T}V_{q-l+1}^\dagger \otimes Y_T} \\
		  &\leq \snorm{V_{q-l+1}V^\dagger_{q-l+1}} \snorm {\sum_{T} \ketbra{e_T}{e_T} \otimes Y_T} \\
		  &= \snorm{S_{q-l+1}} \max_T \snorm{Y_T} \\
		  &\leq s_t(q-l+1) \left( s_t(l-1) - 1 \right).
		\end{myalign}
		Thus, we arrive at
		\begin{myalign}
		R_{q,l} &\leq s_t(q-l+1) \left( s_t(l-1) - 1 \right) + \left( s_t(l) - 1 \right) \left( 2 s_t(l) + 1 \right) \\
		 &\leq s_t(1) \left( s_t(l-1) - 1 \right) + \left( s_t(l) - 1 \right) \left( 2 s_t(l) + 1 \right).
		\end{myalign}
		For $l+1\geq t+\log_2(t)$, we can use Lemma~\ref{lemma:overlaps} to get:
		\begin{myalign}
		R_{q,l} &\leq t 2^{t-l+1} \left( 1 + t 2^{t-1} \right) + t 2^{t-l} \left( 3 + t 2^{t-l}\right) \\
		&= t^2 2^{2t-l}\left( \frac{5}{t} 2^{-t} + 2^{-l} + 1 \right) \\
		&\leq 4 t^2 2^{2t-l}.
		\end{myalign}
		Finally choose any $l \geq 4t + 4\log_2(t) + 6$, then we find
		\begin{equation}
		 l \leq \frac{4^{l-2t}}{64t^2} \quad \Rightarrow \quad R_{q,l} \leq 4 t^2 2^{2t-l} \leq \frac{1}{2\sqrt{l}} < \frac{1}{\sqrt{l}}, \quad \forall q \geq l.
		\end{equation}
		In particular, we can choose $l = 12t$, $\varepsilon_l = 1/2\sqrt{l}$ to get the desired bound in Lemma~\ref{lemma:nachtergaele} $\forall q \geq l$.
	\end{enumerate}

	Hence, for the choices $l = 12t$, $d_l=l-1$, $Q_l=l$, $\gamma_l = \Delta(H_{12t,t})$ and $\varepsilon_l = 1/2\sqrt{l}$, Lemma \ref{lemma:nachtergaele} gives the claimed bound on the spectral gap:
	\begin{equation}
	 \Delta(H_{n,t}) \geq \frac{\gamma_l}{d_l}\left(1-\varepsilon_l^2 \sqrt{l}\right) \geq \frac{\Delta(H_{12t,t})}{48t}.
	\end{equation}
\end{proof}

\section{Summary and open questions}
We have found that a number of non-Clifford gates independent of the system size suffices to generate $\varepsilon$-approximate unitary $t$-designs. 
This is surprising, conceptually interesting and practically
relevant: After all, it is the main objective in quantum gate synthesis to minimize
the number of non-Clifford gates in a circuit implementation of a given unitary.
There are multiple open questions and ways to continue this work:
\begin{itemize}
	\item Similar to the result in Ref.~\cite{brandao_local_2016},
	the scaling in $n$ is near to optimal, the scaling in $t$ can probably be improved.
	\item Another natural open question is whether the condition $n=O(t^2)$ can be lifted.
	Notably, this is reminiscent to the situation
	discussed in Ref.~\cite{nakata_efficient_2017}, where the improved scaling can be proven only in the regime $t=o(n^{\frac12})$.
		\newtext{In this work, the condition $n=O(t^2)$ is related to the approximate orthogonality of the Lagrangian subspace. We use this fact repeatedly and in different flavours, but we can only prove it in this regime. 
	In fact, in Lemma~\ref{lemma:Snormbound} we use the same technique that has been used in Ref.~\cite{brandao_local_2016} to prove approximate orthogonality of permutations in the regimes $t\leq 2^{O(0.4n)}$. 
	However, the commutant of the Clifford group is far larger than the span of permutations and we suspect that this bound is tight.
	Nevertheless, we cannot rule out that similar results can be proven without exploiting  approximate orthogonality.
	This likely requires a detailed understanding of the representation theory of the Clifford group.}
	\item Our result holds for additive errors in the diamond norm.
\newtext{For relative errors,} our bounds can be used to obtain a quadratic advantage in the number of non-Clifford gates in Corollary~\ref{cor:dropindependence}.
This still allows the density of non-Clifford gates to go to zero in the thermodynamic limit\newtext{, but is not system-size independent anymore.
In fact, it has been proven in Ref.~\cite{leone2021quantum} that this scaling is optimal for relative errors.
It would be interesting to delineate more precisely for which notions of approximations a system-size independent result holds.}
	\item We strongly expect that the results can be generalized to qu$d$its for arbitrary $d$,
	\newtext{giving rise to analogous conclusions concerning an independence of the system size for additive errors in the diamond norm}.
\end{itemize}
We hope the present work stimulates such endeavors.

\section{Acknowledgements}
We would like to thank Richard Kueng, Lorenz Mayer and Adam Sawicki for helpful discussions. \newtext{Moreover, we would like to thank Nick Hunter-Jones for pointing out the application presented in Appendix~\ref{section:renyi}.}
Funded by the Deutsche Forschungsgemeinschaft (DFG, German Research Foundation) under Germany's Excellence Strategy - Cluster of Excellence Matter and Light for 
Quantum Computing (ML4Q) 
EXC 2004/1 - 390534769, the ARO under contract W911NF-14-1-0098 (Quantum Characterization, Verification, 
and Validation), and the DFG (SPP1798 CoSIP, project B01 of CRC 183). The Berlin group has been supported in this work by the DFG (SPP1798 CoSIP, projects B01 and A03 of CRC 183, 
\newtext{FOR 2724} and EI 519/14-1), \newtext{the Einstein Research Foundation (Einstein Research Unit on quantum devices)} and the Templeton Foundation.
 This  work  has  also  received  funding  from  the  European  Union's  Horizon2020  research  and innovation  programme  under  grant  agreement No.~817482 (PASQuanS).
\section{Data Availability Statement}
No data was produced in this project.

\appendix
\section{Unitary $t$-designs}
\label{app:tdesigns}

In the following, we review the concept of a \emph{unitary $t$-design} \cite{dankert_exact_2009,DankertThesis,GroAudEis},
giving different but equivalent definitions which prove to be useful in different contexts. They also serve as starting point to explore connections to other mathematical fields, \eg~representation theory.
To this end, let us introduce some notation. Define $\haar$ to be the (normalized) Haar measure on $\U(d)$ and let $\Hom_{(t,t)}(\U(d))$ be the space of homogeneous polynomials of degree $t$ in both the entries of $U\in\U(d)$ as well as $ \overline{U}$.

\begin{definition}[Unitary $t$-design]
\label{def:tdesign-poly}
 A probability measure $\nu$ on $\U(d)$ is called a \emph{unitary $t$-design} if the following holds for all $p\in\Hom_{(t,t)}(\U(d))$:
 \begin{equation}
 \label{eq:tdesign-poly}
  \int_{\U(d)} p(U)\,\nu(U) = \int_{\U(d)} p(U)\,\haar(U).
 \end{equation}
 A subset $D\subseteq \U(d)$ is called a
 \emph{unitary $t$-design}, if it comes with a probability measure $\nu_D$ which, continued trivially to $\U(d)$, is a unitary $t$-design. In particular, if $D$ is finite, $\nu_D$ is usually taken to be the (normalized) counting measure.
\end{definition}

It might not come as a surprise that Def.~\ref{def:tdesign-poly} has not to be checked for any polynomial. Since any homogeneous polynomial $p\in\Hom_{(t,t)}(\U(d))$ can be linearized as
\begin{equation}
 p(U) = \tr\left(A U^{\otimes t,t}\right), \qquad U^{\otimes t,t}:=U^{\otimes t}\otimes \overline{U}^{\otimes t},
\end{equation}
the defining Eq.~\eqref{eq:tdesign-poly} becomes
\begin{equation}
\label{eq:tdesign-moments}
 M_t(\nu):=\int_{\U(d)} U^{\otimes t,t}\,\nu(U) = \int_{\U(d)} U^{\otimes t,t}\,\haar(U) =: M_t(\haar).
\end{equation}
Thus $\nu$ is a unitary $t$-design if and only if its moment operator $M_t(\nu)$ agrees with the one of the Haar measure. Note that the operators $U^{\otimes t,t}$ are the matrix representation of the $t$-diagonal adjoint action $\Ad(U^{\otimes t}) = U^{\otimes t}\bullet (U^\dagger)^{\otimes t}$ with respect to the standard basis $\ketbra{i}{j}$ of $L(\C^d)$. Thus, this can be equivalently stated as equality of the twirls $\MO_t(\nu)=\MO_t(\haar)$ over the two measures.

A particularly fruitful theory of designs is possible in the case where the design $(G,\nu)$ itself constitutes a (locally compact) subgroup $G\subseteq\U(d)$ and $\nu$ is the normalized Haar measure on $G$. Following Ref.~\cite{bannai_unitary_2020}, we call these \emph{unitary $t$-groups}. In this case, we see that Eq.~\eqref{eq:tdesign-moments} implies that the trivial isotype of the representation $G \ni g \mapsto \Ad_g^{\otimes t}$ shall agree with the trivial isotype of $\U(d) \ni U \mapsto \Ad_U^{\otimes t}$. Since the trivial isotype exactly corresponds to the commutant of the respective diagonal representations $\tau_t: U \mapsto U^{\otimes t}$, this is equivalent to the statement that the commutant of the representation $\tau_t$ agrees with the commutant of the restriction $\tau_t|_G$. However, this is the case if and only if $\tau_t|_G$ decomposes into the same irreducible representations as $\tau_t$.

\section{Representations of the unitary group}
\label{section:irrepunitarygroup}

The representation theory of the unitary group can be understood using the theory of highest weight for compact Lie groups, see, for example Refs.~\cite{brocker_representations_1985,fulton_representation_2004,goodman_symmetry_2009}. We present a short summary of the part relevant to us here.
Let $\rho$ be an irreducible representation of $\U(d)$, and consider the restriction $\rho|_{D(d)}$ to the diagonal
subgroup $D(d)\simeq (S^1)^{\times d}$ (which is a so-called \emph{maximal torus} in $\U(d)$). In general, this is a reducible representation of $D(d)$. Since $D(d)$ is Abelian, $\rho|_{D(d)}$ decomposes
into one-dimensional irreducible representations , \ie~characters of $D(d)\simeq(S^1)^{\times d}$. 
Those are of the form $\chi_u(\theta) := e^{i u^T \theta}$ for some vector $u\in \Z^d$, and thus we find
\begin{myalign}
  \rho|_{D(d)} \simeq\bigoplus_{u\in\Z^d} \chi_u\otimes \one_{m_u},
\end{myalign}
where $m_u \in \mathbb{N}$ are multiplicities. The vectors $u$ for which $m_u\neq 0$ are called the \emph{weights} of $\rho$.
Introducing a lexicographical ordering of the weights, we call a weight $u$ higher than the weight $v$ if $u>v$. The \emph{theorem of the highest weight} states that any irreducible representation $\rho$ has a highest weight and that irreducible representations  with the same highest weight are isomorphic. Thus, irreducible representations are unambiguously labeled by their highest weight. Next, let us consider the tensor product $\pi_u\otimes\pi_v$ of two irreducible representations  labeled by their highest weights $u$ and $v$. One can easily check that the weights of irreducible representations  in $\pi_u\otimes\pi_v$ have to be sums of weights of $\pi_u$ and $\pi_v$. In particular, the highest weight of all irreducible representations  is at most $u+v$.

As a relevant example consider the (irreducible) defining representation $\rho: U \mapsto U$ of $U(2)$. Its restriction to the diagonal subgroup $S^1\times S^1$ decomposes as
\begin{equation*}
 \rho|_{S^1\times S^1} \simeq \chi_{e_1}\oplus \chi_{e_2},
\end{equation*}
with highest weight $e_1=(1,0)$. Using $\bar{\chi}_u = \chi_{-u}$, the highest weight of the complex conjugate representation $\bar{\rho}:\,U\mapsto\bar U$ can be immediately determined as $(0,-1)$.
Hence, the weights of $\rho\otimes\bar\rho$ are $\{(0,0), (1,-1), (-1,1)\}$. Here, $(0,0)$ is the highest weight of the trivial irreducible representation and $(1,-1)$ the highest weight of the adjoint irrep.
Finally, all irreducible representations  appearing in $(\rho\otimes\bar\rho)^{\otimes t}$ have weights $w$ satisfying
$(-t,t)\leq w\leq (t,-t)$ and, in particular,
\begin{align*}
  w = \sum_{i=1}^tu_i
\end{align*}
where $u_i \in \{(0,0), (1,-1), (-1,1)\}$.
It follows that the Euclidean norm of these weights is at most $\sqrt{2}t$.

\section{Converse bounds for estimates in Section \ref{sec:r(T) bounds}}
\label{sec:converse bounds}

Here, we collect various tightness results that limit the degree by which the estimates in Section~\ref{sec:r(T) bounds} can be improved.
The bound in Proposition~\ref{prop:hamming bound} is tight in many cases.
Most interestingly, the anti-identity \cite{gross2017schur}
\begin{myalign}\label{eqn:anti-id}
  \overline\one=
  \begin{pmatrix}
    0 & 1 & \cdots & 1\\
    1 & \ddots & \ddots & \vdots \\
    \vdots & \ddots & \ddots & 1\\
	1 &  \cdots & 1 & 0
  \end{pmatrix}
  \in O_t\, ,
\end{myalign}
meets the bound if both
\begin{myalign}
  r=t-1 \qquad\text{and}\qquad t/2=(r+1)/2 \qquad\text{ are odd}.
  \label{eqn:tightness condition}
\end{myalign}
Indeed,
the anti-identity flips the components of the input if its parity is odd, and leaves the
input invariant if the parity is even. The flipping step preserves the Hamming weight if
and only if $h(a)=t/2$.
Thus
\begin{align*}
  \mathrm{Pr}[ h(Oa) = h(a) ]
  &=
  \mathrm{Pr}[ h(a) \text{ even}]
  +
  \mathrm{Pr}[ h(a) \text{ odd} \wedge h(a)=t/2] \\
  &=
  \mathrm{Pr}[ h(a) \text{ even}]
  +
  \mathrm{Pr}[h(a)=t/2]
  && \text{(using (\ref{eqn:tightness condition}))}
  \\
  &=
  \frac12
  +
  2^{-t}
  {{t}\choose{t/2}}\\
  &= \frac12 + 2^{-(r+1)} {{r+1}\choose{(r+1)/2}}.
\end{align*}
Likewise, both estimates in Proposition~\ref{prop:defect hamming bound} are tight.
The first bound is saturated for $N=\{0, (1,1,1,1)\}$.
Indeed, $N^\perp$ is the space of all even-weight elements of $\FF_2^4$.
The only non-trivial element of $N$ is $(1,1,1,1)$ and adding it to an even-weight vector changes its weight if and only if the vector is in $N$ itself.
But $|N|/|N^\perp|=1/4$.
In an exactly analogous way, the second bound is tight for the stochastic Lagrangian with left and right defect spaces equal to the same $N$.
As detailed in Example~4.27 of Ref.~\cite{gross2017schur}, this stochastic Lagrangian is the one identified in Ref.~\cite{ZhuKueGra16} as the sole non-trivial one in case of $t=4$.

In contrast, we do not know (but suspect) that we pay a price by restricting from the full Haar symmetrizer to the one over diagonal matrices in Eq.~(\ref{eqn:diagonal haar bound}).
For the two cases that saturate the bounds in Proposition~\ref{prop:hamming bound} and Proposition~\ref{prop:defect hamming bound}, we can compute the full projection explictily and show that at least there, Eq.~(\ref{eqn:diagonal haar bound}) indeed fails to be tight.

One can expand the anti-id $\overline\one$ in terms of Pauli operators
\cite{gross2017schur}
\begin{myalign}\label{eqn:anti-id expansion}
  \overline\one = \frac12 \big(\one^{\otimes t} + X^{\otimes t} + Y^{\otimes t} + Z^{\otimes t}\big).
\end{myalign}
Then
\begin{align}
  2^{-t} \big(r(\overline\one), P_H[r(\overline\one)]\big) 
  =&
  2^{-t}
  \int
  \tr
  r(\overline\one)U^{\otimes t} r(\overline\one)^\dagger (U^\dagger)^{\otimes t}\; \mathrm{d} \newtext{\mu_H(U)} \nonumber \\
  =&
  2^{-t-2}
  \sum_{i,j=0}^3
  \int
  \tr \sigma_i^{\otimes t} U^{\otimes t} \sigma_j^{\otimes t} (U^\dagger)^{\otimes t} \, \mathrm{d} \newtext{\mu_H(U)} \nonumber \\
  =&
  2^{-t-2}
  \sum_{i,j}
  \int
  \Big(\tr \sigma_i U \sigma_j U^\dagger\Big)^{t} \, \mathrm{d} \newtext{\mu_H(U)}\nonumber \\
  =&
    2^{-2}
	+
    2^{-t-2}
	\sum_{i,j\neq 0}
	\int
	\Big(\tr \sigma_i U \sigma_j U^\dagger\Big)^{t} \, \mathrm{d} \newtext{\mu_H(U)}\nonumber  \\
  =&
    2^{-2}
	+
    2^{-2}
	9
	\frac{1}{4\pi}
	\int_{S^2} x_1^{t} \mathrm{d} x \label{eqn:sphere integration} \\
  =&
  \frac14
  +
	\frac{9}{4}
	\frac{1}{4\pi}
	  \frac{4\pi}{1+t}
  =
  \frac14\Big(1+\frac{9}{t+1}\Big), \nonumber
\end{align}
where in (\ref{eqn:sphere integration}), we have interpreted the Haar integral over inner products of Paulis as an integral over the Bloch sphere and in the next line, used the formula from \cite{folland2001integrate}.
For $t=2$, Eq.~(\ref{eqn:anti-id}) is just the swap operator (i.e., a permutation), and the formula gives $1$, as it should.
The smallest non-trivial case is $t=6$ \cite{gross2017schur} , where we get roughly $0.571<0.65$.

Next, we consider the CSS code $P_N$ for $N=(1,1,1,1)$.
We use the results in Section~3 of Ref.~\cite{ZhuKueGra16}.
For a given partition $\lambda$, let $W_\lambda$ be the associated Weyl module and $S_\lambda$ the Schur module.
As in Ref.~\cite{ZhuKueGra16}, let $W^+_\lambda \subset W_\lambda$ be the subspace such that
\begin{align*}
  \big(W_\lambda \otimes S_\lambda\big)\cap \mathrm{range}\, P_N = W_\lambda^+ \otimes S_\lambda.
\end{align*}
For the projection operators onto the various spaces, we write $P_\lambda$ (Schur module), $Q_\lambda$ (Weyl module), and $Q^+_\lambda$ (the subspace defined above).
Then \cite{ZhuKueGra16}
\begin{align*}
  P_N = \sum_\lambda Q^+_\lambda \otimes P_\lambda.
\end{align*}
By Schur's Lemma,
\begin{align*}
  P_H[P_N] = \sum_\lambda c_\lambda Q_\lambda \otimes P_\lambda,
\end{align*}
for suitable coefficients $c_\lambda$, which are seen to equal $c_\lambda = D^+_\lambda/D_\lambda$ by the fact that Haar averaging preserves the trace.
Hence, using Table~1 of Ref.~\cite{ZhuKueGra16} for $d=2$,
\begin{align*}
  2^{-t+2\dim N}
  (P_N, P_\mathrm{H}[P_N])
  =
  2^{-2}
  \sum_\lambda \frac{d_\lambda  (D^+_\lambda)^2 }{D_\lambda}
  =
  \frac7{10} < \frac78.
\end{align*}

\section{Saturation of higher R\'{e}nyi-entropies in $K$-interleaved Clifford circuits}\label{section:renyi}
\newtext{Consider the R\'{e}nyi-entropies which are defined as
\begin{equation}
S_{\alpha}(\rho):=\frac{1}{1-\alpha}\log\mathrm{Tr}[\rho^{\alpha}]
\end{equation}
for $\alpha>0$. For $\alpha\searrow1$ the standard von Neumann entropy is recovered.
Here, we are interested in the entanglement properties of random state vectors $|\psi\rangle$ on $n$ qubits.
We consider a bi-partition of the $n$ qubits into a set $A$ consisting of constantly many qubits $n_A$ and a set $B$ of $n_B=n-n_A$ many qubits that constitutes the complement of $A$. To derive concentration bounds on these quantities over random ensembles of states, we study the ``higher purities'' $\mathrm{Tr}[\rho^{\alpha}]$ for positive 
integer $\alpha$ in more detail. 
First, we compute the Haar average of this quantity.
Let $\pi_{\mathrm{cyc}}\in S_{\alpha}$ be any full $\alpha$-cycle.
We compute
\begin{align}
\begin{split}
\mathbb{E}_{U\sim\mu_H}\mathrm{Tr}[\rho_A^{\alpha}]&= \mathbb{E}_{U\sim\mu_H}\mathrm{Tr}\left[\mathrm{Tr}_B[|\psi\rangle\langle\psi|]^{\alpha}\right]\\
&=\mathbb{E}_{U\sim\mu_H}\mathrm{Tr}\left[r(\pi_{\mathrm{cyc}})_A\otimes \mathbbm{1}_B(|\psi\rangle\langle\psi|)^{\otimes \alpha}\right]\\
&= {2^n+\alpha-1\choose \alpha}^{-1}\mathrm{Tr}
\left[r(\pi_{\mathrm{cyc}})_A\otimes \mathbbm{1}_B P_{\mathrm{sym},\alpha}\right]\\
&={2^n+\alpha-1\choose \alpha}^{-1}\alpha!^{-1}\sum_{\sigma\in S_{\alpha}}\mathrm{Tr}\left[r(\pi_{\mathrm{cyc}}\circ \sigma)_A\otimes r(\sigma)_B\right]\\
&={2^n+\alpha-1\choose \alpha}^{-1}\alpha !^{-1}\sum_{\sigma\in S_{\alpha}} 2^{n_A \#\mathrm{cyc}(\pi_{\mathrm{cyc}}\circ \sigma)}2^{n_B \#\mathrm{cyc}(\sigma)}\\
&=\frac{1}{2^n(2^n+1)\ldots(2^n+\alpha-1)}\sum_{\sigma\in S_{\alpha}} 2^{n_A \#\mathrm{cyc}(\pi_{\mathrm{cyc}}\circ \sigma)}2^{n_B \#\mathrm{cyc}(\sigma)}\\
&=\frac{2^{\alpha n_B}2^{n_A}}{2^n(2^n+1)\ldots(2^n+\alpha-1)}+O(2^{-n_B})\\
&=2^{-(\alpha-1)n_A}+O(2^{-n_B}),
\end{split}
 \end{align}
where $O(2^{-n_B})$ depends on $\alpha$.
Therefore, up to an exponentially small correction, the average higher purity is minimal.

Next, we compute the same average over an additive $\varepsilon$-approximate unitary $t$-design.
Recall that this is a probability distribution $\nu$ such that 
\begin{equation}
||M_t(\nu)-M_t(\mu_H)||_{\Diamond}\leq \varepsilon.
\end{equation}
By definition of the diamond norm, this also implies
\begin{equation}
||M_t(\nu)-M_t(\mu_H)||_{1\to 1}\leq \varepsilon.
\end{equation}
From this, we obtain
\begin{align}
\begin{split}
&\mathbb{E}_{U\sim\nu}\mathrm{Tr}[\rho_A^{\alpha}]= \mathbb{E}_{U\sim\nu}\mathrm{Tr}\left[\mathrm{Tr}_B[|\psi\rangle\langle\psi|]^{\alpha}\right]\\
&=\mathrm{Tr}\left[r(\pi_{\mathrm{cyc}})_A\otimes \mathbbm{1}_B\mathbb{E}_{U\sim\nu}(|\psi\rangle\langle\psi|)^{\otimes \alpha}\right]\\
&\leq \mathrm{Tr}\left[r(\pi_{\mathrm{cyc}})_A\otimes \mathbbm{1}_B\mathbb{E}_{U\sim\mu_H}(|\psi\rangle\langle\psi|)^{\otimes \alpha}\right]+\left|\mathrm{Tr}[r(\pi_{\mathrm{cyc}})_A\otimes \mathbbm{1}_B(M_t(\nu)-M_t(\mu_H))\left[(|\psi_0\rangle\langle\psi_0|)^{\otimes \alpha}\right]\right|\\
&\leq \mathrm{Tr}\left[r(\pi_{\mathrm{cyc}})_A\otimes \mathbbm{1}_B\mathbb{E}_{U\sim\mu_H}(|\psi\rangle\langle\psi|)^{\otimes \alpha}\right]+\left|\left|(M_t(\nu)-M_t(\mu_H))\left[(|\psi\rangle\langle\psi|)^{\otimes \alpha}\right]\right|\right|_1\\
&\leq 2^{-(\alpha-1)n_A}+O(2^{-n})+\varepsilon.
\end{split}
\end{align}
It suffices to insert $C(K)\log^2(t)(t^4+t\log(1/\varepsilon))$ non-Clifford gates into random Clifford circuits to generate an additive $\varepsilon$-approximate $t$-designs.
Therefore, we can choose $\varepsilon=2^{-2(\alpha-1)n_A}$ and $t=\alpha$ and find that a $K$-interleaved Clifford circuit with $k=C(K)\log^2(\alpha)(\alpha^4+2(\alpha-1)n_A)$ satisfies
\begin{equation}
\mathbb{E}_{U\sim\sigma^{*k}}\mathrm{Tr}[\rho_A^{\alpha}]\leq (1-2^{-(\alpha-1)n_A})2^{-(\alpha-1)n_A}+O(2^{-n})\leq (1-2^{-(\alpha-1)n_A}-O(2^{-n}))2^{-(\alpha-1)n_A}.
\end{equation}
Therefore, for every constant $n_A$ and $\alpha$, there is a classically simulable ensemble of quantum circuits that generate essentially minimal higher purities on average.
}

%\bibliography{Supremacy,Average_PEPS,Clifford-RB-tomo,BigReferences57}

\begin{thebibliography}{75}%
	\makeatletter
	\providecommand \@ifxundefined [1]{%
		\@ifx{#1\undefined}
	}%
	\providecommand \@ifnum [1]{%
		\ifnum #1\expandafter \@firstoftwo
		\else \expandafter \@secondoftwo
		\fi
	}%
	\providecommand \@ifx [1]{%
		\ifx #1\expandafter \@firstoftwo
		\else \expandafter \@secondoftwo
		\fi
	}%
	\providecommand \natexlab [1]{#1}%
	\providecommand \enquote  [1]{``#1''}%
	\providecommand \bibnamefont  [1]{#1}%
	\providecommand \bibfnamefont [1]{#1}%
	\providecommand \citenamefont [1]{#1}%
	\providecommand \href@noop [0]{\@secondoftwo}%
	\providecommand \href [0]{\begingroup \@sanitize@url \@href}%
	\providecommand \@href[1]{\@@startlink{#1}\@@href}%
	\providecommand \@@href[1]{\endgroup#1\@@endlink}%
	\providecommand \@sanitize@url [0]{\catcode `\\12\catcode `\$12\catcode
		`\&12\catcode `\#12\catcode `\^12\catcode `\_12\catcode `\%12\relax}%
	\providecommand \@@startlink[1]{}%
	\providecommand \@@endlink[0]{}%
	\providecommand \url  [0]{\begingroup\@sanitize@url \@url }%
	\providecommand \@url [1]{\endgroup\@href {#1}{\urlprefix }}%
	\providecommand \urlprefix  [0]{URL }%
	\providecommand \Eprint [0]{\href }%
	\providecommand \doibase [0]{http://dx.doi.org/}%
	\providecommand \selectlanguage [0]{\@gobble}%
	\providecommand \bibinfo  [0]{\@secondoftwo}%
	\providecommand \bibfield  [0]{\@secondoftwo}%
	\providecommand \translation [1]{[#1]}%
	\providecommand \BibitemOpen [0]{}%
	\providecommand \bibitemStop [0]{}%
	\providecommand \bibitemNoStop [0]{.\EOS\space}%
	\providecommand \EOS [0]{\spacefactor3000\relax}%
	\providecommand \BibitemShut  [1]{\csname bibitem#1\endcsname}%
	\let\auto@bib@innerbib\@empty
	%</preamble>
	\bibitem [{\citenamefont {Emerson}\ \emph {et~al.}(2005)\citenamefont
		{Emerson}, \citenamefont {Alicki},\ and\ \citenamefont
		{Zyczkowski}}]{ZyczkowskiRB}%
	\BibitemOpen
	\bibfield  {author} {\bibinfo {author} {\bibfnamefont {J.}~\bibnamefont
			{Emerson}}, \bibinfo {author} {\bibfnamefont {R.}~\bibnamefont {Alicki}}, \
		and\ \bibinfo {author} {\bibfnamefont {K.}~\bibnamefont {Zyczkowski}},\
	}\bibfield  {title} {\enquote {\bibinfo {title} {Scalable noise estimation
				with random unitary operators},}\ }\href@noop {} {\bibfield  {journal}
		{\bibinfo  {journal} {J. Opt. B}\ }\textbf {\bibinfo {volume} {7}},\ \bibinfo
		{pages} {S347--S352} (\bibinfo {year} {2005})}\BibitemShut {NoStop}%
	\bibitem [{\citenamefont {{Magesan}}\ \emph {et~al.}(2012)\citenamefont
		{{Magesan}}, \citenamefont {{Gambetta}},\ and\ \citenamefont
		{{Emerson}}}]{MagGamEmer}%
	\BibitemOpen
	\bibfield  {author} {\bibinfo {author} {\bibfnamefont {E.}~\bibnamefont
			{{Magesan}}}, \bibinfo {author} {\bibfnamefont {J.~M.}\ \bibnamefont
			{{Gambetta}}}, \ and\ \bibinfo {author} {\bibfnamefont {J.}~\bibnamefont
			{{Emerson}}},\ }\bibfield  {title} {\enquote {\bibinfo {title}
			{Characterizing quantum gates via randomized benchmarking},}\ }\href@noop {}
	{\bibfield  {journal} {\bibinfo  {journal} {\pra}\ }\textbf {\bibinfo
			{volume} {85}},\ \bibinfo {eid} {042311} (\bibinfo {year}
		{2012})}\BibitemShut {NoStop}%
	\bibitem [{\citenamefont {Knill}\ \emph {et~al.}(2008)\citenamefont {Knill},
		\citenamefont {Leibfried}, \citenamefont {Reichle}, \citenamefont {Britton},
		\citenamefont {Blakestad}, \citenamefont {Jost}, \citenamefont {Langer},
		\citenamefont {Ozeri}, \citenamefont {Seidelin},\ and\ \citenamefont
		{Wineland}}]{KnillBenchmarking}%
	\BibitemOpen
	\bibfield  {author} {\bibinfo {author} {\bibfnamefont {E.}~\bibnamefont
			{Knill}}, \bibinfo {author} {\bibfnamefont {D.}~\bibnamefont {Leibfried}},
		\bibinfo {author} {\bibfnamefont {R.}~\bibnamefont {Reichle}}, \bibinfo
		{author} {\bibfnamefont {J.}~\bibnamefont {Britton}}, \bibinfo {author}
		{\bibfnamefont {R.~B.}\ \bibnamefont {Blakestad}}, \bibinfo {author}
		{\bibfnamefont {J.~D.}\ \bibnamefont {Jost}}, \bibinfo {author}
		{\bibfnamefont {C.}~\bibnamefont {Langer}}, \bibinfo {author} {\bibfnamefont
			{R.}~\bibnamefont {Ozeri}}, \bibinfo {author} {\bibfnamefont
			{S.}~\bibnamefont {Seidelin}}, \ and\ \bibinfo {author} {\bibfnamefont
			{D.~J.}\ \bibnamefont {Wineland}},\ }\bibfield  {title} {\enquote {\bibinfo
			{title} {Randomized benchmarking of quantum gates},}\ }\href@noop {}
	{\bibfield  {journal} {\bibinfo  {journal} {Phys. Rev. A}\ }\textbf {\bibinfo
			{volume} {77}},\ \bibinfo {pages} {012307} (\bibinfo {year}
		{2008})}\BibitemShut {NoStop}%
	\bibitem [{\citenamefont {Hayden}\ and\ \citenamefont
		{Preskill}(2007)}]{HaydenBlackHoles}%
	\BibitemOpen
	\bibfield  {author} {\bibinfo {author} {\bibfnamefont {P.}~\bibnamefont
			{Hayden}}\ and\ \bibinfo {author} {\bibfnamefont {J.}~\bibnamefont
			{Preskill}},\ }\bibfield  {title} {\enquote {\bibinfo {title} {Black holes as
				mirrors: quantum information in random subsystems},}\ }\href@noop {}
	{\bibfield  {journal} {\bibinfo  {journal} {JHEP}\ }\textbf {\bibinfo
			{volume} {0709}},\ \bibinfo {pages} {120} (\bibinfo {year}
		{2007})}\BibitemShut {NoStop}%
	\bibitem [{\citenamefont {Dankert}\ \emph {et~al.}(2009)\citenamefont
		{Dankert}, \citenamefont {Cleve}, \citenamefont {Emerson},\ and\
		\citenamefont {Livine}}]{dankert_exact_2009}%
	\BibitemOpen
	\bibfield  {author} {\bibinfo {author} {\bibfnamefont {C.}~\bibnamefont
			{Dankert}}, \bibinfo {author} {\bibfnamefont {R.}~\bibnamefont {Cleve}},
		\bibinfo {author} {\bibfnamefont {J.}~\bibnamefont {Emerson}}, \ and\
		\bibinfo {author} {\bibfnamefont {E.}~\bibnamefont {Livine}},\ }\bibfield
	{title} {\enquote {\bibinfo {title} {Exact and approximate unitary 2-designs
				and their application to fidelity estimation},}\ }\href@noop {} {\bibfield
		{journal} {\bibinfo  {journal} {Phys. Rev. A}\ }\textbf {\bibinfo {volume}
			{80}},\ \bibinfo {pages} {012304} (\bibinfo {year} {2009})}\BibitemShut
	{NoStop}%
	\bibitem [{\citenamefont {{Dankert}}(2005)}]{DankertThesis}%
	\BibitemOpen
	\bibfield  {author} {\bibinfo {author} {\bibfnamefont {C.}~\bibnamefont
			{{Dankert}}},\ }\href@noop {} {\enquote {\bibinfo {title} {{MSc} thesis,
				{U}niversity of {W}aterloo},}\ } (\bibinfo {year} {2005}),\ \bibinfo {note}
	{arXiv:quant-ph/0512217}\BibitemShut {NoStop}%
	\bibitem [{\citenamefont {Gross}\ \emph {et~al.}(2007)\citenamefont {Gross},
		\citenamefont {Audenaert},\ and\ \citenamefont {Eisert}}]{GroAudEis}%
	\BibitemOpen
	\bibfield  {author} {\bibinfo {author} {\bibfnamefont {D.}~\bibnamefont
			{Gross}}, \bibinfo {author} {\bibfnamefont {K.}~\bibnamefont {Audenaert}}, \
		and\ \bibinfo {author} {\bibfnamefont {J.}~\bibnamefont {Eisert}},\
	}\bibfield  {title} {\enquote {\bibinfo {title} {Evenly distributed
				unitaries: on the structure of unitary designs},}\ }\href@noop {} {\bibfield
		{journal} {\bibinfo  {journal} {J. Math. Phys.}\ }\textbf {\bibinfo {volume}
			{48}},\ \bibinfo {pages} {052104} (\bibinfo {year} {2007})}\BibitemShut
	{NoStop}%
	\bibitem [{\citenamefont {Ambainis}\ \emph {et~al.}(2009)\citenamefont
		{Ambainis}, \citenamefont {Bouda},\ and\ \citenamefont
		{Winter}}]{AmbainisEtAl:2009}%
	\BibitemOpen
	\bibfield  {author} {\bibinfo {author} {\bibfnamefont {A.}~\bibnamefont
			{Ambainis}}, \bibinfo {author} {\bibfnamefont {J.}~\bibnamefont {Bouda}}, \
		and\ \bibinfo {author} {\bibfnamefont {A.}~\bibnamefont {Winter}},\
	}\bibfield  {title} {\enquote {\bibinfo {title} {Nonmalleable encryption of
				quantum information},}\ }\href@noop {} {\bibfield  {journal} {\bibinfo
			{journal} {J. Math. Phys.}\ }\textbf {\bibinfo {volume} {50}},\ \bibinfo
		{pages} {042106} (\bibinfo {year} {2009})}\BibitemShut {NoStop}%
	\bibitem [{\citenamefont {DiVincenzo}\ \emph {et~al.}(2002)\citenamefont
		{DiVincenzo}, \citenamefont {Leung},\ and\ \citenamefont
		{Terhal}}]{divinzeno_data_2001}%
	\BibitemOpen
	\bibfield  {author} {\bibinfo {author} {\bibfnamefont {D.~P}\ \bibnamefont
			{DiVincenzo}}, \bibinfo {author} {\bibfnamefont {D.~W.}\ \bibnamefont
			{Leung}}, \ and\ \bibinfo {author} {\bibfnamefont {B.~M.}\ \bibnamefont
			{Terhal}},\ }\bibfield  {title} {\enquote {\bibinfo {title} {Quantum data
				hiding},}\ }\href@noop {} {\bibfield  {journal} {\bibinfo  {journal} {IEEE,
				Trans. Inf Theory}\ }\textbf {\bibinfo {volume} {48}},\ \bibinfo {pages}
		{3580--599} (\bibinfo {year} {2002})}\BibitemShut {NoStop}%
	\bibitem [{\citenamefont {Matthews}\ \emph {et~al.}(2009)\citenamefont
		{Matthews}, \citenamefont {Wehner},\ and\ \citenamefont
		{Winter}}]{0810.2327}%
	\BibitemOpen
	\bibfield  {author} {\bibinfo {author} {\bibfnamefont {W.}~\bibnamefont
			{Matthews}}, \bibinfo {author} {\bibfnamefont {S.}~\bibnamefont {Wehner}}, \
		and\ \bibinfo {author} {\bibfnamefont {A.}~\bibnamefont {Winter}},\
	}\bibfield  {title} {{\selectlanguage {English}\enquote {\bibinfo {title}
				{Distinguishability of quantum states under restricted families of
					measurements with an application to quantum data hiding},}\ }}\href@noop {}
	{\bibfield  {journal} {\bibinfo  {journal} {Commun. Math. Phys.}\ }\textbf
		{\bibinfo {volume} {291}},\ \bibinfo {pages} {813--843} (\bibinfo {year}
		{2009})}\BibitemShut {NoStop}%
	\bibitem [{\citenamefont {Sen}(2006)}]{sen_random_2006}%
	\BibitemOpen
	\bibfield  {author} {\bibinfo {author} {\bibfnamefont {P.}~\bibnamefont
			{Sen}},\ }\bibfield  {title} {\enquote {\bibinfo {title} {Random measurement
				bases, quantum state distinction and applications to the hidden subgroup
				problem},}\ }\href@noop {} {\bibfield  {journal} {\bibinfo  {journal} {IEEE
				Conference on Computational Complexity}\ ,\ \bibinfo {pages} {274--287}}
		(\bibinfo {year} {2006})}\BibitemShut {NoStop}%
	\bibitem [{\citenamefont {Hayashi}\ \emph {et~al.}(2005)\citenamefont
		{Hayashi}, \citenamefont {Hashimoto},\ and\ \citenamefont
		{Horibe}}]{PhysRevA.72.032325}%
	\BibitemOpen
	\bibfield  {author} {\bibinfo {author} {\bibfnamefont {A.}~\bibnamefont
			{Hayashi}}, \bibinfo {author} {\bibfnamefont {T.}~\bibnamefont {Hashimoto}},
		\ and\ \bibinfo {author} {\bibfnamefont {M.}~\bibnamefont {Horibe}},\
	}\bibfield  {title} {\enquote {\bibinfo {title} {Reexamination of optimal
				quantum state estimation of pure states},}\ }\href@noop {} {\bibfield
		{journal} {\bibinfo  {journal} {Phys. Rev. A}\ }\textbf {\bibinfo {volume}
			{72}},\ \bibinfo {pages} {032325} (\bibinfo {year} {2005})}\BibitemShut
	{NoStop}%
	\bibitem [{\citenamefont {Scott}(2008)}]{ScottDesigns}%
	\BibitemOpen
	\bibfield  {author} {\bibinfo {author} {\bibfnamefont {A.~J.}\ \bibnamefont
			{Scott}},\ }\bibfield  {title} {\enquote {\bibinfo {title} {Optimizing
				quantum process tomography with unitary 2-designs},}\ }\href@noop {}
	{\bibfield  {journal} {\bibinfo  {journal} {J. Phys. A}\ }\textbf {\bibinfo
			{volume} {41}},\ \bibinfo {pages} {055308} (\bibinfo {year} {2008})},\
	\bibinfo {note} {arXiv:0711.1017}\BibitemShut {NoStop}%
	\bibitem [{\citenamefont {Zhu}\ and\ \citenamefont
		{Englert}(2011)}]{PhysRevA.84.022327}%
	\BibitemOpen
	\bibfield  {author} {\bibinfo {author} {\bibfnamefont {H.}~\bibnamefont
			{Zhu}}\ and\ \bibinfo {author} {\bibfnamefont {B.-G.}\ \bibnamefont
			{Englert}},\ }\bibfield  {title} {\enquote {\bibinfo {title} {Quantum state
				tomography with fully symmetric measurements and product measurements},}\
	}\href@noop {} {\bibfield  {journal} {\bibinfo  {journal} {Phys. Rev. A}\
		}\textbf {\bibinfo {volume} {84}},\ \bibinfo {pages} {022327} (\bibinfo
		{year} {2011})}\BibitemShut {NoStop}%
	\bibitem [{\citenamefont {Roth}\ \emph {et~al.}(2018)\citenamefont {Roth},
		\citenamefont {Kueng}, \citenamefont {Kimmel}, \citenamefont {Liu},
		\citenamefont {Gross}, \citenamefont {Eisert},\ and\ \citenamefont
		{Kliesch}}]{AverageGateFidelities}%
	\BibitemOpen
	\bibfield  {author} {\bibinfo {author} {\bibfnamefont {I.}~\bibnamefont
			{Roth}}, \bibinfo {author} {\bibfnamefont {R.}~\bibnamefont {Kueng}},
		\bibinfo {author} {\bibfnamefont {S.}~\bibnamefont {Kimmel}}, \bibinfo
		{author} {\bibfnamefont {Y.-K.}\ \bibnamefont {Liu}}, \bibinfo {author}
		{\bibfnamefont {D.}~\bibnamefont {Gross}}, \bibinfo {author} {\bibfnamefont
			{J.}~\bibnamefont {Eisert}}, \ and\ \bibinfo {author} {\bibfnamefont
			{M.}~\bibnamefont {Kliesch}},\ }\bibfield  {title} {\enquote {\bibinfo
			{title} {Recovering quantum gates from few average gate fidelities},}\
	}\href@noop {} {\bibfield  {journal} {\bibinfo  {journal} {Phys. Rev. Lett.}\
		}\textbf {\bibinfo {volume} {121}},\ \bibinfo {pages} {170502} (\bibinfo
		{year} {2018})}\BibitemShut {NoStop}%
	\bibitem [{\citenamefont {{Kueng}}\ \emph {et~al.}(2016)\citenamefont
		{{Kueng}}, \citenamefont {{Zhu}},\ and\ \citenamefont
		{{Gross}}}]{KueZhuGro16b}%
	\BibitemOpen
	\bibfield  {author} {\bibinfo {author} {\bibfnamefont {R.}~\bibnamefont
			{{Kueng}}}, \bibinfo {author} {\bibfnamefont {H.}~\bibnamefont {{Zhu}}}, \
		and\ \bibinfo {author} {\bibfnamefont {D.}~\bibnamefont {{Gross}}},\
	}\href@noop {} {\enquote {\bibinfo {title} {{Distinguishing quantum states
					using Clifford orbits}},}\ } (\bibinfo {year} {2016}),\ \bibinfo {note}
	{arXiv:1609.08595}\BibitemShut {NoStop}%
	\bibitem [{\citenamefont {Gross}\ \emph {et~al.}(2015)\citenamefont {Gross},
		\citenamefont {Krahmer},\ and\ \citenamefont {Kueng}}]{GroKraKue15_partial}%
	\BibitemOpen
	\bibfield  {author} {\bibinfo {author} {\bibfnamefont {D.}~\bibnamefont
			{Gross}}, \bibinfo {author} {\bibfnamefont {F.}~\bibnamefont {Krahmer}}, \
		and\ \bibinfo {author} {\bibfnamefont {R.}~\bibnamefont {Kueng}},\ }\bibfield
	{title} {\enquote {\bibinfo {title} {A partial derandomization of
				{PhaseLift} using spherical designs},}\ }\href@noop {} {\bibfield  {journal}
		{\bibinfo  {journal} {J. Fourier Anal. Appl.}\ }\textbf {\bibinfo {volume}
			{21}},\ \bibinfo {pages} {229--266} (\bibinfo {year} {2015})}\BibitemShut
	{NoStop}%
	\bibitem [{\citenamefont {Szehr}\ \emph {et~al.}(2013)\citenamefont {Szehr},
		\citenamefont {Dupuis}, \citenamefont {Tomamichel},\ and\ \citenamefont
		{Renner}}]{szehr2013decoupling}%
	\BibitemOpen
	\bibfield  {author} {\bibinfo {author} {\bibfnamefont {O.}~\bibnamefont
			{Szehr}}, \bibinfo {author} {\bibfnamefont {F.}~\bibnamefont {Dupuis}},
		\bibinfo {author} {\bibfnamefont {M.}~\bibnamefont {Tomamichel}}, \ and\
		\bibinfo {author} {\bibfnamefont {R.}~\bibnamefont {Renner}},\ }\bibfield
	{title} {\enquote {\bibinfo {title} {Decoupling with unitary approximate
				two-designs},}\ }\href@noop {} {\bibfield  {journal} {\bibinfo  {journal}
			{New J. Phys.}\ }\textbf {\bibinfo {volume} {15}},\ \bibinfo {pages} {053022}
		(\bibinfo {year} {2013})}\BibitemShut {NoStop}%
	\bibitem [{\citenamefont {Brandao}\ and\ \citenamefont
		{Horodecki}(2013)}]{Generic}%
	\BibitemOpen
	\bibfield  {author} {\bibinfo {author} {\bibfnamefont {F.~G. S.~L.}\
			\bibnamefont {Brandao}}\ and\ \bibinfo {author} {\bibfnamefont
			{M.}~\bibnamefont {Horodecki}},\ }\bibfield  {title} {\enquote {\bibinfo
			{title} {Exponential quantum speed-ups are generic},}\ }\href@noop {}
	{\bibfield  {journal} {\bibinfo  {journal} {Quant. Inf. Comp.}\ }\textbf
		{\bibinfo {volume} {13}},\ \bibinfo {pages} {0901} (\bibinfo {year}
		{2013})}\BibitemShut {NoStop}%
	\bibitem [{\citenamefont {Haferkamp}\ \emph {et~al.}(2022)\citenamefont
		{Haferkamp}, \citenamefont {Faist}, \citenamefont {N}, \citenamefont
		{Eisert},\ and\ \citenamefont {Halpern}}]{ComplexityGrowth}%
	\BibitemOpen
	\bibfield  {author} {\bibinfo {author} {\bibfnamefont {J.}~\bibnamefont
			{Haferkamp}}, \bibinfo {author} {\bibfnamefont {P.}~\bibnamefont {Faist}},
		\bibinfo {author} {\bibfnamefont {B.~T.~Kothakonda}\ \bibnamefont {N}},
		\bibinfo {author} {\bibfnamefont {J.}~\bibnamefont {Eisert}}, \ and\ \bibinfo
		{author} {\bibfnamefont {N.~Yunger}\ \bibnamefont {Halpern}},\ }\bibfield
	{title} {\enquote {\bibinfo {title} {Linear growth of quantum circuit
				complexity},}\ }\href {\doibase 10.1038/s41567-022-01539-6} {\bibfield
		{journal} {\bibinfo  {journal} {Nature Phys.}\ }\textbf {\bibinfo {volume}
			{18}},\ \bibinfo {pages} {528--532} (\bibinfo {year} {2022})}\BibitemShut
	{NoStop}%
	\bibitem [{\citenamefont {Roberts}\ and\ \citenamefont
		{Yoshida}(2017)}]{RobertsYoshida}%
	\BibitemOpen
	\bibfield  {author} {\bibinfo {author} {\bibfnamefont {D.~A.}\ \bibnamefont
			{Roberts}}\ and\ \bibinfo {author} {\bibfnamefont {B.}~\bibnamefont
			{Yoshida}},\ }\bibfield  {title} {\enquote {\bibinfo {title} {Chaos and
				complexity by design},}\ }\href@noop {} {\bibfield  {journal} {\bibinfo
			{journal} {JHEP}\ }\textbf {\bibinfo {volume} {04}},\ \bibinfo {pages} {121}
		(\bibinfo {year} {2017})}\BibitemShut {NoStop}%
	\bibitem [{\citenamefont {Masanes}\ \emph {et~al.}(2013)\citenamefont
		{Masanes}, \citenamefont {Roncaglia},\ and\ \citenamefont
		{Ac\'{\i}n}}]{PhysRevE.87.032137}%
	\BibitemOpen
	\bibfield  {author} {\bibinfo {author} {\bibfnamefont {L.}~\bibnamefont
			{Masanes}}, \bibinfo {author} {\bibfnamefont {A.~J.}\ \bibnamefont
			{Roncaglia}}, \ and\ \bibinfo {author} {\bibfnamefont {A.}~\bibnamefont
			{Ac\'{\i}n}},\ }\bibfield  {title} {\enquote {\bibinfo {title} {Complexity of
				energy eigenstates as a mechanism for equilibration},}\ }\href@noop {}
	{\bibfield  {journal} {\bibinfo  {journal} {Phys. Rev. E}\ }\textbf {\bibinfo
			{volume} {87}},\ \bibinfo {pages} {032137} (\bibinfo {year}
		{2013})}\BibitemShut {NoStop}%
	\bibitem [{\citenamefont {Onorati}\ \emph {et~al.}(2017)\citenamefont
		{Onorati}, \citenamefont {Buerschaper}, \citenamefont {Kliesch},
		\citenamefont {Brown}, \citenamefont {Werner},\ and\ \citenamefont
		{Eisert}}]{onorati_mixing_2017}%
	\BibitemOpen
	\bibfield  {author} {\bibinfo {author} {\bibfnamefont {E.}~\bibnamefont
			{Onorati}}, \bibinfo {author} {\bibfnamefont {O.}~\bibnamefont
			{Buerschaper}}, \bibinfo {author} {\bibfnamefont {M.}~\bibnamefont
			{Kliesch}}, \bibinfo {author} {\bibfnamefont {W.}~\bibnamefont {Brown}},
		\bibinfo {author} {\bibfnamefont {A.~H.}\ \bibnamefont {Werner}}, \ and\
		\bibinfo {author} {\bibfnamefont {J.}~\bibnamefont {Eisert}},\ }\bibfield
	{title} {\enquote {\bibinfo {title} {Mixing {properties} of {stochastic}
				{quantum} {Hamiltonians}},}\ }\href@noop {} {\bibfield  {journal} {\bibinfo
			{journal} {Commun. Math. Phys.}\ }\textbf {\bibinfo {volume} {355}},\
		\bibinfo {pages} {905--947} (\bibinfo {year} {2017})}\BibitemShut {NoStop}%
	\bibitem [{\citenamefont {Brand\~{a}o}\ \emph
		{et~al.}(2016{\natexlab{a}})\citenamefont {Brand\~{a}o}, \citenamefont
		{Harrow},\ and\ \citenamefont {Horodecki}}]{brandao_local_2016}%
	\BibitemOpen
	\bibfield  {author} {\bibinfo {author} {\bibfnamefont {F.~G. S.~L.}\
			\bibnamefont {Brand\~{a}o}}, \bibinfo {author} {\bibfnamefont {A.~W.}\
			\bibnamefont {Harrow}}, \ and\ \bibinfo {author} {\bibfnamefont
			{M.}~\bibnamefont {Horodecki}},\ }\bibfield  {title} {\enquote {\bibinfo
			{title} {Local {random} {quantum} {circuits} are {approximate}
				{polynomial}-{designs}},}\ }\href@noop {} {\bibfield  {journal} {\bibinfo
			{journal} {Commun. Math. Phys.}\ }\textbf {\bibinfo {volume} {346}},\
		\bibinfo {pages} {397--434} (\bibinfo {year}
		{2016}{\natexlab{a}})}\BibitemShut {NoStop}%
	\bibitem [{\citenamefont {Brand\~{a}o}\ \emph
		{et~al.}(2016{\natexlab{b}})\citenamefont {Brand\~{a}o}, \citenamefont
		{Harrow},\ and\ \citenamefont {Horodecki}}]{brandao_efficient_2016}%
	\BibitemOpen
	\bibfield  {author} {\bibinfo {author} {\bibfnamefont {F.~G. S.~L.}\
			\bibnamefont {Brand\~{a}o}}, \bibinfo {author} {\bibfnamefont {A.~W.}\
			\bibnamefont {Harrow}}, \ and\ \bibinfo {author} {\bibfnamefont
			{M.}~\bibnamefont {Horodecki}},\ }\bibfield  {title} {\enquote {\bibinfo
			{title} {Efficient {quantum} {pseudorandomness}},}\ }\href@noop {} {\bibfield
		{journal} {\bibinfo  {journal} {Phys. Rev. Lett.}\ }\textbf {\bibinfo
			{volume} {116}} (\bibinfo {year} {2016}{\natexlab{b}})}\BibitemShut {NoStop}%
	\bibitem [{\citenamefont {Cleve}\ \emph {et~al.}(2015)\citenamefont {Cleve},
		\citenamefont {Leung}, \citenamefont {Liu},\ and\ \citenamefont
		{Wang}}]{cleve2015near}%
	\BibitemOpen
	\bibfield  {author} {\bibinfo {author} {\bibfnamefont {R.}~\bibnamefont
			{Cleve}}, \bibinfo {author} {\bibfnamefont {D.}~\bibnamefont {Leung}},
		\bibinfo {author} {\bibfnamefont {L.}~\bibnamefont {Liu}}, \ and\ \bibinfo
		{author} {\bibfnamefont {C.}~\bibnamefont {Wang}},\ }\bibfield  {title}
	{\enquote {\bibinfo {title} {Near-linear constructions of exact unitary
				2-designs},}\ }\href@noop {} {\bibfield  {journal} {\bibinfo  {journal}
			{Quant. Inf. Comp.}\ }\textbf {\bibinfo {volume} {16}},\ \bibinfo {pages}
		{0721--0756} (\bibinfo {year} {2015})}\BibitemShut {NoStop}%
	\bibitem [{\citenamefont {Harrow}\ and\ \citenamefont
		{Low}(2009)}]{harrow_random_2009}%
	\BibitemOpen
	\bibfield  {author} {\bibinfo {author} {\bibfnamefont {A.~W.}\ \bibnamefont
			{Harrow}}\ and\ \bibinfo {author} {\bibfnamefont {R.~A.}\ \bibnamefont
			{Low}},\ }\bibfield  {title} {\enquote {\bibinfo {title} {Random {quantum}
				{circuits} are {approximate} 2-designs},}\ }\href
	{http://arxiv.org/abs/0802.1919} {\bibfield  {journal} {\bibinfo  {journal}
			{Commun. Math. Phys.}\ }\textbf {\bibinfo {volume} {291}},\ \bibinfo {pages}
		{257--302} (\bibinfo {year} {2009})},\ \bibinfo {note} {arXiv:
		0802.1919}\BibitemShut {NoStop}%
	\bibitem [{\citenamefont {Hunter-Jones}(2019)}]{hunter2019unitary}%
	\BibitemOpen
	\bibfield  {author} {\bibinfo {author} {\bibfnamefont {N.}~\bibnamefont
			{Hunter-Jones}},\ }\bibfield  {title} {\enquote {\bibinfo {title} {Unitary
				designs from statistical mechanics in random quantum circuits},}\ }\href@noop
	{} {\  (\bibinfo {year} {2019})},\ \bibinfo {note}
	{arXiv:1905.12053}\BibitemShut {NoStop}%
	\bibitem [{\citenamefont {Gottesman}()}]{QEC2}%
	\BibitemOpen
	\bibfield  {author} {\bibinfo {author} {\bibfnamefont {D.}~\bibnamefont
			{Gottesman}},\ }\href@noop {} {\enquote {\bibinfo {title} {An introduction to
				quantum error correction and fault-tolerant quantum computation},}\ }\bibinfo
	{note} {ArXiv:0904.2557}\BibitemShut {NoStop}%
	\bibitem [{\citenamefont {Campbell}\ \emph {et~al.}(2017)\citenamefont
		{Campbell}, \citenamefont {Terhal},\ and\ \citenamefont
		{Vuillot}}]{EarlFaultTolerant}%
	\BibitemOpen
	\bibfield  {author} {\bibinfo {author} {\bibfnamefont {E.~T.}\ \bibnamefont
			{Campbell}}, \bibinfo {author} {\bibfnamefont {B.~M.}\ \bibnamefont
			{Terhal}}, \ and\ \bibinfo {author} {\bibfnamefont {C.}~\bibnamefont
			{Vuillot}},\ }\bibfield  {title} {\enquote {\bibinfo {title} {Roads towards
				fault-tolerant universal quantum computation},}\ }\href@noop {} {\bibfield
		{journal} {\bibinfo  {journal} {Nature}\ }\textbf {\bibinfo {volume} {549}},\
		\bibinfo {pages} {172--179} (\bibinfo {year} {2017})}\BibitemShut {NoStop}%
	\bibitem [{\citenamefont {Veitch}\ \emph {et~al.}(2014)\citenamefont {Veitch},
		\citenamefont {Mousavian}, \citenamefont {Gottesman},\ and\ \citenamefont
		{Emerson}}]{ResourceTheory}%
	\BibitemOpen
	\bibfield  {author} {\bibinfo {author} {\bibfnamefont {V.}~\bibnamefont
			{Veitch}}, \bibinfo {author} {\bibfnamefont {A.~H.}\ \bibnamefont
			{Mousavian}}, \bibinfo {author} {\bibfnamefont {D.}~\bibnamefont
			{Gottesman}}, \ and\ \bibinfo {author} {\bibfnamefont {J.}~\bibnamefont
			{Emerson}},\ }\bibfield  {title} {\enquote {\bibinfo {title} {The resource
				theory of stabilizer quantum computation},}\ }\href@noop {} {\bibfield
		{journal} {\bibinfo  {journal} {New J. Phys.}\ }\textbf {\bibinfo {volume}
			{16}},\ \bibinfo {pages} {013009} (\bibinfo {year} {2014})}\BibitemShut
	{NoStop}%
	\bibitem [{\citenamefont {Howard}\ and\ \citenamefont
		{Campbell}(2017)}]{PhysRevLett.118.090501}%
	\BibitemOpen
	\bibfield  {author} {\bibinfo {author} {\bibfnamefont {M.}~\bibnamefont
			{Howard}}\ and\ \bibinfo {author} {\bibfnamefont {E.}~\bibnamefont
			{Campbell}},\ }\bibfield  {title} {\enquote {\bibinfo {title} {Application of
				a resource theory for magic states to fault-tolerant quantum computing},}\
	}\href@noop {} {\bibfield  {journal} {\bibinfo  {journal} {Phys. Rev. Lett.}\
		}\textbf {\bibinfo {volume} {118}},\ \bibinfo {pages} {090501} (\bibinfo
		{year} {2017})}\BibitemShut {NoStop}%
	\bibitem [{\citenamefont {Webb}(2015)}]{Webb3Design}%
	\BibitemOpen
	\bibfield  {author} {\bibinfo {author} {\bibfnamefont {Z.}~\bibnamefont
			{Webb}},\ }\bibfield  {title} {\enquote {\bibinfo {title} {{The Clifford
					group forms a unitary 3-design}},}\ }\href@noop {} {\  (\bibinfo {year}
		{2015})},\ \bibinfo {note} {arXiv:1510.02769}\BibitemShut {NoStop}%
	\bibitem [{\citenamefont {Zhu}(2017)}]{PhysRevA.96.062336}%
	\BibitemOpen
	\bibfield  {author} {\bibinfo {author} {\bibfnamefont {H.}~\bibnamefont
			{Zhu}},\ }\bibfield  {title} {\enquote {\bibinfo {title} {Multiqubit clifford
				groups are unitary 3-designs},}\ }\href@noop {} {\bibfield  {journal}
		{\bibinfo  {journal} {Phys. Rev. A}\ }\textbf {\bibinfo {volume} {96}},\
		\bibinfo {pages} {062336} (\bibinfo {year} {2017})}\BibitemShut {NoStop}%
	\bibitem [{\citenamefont {Kueng}\ and\ \citenamefont
		{Gross}(2015)}]{Kueng3Design}%
	\BibitemOpen
	\bibfield  {author} {\bibinfo {author} {\bibfnamefont {R.}~\bibnamefont
			{Kueng}}\ and\ \bibinfo {author} {\bibfnamefont {D.}~\bibnamefont {Gross}},\
	}\bibfield  {title} {\enquote {\bibinfo {title} {Qubit stabilizer states are
				complex projective 3-designs},}\ }\href@noop {} {\  (\bibinfo {year}
		{2015})},\ \bibinfo {note} {arXiv:1510.02767}\BibitemShut {NoStop}%
	\bibitem [{\citenamefont {Zhu}\ \emph {et~al.}()\citenamefont {Zhu},
		\citenamefont {Kueng}, \citenamefont {Grassl},\ and\ \citenamefont
		{Gross}}]{zhu_clifford_2016}%
	\BibitemOpen
	\bibfield  {author} {\bibinfo {author} {\bibfnamefont {H.}~\bibnamefont
			{Zhu}}, \bibinfo {author} {\bibfnamefont {R.}~\bibnamefont {Kueng}}, \bibinfo
		{author} {\bibfnamefont {M.}~\bibnamefont {Grassl}}, \ and\ \bibinfo {author}
		{\bibfnamefont {D.}~\bibnamefont {Gross}},\ }\bibfield  {title} {\enquote
		{\bibinfo {title} {The {Clifford} group fails gracefully to be a unitary
				4-design},}\ }\href {http://arxiv.org/abs/1609.08172} {\ }\bibinfo {note}
	{ArXiv:1609.08172}\BibitemShut {NoStop}%
	\bibitem [{\citenamefont {Helsen}\ \emph {et~al.}(2018)\citenamefont {Helsen},
		\citenamefont {Wallman},\ and\ \citenamefont {Wehner}}]{HelsenClifford}%
	\BibitemOpen
	\bibfield  {author} {\bibinfo {author} {\bibfnamefont {J.}~\bibnamefont
			{Helsen}}, \bibinfo {author} {\bibfnamefont {J.~J.}\ \bibnamefont {Wallman}},
		\ and\ \bibinfo {author} {\bibfnamefont {S.}~\bibnamefont {Wehner}},\
	}\bibfield  {title} {\enquote {\bibinfo {title} {{Representations of the
					multi-qubit Clifford group}},}\ }\href@noop {} {\bibfield  {journal}
		{\bibinfo  {journal} {J. Math. Phys.}\ }\textbf {\bibinfo {volume} {59}},\
		\bibinfo {pages} {072201} (\bibinfo {year} {2018})}\BibitemShut {NoStop}%
	\bibitem [{\citenamefont {Bannai}\ \emph {et~al.}(2020)\citenamefont {Bannai},
		\citenamefont {Navarro}, \citenamefont {Rizo},\ and\ \citenamefont
		{Tiep}}]{bannai_unitary_2020}%
	\BibitemOpen
	\bibfield  {author} {\bibinfo {author} {\bibfnamefont {E.}~\bibnamefont
			{Bannai}}, \bibinfo {author} {\bibfnamefont {G.}~\bibnamefont {Navarro}},
		\bibinfo {author} {\bibfnamefont {N.}~\bibnamefont {Rizo}}, \ and\ \bibinfo
		{author} {\bibfnamefont {P.~H.}\ \bibnamefont {Tiep}},\ }\bibfield  {title}
	{\enquote {\bibinfo {title} {Unitary $t$-groups},}\ }\href {\doibase
		10.2969/jmsj/82228222} {\bibfield  {journal} {\bibinfo  {journal} {J. Math.
				Soc. Japan}\ }\textbf {\bibinfo {volume} {72}},\ \bibinfo {pages} {909 --
			921} (\bibinfo {year} {2020})}\BibitemShut {NoStop}%
	\bibitem [{\citenamefont {Sawicki}\ and\ \citenamefont
		{Karnas}(2017)}]{sawicki_universal_2017}%
	\BibitemOpen
	\bibfield  {author} {\bibinfo {author} {\bibfnamefont {A.}~\bibnamefont
			{Sawicki}}\ and\ \bibinfo {author} {\bibfnamefont {K.}~\bibnamefont
			{Karnas}},\ }\bibfield  {title} {\enquote {\bibinfo {title} {Universality of
				single qudit gates},}\ }\href@noop {} {\bibfield  {journal} {\bibinfo
			{journal} {Ann. Henri Poincar\'e}\ }\textbf {\bibinfo {volume} {18}},\
		\bibinfo {pages} {3515--3552} (\bibinfo {year} {2017})}\BibitemShut {NoStop}%
	\bibitem [{\citenamefont {Koenig}\ and\ \citenamefont
		{Smolin}(2014)}]{koenig_how_2014}%
	\BibitemOpen
	\bibfield  {author} {\bibinfo {author} {\bibfnamefont {R.}~\bibnamefont
			{Koenig}}\ and\ \bibinfo {author} {\bibfnamefont {J.~A.}\ \bibnamefont
			{Smolin}},\ }\bibfield  {title} {\enquote {\bibinfo {title} {How to
				efficiently select an arbitrary {Clifford} group element},}\ }\href
	{http://arxiv.org/abs/1406.2170} {\bibfield  {journal} {\bibinfo  {journal}
			{J. Math. Phys.}\ }\textbf {\bibinfo {volume} {55}},\ \bibinfo {pages}
		{122202} (\bibinfo {year} {2014})},\ \bibinfo {note} {arXiv:
		1406.2170}\BibitemShut {NoStop}%
	\bibitem [{\citenamefont {Nezami}\ and\ \citenamefont
		{Walter}(2016)}]{nezami2016multipartite}%
	\BibitemOpen
	\bibfield  {author} {\bibinfo {author} {\bibfnamefont {S.}~\bibnamefont
			{Nezami}}\ and\ \bibinfo {author} {\bibfnamefont {M.}~\bibnamefont
			{Walter}},\ }\bibfield  {title} {\enquote {\bibinfo {title} {Multipartite
				entanglement in stabilizer tensor networks},}\ }\href@noop {} {\  (\bibinfo
		{year} {2016})},\ \bibinfo {note} {arXiv:1608.02595}\BibitemShut {NoStop}%
	\bibitem [{\citenamefont {Gross}\ \emph {et~al.}(2008)\citenamefont {Gross},
		\citenamefont {Nezami}, \citenamefont {WalMain},\ and\ \citenamefont
		{Gamburd}}]{gross2017schur}%
	\BibitemOpen
	\bibfield  {author} {\bibinfo {author} {\bibfnamefont {D.}~\bibnamefont
			{Gross}}, \bibinfo {author} {\bibfnamefont {S.}~\bibnamefont {Nezami}},
		\bibinfo {author} {\bibfnamefont {J.}~\bibnamefont {WalMain}}, \ and\
		\bibinfo {author} {\bibfnamefont {A.}~\bibnamefont {Gamburd}},\ }\bibfield
	{title} {\enquote {\bibinfo {title} {{Schur}-{Weyl} {duality} for the
				{Clifford} {group} with {applications}},}\ }\href
	{http://dx.optdoi.org/10.1007/s00222-007-0072-z} {\bibfield  {journal}
		{\bibinfo  {journal} {Invent. Math.}\ }\textbf {\bibinfo {volume} {171}},\
		\bibinfo {pages} {83--121} (\bibinfo {year} {2008})}\BibitemShut {NoStop}%
	\bibitem [{\citenamefont {Montealegre-Mora}\ and\ \citenamefont
		{Gross}(2019)}]{FelipeGross}%
	\BibitemOpen
	\bibfield  {author} {\bibinfo {author} {\bibfnamefont {F.}~\bibnamefont
			{Montealegre-Mora}}\ and\ \bibinfo {author} {\bibfnamefont {D.}~\bibnamefont
			{Gross}},\ }\bibfield  {title} {\enquote {\bibinfo {title} {Rank-deficient
				representations in howe duality over finite fields arise from quantum
				codes},}\ }\href@noop {} {\  (\bibinfo {year} {2019})},\ \bibinfo {note}
	{arXiv:1906.07230}\BibitemShut {NoStop}%
	\bibitem [{\citenamefont {Zhou}\ \emph {et~al.}(2019)\citenamefont {Zhou},
		\citenamefont {Yang}, \citenamefont {Hamma},\ and\ \citenamefont
		{Chamon}}]{zhou_entanglement_2019}%
	\BibitemOpen
	\bibfield  {author} {\bibinfo {author} {\bibfnamefont {S.}~\bibnamefont
			{Zhou}}, \bibinfo {author} {\bibfnamefont {Z.-C.}\ \bibnamefont {Yang}},
		\bibinfo {author} {\bibfnamefont {A.}~\bibnamefont {Hamma}}, \ and\ \bibinfo
		{author} {\bibfnamefont {C.}~\bibnamefont {Chamon}},\ }\bibfield  {title}
	{\enquote {\bibinfo {title} {Single $t$ gate in a {Clifford} circuit drives
				transition to universal entanglement spectrum statistics},}\ }\href@noop {}
	{\  (\bibinfo {year} {2019})},\ \bibinfo {note}
	{arXiv:1906.01079}\BibitemShut {NoStop}%
	\bibitem [{\citenamefont {Aaronson}\ and\ \citenamefont
		{Gottesman}(2004{\natexlab{a}})}]{aaronson_stabilizer_2004}%
	\BibitemOpen
	\bibfield  {author} {\bibinfo {author} {\bibfnamefont {S.}~\bibnamefont
			{Aaronson}}\ and\ \bibinfo {author} {\bibfnamefont {D.}~\bibnamefont
			{Gottesman}},\ }\bibfield  {title} {\enquote {\bibinfo {title} {Improved
				simulation of stabilizer circuits},}\ }\href@noop {} {\bibfield  {journal}
		{\bibinfo  {journal} {Phys. Rev. A}\ }\textbf {\bibinfo {volume} {70}},\
		\bibinfo {pages} {052328} (\bibinfo {year} {2004}{\natexlab{a}})}\BibitemShut
	{NoStop}%
	\bibitem [{\citenamefont {Cwiklinski}\ \emph {et~al.}(2013)\citenamefont
		{Cwiklinski}, \citenamefont {Howodecki}, \citenamefont {Mozrzymas},
		\citenamefont {Pankowski},\ and\ \citenamefont
		{Studzinski}}]{cwiklinski_local_2013}%
	\BibitemOpen
	\bibfield  {author} {\bibinfo {author} {\bibfnamefont {P.}~\bibnamefont
			{Cwiklinski}}, \bibinfo {author} {\bibfnamefont {M.}~\bibnamefont
			{Howodecki}}, \bibinfo {author} {\bibfnamefont {M.}~\bibnamefont
			{Mozrzymas}}, \bibinfo {author} {\bibfnamefont {L.}~\bibnamefont
			{Pankowski}}, \ and\ \bibinfo {author} {\bibfnamefont {M.}~\bibnamefont
			{Studzinski}},\ }\bibfield  {title} {\enquote {\bibinfo {title} {Local random
				quantum circuits are approximate polnomial-designs - numerical results},}\
	}\href@noop {} {\bibfield  {journal} {\bibinfo  {journal} {J. Phys. A}\
		}\textbf {\bibinfo {volume} {46}},\ \bibinfo {pages} {305301} (\bibinfo
		{year} {2013})}\BibitemShut {NoStop}%
	\bibitem [{\citenamefont {Bravyi}\ \emph {et~al.}(2019)\citenamefont {Bravyi},
		\citenamefont {Browne}, \citenamefont {Calpin}, \citenamefont {Campbell},
		\citenamefont {Gosset},\ and\ \citenamefont
		{Howard}}]{bravyi_stabilizer_2019}%
	\BibitemOpen
	\bibfield  {author} {\bibinfo {author} {\bibfnamefont {S.}~\bibnamefont
			{Bravyi}}, \bibinfo {author} {\bibfnamefont {D.}~\bibnamefont {Browne}},
		\bibinfo {author} {\bibfnamefont {P.}~\bibnamefont {Calpin}}, \bibinfo
		{author} {\bibfnamefont {E.}~\bibnamefont {Campbell}}, \bibinfo {author}
		{\bibfnamefont {D.}~\bibnamefont {Gosset}}, \ and\ \bibinfo {author}
		{\bibfnamefont {M.}~\bibnamefont {Howard}},\ }\bibfield  {title} {\enquote
		{\bibinfo {title} {Simulation of quantum circuits by low-rank stabilizer
				decomposition},}\ }\href@noop {} {\bibfield  {journal} {\bibinfo  {journal}
			{Quantum}\ }\textbf {\bibinfo {volume} {3}},\ \bibinfo {pages} {181}
		(\bibinfo {year} {2019})}\BibitemShut {NoStop}%
	\bibitem [{\citenamefont {Pashayan}\ \emph {et~al.}(2015)\citenamefont
		{Pashayan}, \citenamefont {Wallman},\ and\ \citenamefont
		{Bartlett}}]{PhysRevLett.115.070501}%
	\BibitemOpen
	\bibfield  {author} {\bibinfo {author} {\bibfnamefont {H.}~\bibnamefont
			{Pashayan}}, \bibinfo {author} {\bibfnamefont {J.~J.}\ \bibnamefont
			{Wallman}}, \ and\ \bibinfo {author} {\bibfnamefont {S.~D.}\ \bibnamefont
			{Bartlett}},\ }\bibfield  {title} {\enquote {\bibinfo {title} {Estimating
				outcome probabilities of quantum circuits using quasiprobabilities},}\
	}\href@noop {} {\bibfield  {journal} {\bibinfo  {journal} {Phys. Rev. Lett.}\
		}\textbf {\bibinfo {volume} {115}},\ \bibinfo {pages} {070501} (\bibinfo
		{year} {2015})}\BibitemShut {NoStop}%
	\bibitem [{\citenamefont {Heinrich}\ and\ \citenamefont
		{Gross}(2019)}]{heinrich2019robustness}%
	\BibitemOpen
	\bibfield  {author} {\bibinfo {author} {\bibfnamefont {M.}~\bibnamefont
			{Heinrich}}\ and\ \bibinfo {author} {\bibfnamefont {D.}~\bibnamefont
			{Gross}},\ }\bibfield  {title} {\enquote {\bibinfo {title} {Robustness of
				magic and symmetries of the stabiliser polytope},}\ }\href@noop {} {\bibfield
		{journal} {\bibinfo  {journal} {Quantum}\ }\textbf {\bibinfo {volume} {3}},\
		\bibinfo {pages} {132} (\bibinfo {year} {2019})}\BibitemShut {NoStop}%
	\bibitem [{\citenamefont {Bravyi}\ and\ \citenamefont
		{Gosset}(2016)}]{PhysRevLett.116.250501}%
	\BibitemOpen
	\bibfield  {author} {\bibinfo {author} {\bibfnamefont {S.}~\bibnamefont
			{Bravyi}}\ and\ \bibinfo {author} {\bibfnamefont {D.}~\bibnamefont
			{Gosset}},\ }\bibfield  {title} {\enquote {\bibinfo {title} {{Improved
					classical simulation of quantum circuits dominated by Clifford gates}},}\
	}\href {\doibase 10.1103/PhysRevLett.116.250501} {\bibfield  {journal}
		{\bibinfo  {journal} {Phys. Rev. Lett.}\ }\textbf {\bibinfo {volume} {116}},\
		\bibinfo {pages} {250501} (\bibinfo {year} {2016})}\BibitemShut {NoStop}%
	\bibitem [{\citenamefont {Seddon}\ \emph {et~al.}(2020)\citenamefont {Seddon},
		\citenamefont {Regular}, \citenamefont {Pashayan}, \citenamefont {Ouyang},\
		and\ \citenamefont {Campbell}}]{seddon2020quantifying}%
	\BibitemOpen
	\bibfield  {author} {\bibinfo {author} {\bibfnamefont {J.}~\bibnamefont
			{Seddon}}, \bibinfo {author} {\bibfnamefont {B.}~\bibnamefont {Regular}},
		\bibinfo {author} {\bibfnamefont {H.}~\bibnamefont {Pashayan}}, \bibinfo
		{author} {\bibfnamefont {Y.}~\bibnamefont {Ouyang}}, \ and\ \bibinfo {author}
		{\bibfnamefont {E.}~\bibnamefont {Campbell}},\ }\bibfield  {title} {\enquote
		{\bibinfo {title} {Quantifying quantum speedups: improved classical
				simulation from tighter magic monotones},}\ }\href@noop {} {\  (\bibinfo
		{year} {2020})},\ \bibinfo {note} {arXiv:2002.06181}\BibitemShut {NoStop}%
	\bibitem [{\citenamefont {Brandao}\ \emph {et~al.}(2021)\citenamefont
		{Brandao}, \citenamefont {Chemissany}, \citenamefont {Hunter-Jones},
		\citenamefont {Kueng},\ and\ \citenamefont
		{Preskill}}]{brandao_complexity_2019}%
	\BibitemOpen
	\bibfield  {author} {\bibinfo {author} {\bibfnamefont {F.~G. S.~L.}\
			\bibnamefont {Brandao}}, \bibinfo {author} {\bibfnamefont {W.}~\bibnamefont
			{Chemissany}}, \bibinfo {author} {\bibfnamefont {N.}~\bibnamefont
			{Hunter-Jones}}, \bibinfo {author} {\bibfnamefont {R.}~\bibnamefont {Kueng}},
		\ and\ \bibinfo {author} {\bibfnamefont {J.}~\bibnamefont {Preskill}},\
	}\bibfield  {title} {\enquote {\bibinfo {title} {Models of quantum complexity
				growth},}\ }\href {\doibase 10.1103/PRXQuantum.2.030316} {\bibfield
		{journal} {\bibinfo  {journal} {PRX Quantum}\ }\textbf {\bibinfo {volume}
			{2}},\ \bibinfo {pages} {030316} (\bibinfo {year} {2021})}\BibitemShut
	{NoStop}%
	\bibitem [{\citenamefont {Varju}(2013)}]{varju_walks_2013}%
	\BibitemOpen
	\bibfield  {author} {\bibinfo {author} {\bibfnamefont {P.}~\bibnamefont
			{Varju}},\ }\bibfield  {title} {\enquote {\bibinfo {title} {Random walks in
				compact groups},}\ }\href@noop {} {\bibfield  {journal} {\bibinfo  {journal}
			{Doc. Math.}\ }\textbf {\bibinfo {volume} {18}},\ \bibinfo {pages}
		{1137--1175} (\bibinfo {year} {2013})}\BibitemShut {NoStop}%
	\bibitem [{\citenamefont {Nielsen}\ and\ \citenamefont
		{Chuang}(2000)}]{NielsenChuang}%
	\BibitemOpen
	\bibfield  {author} {\bibinfo {author} {\bibfnamefont {M.~A.}\ \bibnamefont
			{Nielsen}}\ and\ \bibinfo {author} {\bibfnamefont {I.~L.}\ \bibnamefont
			{Chuang}},\ }\href@noop {} {\emph {\bibinfo {title} {Quantum computation and
				quantum information}}},\ Cambridge Series on Information and the Natural
	Sciences\ (\bibinfo  {publisher} {Cambridge University Press},\ \bibinfo
	{year} {2000})\BibitemShut {NoStop}%
	\bibitem [{\citenamefont {Guralnick}\ and\ \citenamefont
		{Tiep}(2005)}]{guralnick_larsen_2005}%
	\BibitemOpen
	\bibfield  {author} {\bibinfo {author} {\bibfnamefont {R.~M.}\ \bibnamefont
			{Guralnick}}\ and\ \bibinfo {author} {\bibfnamefont {P.~H.}\ \bibnamefont
			{Tiep}},\ }\bibfield  {title} {\enquote {\bibinfo {title} {{Decompositions of
					small tensor powers and Larsen's conjecture}},}\ }\href@noop {} {\bibfield
		{journal} {\bibinfo  {journal} {Represen. Theory}\ }\textbf {\bibinfo
			{volume} {9}},\ \bibinfo {pages} {138--208} (\bibinfo {year}
		{2005})}\BibitemShut {NoStop}%
	\bibitem [{\citenamefont {Low}(2010)}]{low_pseudo-randomness_2010}%
	\BibitemOpen
	\bibfield  {author} {\bibinfo {author} {\bibfnamefont {R.~A.}\ \bibnamefont
			{Low}},\ }\bibfield  {title} {\enquote {\bibinfo {title} {Pseudo-randomness
				and {Learning} in {Quantum} {Computation}},}\ }\href
	{http://arxiv.org/abs/1006.5227} {\bibfield  {journal} {\bibinfo  {journal}
			{arXiv:1006.5227 [quant-ph]}\ } (\bibinfo {year} {2010})},\ \bibinfo {note}
	{arXiv: 1006.5227}\BibitemShut {NoStop}%
	\bibitem [{\citenamefont {Klaus}(1995)}]{klausthesis}%
	\BibitemOpen
	\bibfield  {author} {\bibinfo {author} {\bibfnamefont {Stephan}\ \bibnamefont
			{Klaus}},\ }\href@noop {} {\emph {\bibinfo {title} {Brown-Kervaire
				invariants}}}\ (\bibinfo  {publisher} {Shaker},\ \bibinfo {year}
	{1995})\BibitemShut {NoStop}%
	\bibitem [{\citenamefont {Watrous}(2018)}]{watrous2018theory}%
	\BibitemOpen
	\bibfield  {author} {\bibinfo {author} {\bibfnamefont {John}\ \bibnamefont
			{Watrous}},\ }\href@noop {} {\emph {\bibinfo {title} {The theory of quantum
				information}}}\ (\bibinfo  {publisher} {Cambridge university press},\
	\bibinfo {year} {2018})\BibitemShut {NoStop}%
	\bibitem [{\citenamefont {Brown}\ and\ \citenamefont
		{Viola}()}]{brown_convergence_2010}%
	\BibitemOpen
	\bibfield  {author} {\bibinfo {author} {\bibfnamefont {W.~G.}\ \bibnamefont
			{Brown}}\ and\ \bibinfo {author} {\bibfnamefont {L.}~\bibnamefont {Viola}},\
	}\bibfield  {title} {\enquote {\bibinfo {title} {Convergence rates for
				arbitrary statistical moments of random quantum circuits},}\ }\href@noop {}
	{\bibfield  {journal} {\bibinfo  {journal} {Phys. Rev. Lett.}\ }\textbf
		{\bibinfo {volume} {104}},\ \bibinfo {pages} {250501}}\BibitemShut {NoStop}%
	\bibitem [{\citenamefont {Diaconis}\ and\ \citenamefont
		{Saloff-Coste}(1993)}]{diaconis_random_1993}%
	\BibitemOpen
	\bibfield  {author} {\bibinfo {author} {\bibfnamefont {P.}~\bibnamefont
			{Diaconis}}\ and\ \bibinfo {author} {\bibfnamefont {L.}~\bibnamefont
			{Saloff-Coste}},\ }\bibfield  {title} {\enquote {\bibinfo {title} {Comparison
				techniques for random walk on finite groups},}\ }\href@noop {} {\bibfield
		{journal} {\bibinfo  {journal} {Ann. Probab.}\ }\textbf {\bibinfo {volume}
			{21}},\ \bibinfo {pages} {2131--2156} (\bibinfo {year} {1993})}\BibitemShut
	{NoStop}%
	\bibitem [{\citenamefont {Nachtergaele}(1996)}]{nachtergaele_gap_1994}%
	\BibitemOpen
	\bibfield  {author} {\bibinfo {author} {\bibfnamefont {B.}~\bibnamefont
			{Nachtergaele}},\ }\bibfield  {title} {\enquote {\bibinfo {title} {The
				spectral gap for some spin chains with disrete symmetry breaking},}\
	}\href@noop {} {\bibfield  {journal} {\bibinfo  {journal} {Commun. Math.
				Phys.}\ }\textbf {\bibinfo {volume} {175}},\ \bibinfo {pages} {565--606}
		(\bibinfo {year} {1996})}\BibitemShut {NoStop}%
	\bibitem [{\citenamefont {Aaronson}\ and\ \citenamefont
		{Gottesman}(2004{\natexlab{b}})}]{aaronson_clifford_2004}%
	\BibitemOpen
	\bibfield  {author} {\bibinfo {author} {\bibfnamefont {S.}~\bibnamefont
			{Aaronson}}\ and\ \bibinfo {author} {\bibfnamefont {D.}~\bibnamefont
			{Gottesman}},\ }\bibfield  {title} {\enquote {\bibinfo {title} {Improved
				simulation of stabilizer circuits},}\ }\href@noop {} {\bibfield  {journal}
		{\bibinfo  {journal} {Phys. Rev. A}\ }\textbf {\bibinfo {volume} {70}},\
		\bibinfo {pages} {052328} (\bibinfo {year} {2004}{\natexlab{b}})}\BibitemShut
	{NoStop}%
	\bibitem [{\citenamefont {Bhatia}(2013)}]{bhatia_book}%
	\BibitemOpen
	\bibfield  {author} {\bibinfo {author} {\bibfnamefont {B.}~\bibnamefont
			{Bhatia}},\ }\bibfield  {title} {\enquote {\bibinfo {title} {Matrix
				analysis},}\ }\href@noop {} {\bibfield  {journal} {\bibinfo  {journal}
			{Springer Science \& Business Media}\ }\textbf {\bibinfo {volume} {169}}
		(\bibinfo {year} {2013})}\BibitemShut {NoStop}%
	\bibitem [{\citenamefont {Heinrich}(2021)}]{heinrich_2021}%
	\BibitemOpen
	\bibfield  {author} {\bibinfo {author} {\bibfnamefont {M.}~\bibnamefont
			{Heinrich}},\ }\emph {\bibinfo {title} {On stabiliser techniques and their
			application to simulation and certification of quantum devices}},\ \href
	{https://kups.ub.uni-koeln.de/50465/} {Ph.D. thesis},\ \bibinfo  {school}
	{University of Cologne} (\bibinfo {year} {2021}),\ \bibinfo {note}
	{\url{https://kups.ub.uni-koeln.de/50465/}}\BibitemShut {NoStop}%
	\bibitem [{\citenamefont {Nebe}\ \emph {et~al.}(2001)\citenamefont {Nebe},
		\citenamefont {Rains},\ and\ \citenamefont {Sloane}}]{nebe_clifford_2001}%
	\BibitemOpen
	\bibfield  {author} {\bibinfo {author} {\bibfnamefont {G.}~\bibnamefont
			{Nebe}}, \bibinfo {author} {\bibfnamefont {E.~M.}\ \bibnamefont {Rains}}, \
		and\ \bibinfo {author} {\bibfnamefont {N.~J.~A}\ \bibnamefont {Sloane}},\
	}\bibfield  {title} {\enquote {\bibinfo {title} {The invariants of the
				{Clifford} groups},}\ }\href@noop {} {\  (\bibinfo {year} {2001})},\ \Eprint
	{http://arxiv.org/abs/math/0001038v2} {arXiv:math/0001038v2} \BibitemShut
	{NoStop}%
	\bibitem [{\citenamefont {Bourgain}\ and\ \citenamefont
		{Gamburd}(2011)}]{bourgain_spectral_2011}%
	\BibitemOpen
	\bibfield  {author} {\bibinfo {author} {\bibfnamefont {J.}~\bibnamefont
			{Bourgain}}\ and\ \bibinfo {author} {\bibfnamefont {A.}~\bibnamefont
			{Gamburd}},\ }\bibfield  {title} {\enquote {\bibinfo {title} {A {spectral}
				{gap} {theorem} in {SU}$(d)$},}\ }\href@noop {} {\  (\bibinfo {year}
		{2011})},\ \bibinfo {note} {arXiv: 1108.6264}\BibitemShut {NoStop}%
	\bibitem [{\citenamefont {Mezher}\ \emph {et~al.}(2019)\citenamefont {Mezher},
		\citenamefont {Ghalbouni}, \citenamefont {Dgheim},\ and\ \citenamefont
		{Markham}}]{mezher2019}%
	\BibitemOpen
	\bibfield  {author} {\bibinfo {author} {\bibfnamefont {R.}~\bibnamefont
			{Mezher}}, \bibinfo {author} {\bibfnamefont {J.}~\bibnamefont {Ghalbouni}},
		\bibinfo {author} {\bibfnamefont {J.}~\bibnamefont {Dgheim}}, \ and\ \bibinfo
		{author} {\bibfnamefont {D.}~\bibnamefont {Markham}},\ }\bibfield  {title}
	{\enquote {\bibinfo {title} {Efficient approximate unitary t-designs from
				partially invertible universal sets and their application to quantum
				speedup},}\ }\href@noop {} {\bibfield  {journal} {\bibinfo  {journal}
			{arXiv:1905.01504}\ } (\bibinfo {year} {2019})}\BibitemShut {NoStop}%
	\bibitem [{\citenamefont {de~Montmort}(1753)}]{montmort_permutation_1713}%
	\BibitemOpen
	\bibfield  {author} {\bibinfo {author} {\bibfnamefont {P.~R.}\ \bibnamefont
			{de~Montmort}},\ }\bibfield  {title} {\enquote {\bibinfo {title} {Essay
				d'analyse sur les jex de hazard},}\ }\href@noop {} {\bibfield  {journal}
		{\bibinfo  {journal} {seconde edition, Jacque Quillau, Paris}\ } (\bibinfo
		{year} {1753})}\BibitemShut {NoStop}%
	\bibitem [{\citenamefont {Nakata}\ \emph {et~al.}(2017)\citenamefont {Nakata},
		\citenamefont {Hirche}, \citenamefont {Koashi},\ and\ \citenamefont
		{Winter}}]{nakata_efficient_2017}%
	\BibitemOpen
	\bibfield  {author} {\bibinfo {author} {\bibfnamefont {Y.}~\bibnamefont
			{Nakata}}, \bibinfo {author} {\bibfnamefont {C.}~\bibnamefont {Hirche}},
		\bibinfo {author} {\bibfnamefont {M.}~\bibnamefont {Koashi}}, \ and\ \bibinfo
		{author} {\bibfnamefont {A.}~\bibnamefont {Winter}},\ }\bibfield  {title}
	{\enquote {\bibinfo {title} {Efficient {quantum} {pseudorandomness} with
				{nearly} {time}-{independent} {Hamiltonian} {dynamics}},}\ }\href
	{http://link.aps.org/doi/10.1103/PhysRevX.7.021006} {\bibfield  {journal}
		{\bibinfo  {journal} {Physical Review X}\ }\textbf {\bibinfo {volume} {7}}
		(\bibinfo {year} {2017})}\BibitemShut {NoStop}%
	\bibitem [{\citenamefont {Leone}\ \emph {et~al.}(2021)\citenamefont {Leone},
		\citenamefont {Oliviero}, \citenamefont {Zhou},\ and\ \citenamefont
		{Hamma}}]{leone2021quantum}%
	\BibitemOpen
	\bibfield  {author} {\bibinfo {author} {\bibfnamefont {L.}~\bibnamefont
			{Leone}}, \bibinfo {author} {\bibfnamefont {S.~F.~E.}\ \bibnamefont
			{Oliviero}}, \bibinfo {author} {\bibfnamefont {Y.}~\bibnamefont {Zhou}}, \
		and\ \bibinfo {author} {\bibfnamefont {A.}~\bibnamefont {Hamma}},\ }\bibfield
	{title} {\enquote {\bibinfo {title} {Quantum chaos is quantum},}\
	}\href@noop {} {\bibfield  {journal} {\bibinfo  {journal} {Quantum}\ }\textbf
		{\bibinfo {volume} {5}},\ \bibinfo {pages} {453} (\bibinfo {year}
		{2021})}\BibitemShut {NoStop}%
	\bibitem [{\citenamefont {Br{\"o}cker}\ and\ \citenamefont
		{Dieck}()}]{brocker_representations_1985}%
	\BibitemOpen
	\bibfield  {author} {\bibinfo {author} {\bibfnamefont {T.}~\bibnamefont
			{Br{\"o}cker}}\ and\ \bibinfo {author} {\bibfnamefont {T.}~\bibnamefont
			{Dieck}},\ }\href {https://www.springer.com/de/book/9783540136781} {\emph
		{\bibinfo {title} {Representations of compact Lie groups}}},\ Graduate Texts
	in Mathematics\ (\bibinfo  {publisher} {Springer-Verlag})\BibitemShut
	{NoStop}%
	\bibitem [{\citenamefont {Fulton}\ and\ \citenamefont
		{Harris}()}]{fulton_representation_2004}%
	\BibitemOpen
	\bibfield  {author} {\bibinfo {author} {\bibfnamefont {W.}~\bibnamefont
			{Fulton}}\ and\ \bibinfo {author} {\bibfnamefont {J.}~\bibnamefont
			{Harris}},\ }\href {https://optdoi.org/10.1007/978-1-4612-0979-9_2} {\emph
		{\bibinfo {title} {Representation theory}}},\ edited by\ \bibinfo {editor}
	{\bibfnamefont {W.}~\bibnamefont {Fulton}}\ and\ \bibinfo {editor}
	{\bibfnamefont {J.}~\bibnamefont {Harris}},\ Graduate Texts in Mathematics\
	(\bibinfo  {publisher} {Springer})\BibitemShut {NoStop}%
	\bibitem [{\citenamefont {Goodman}\ and\ \citenamefont
		{Wallach}()}]{goodman_symmetry_2009}%
	\BibitemOpen
	\bibfield  {author} {\bibinfo {author} {\bibfnamefont {R.}~\bibnamefont
			{Goodman}}\ and\ \bibinfo {author} {\bibfnamefont {N.~R.}\ \bibnamefont
			{Wallach}},\ }\href {https://optdoi.org/10.1007/978-0-387-79852-3_1} {\emph
		{\bibinfo {title} {Symmetry, representations, and invariants}}},\ edited by\
	\bibinfo {editor} {\bibfnamefont {R.}~\bibnamefont {Goodman}}\ and\ \bibinfo
	{editor} {\bibfnamefont {N.~R.}\ \bibnamefont {Wallach}},\ Graduate Texts in
	Mathematics\ (\bibinfo  {publisher} {Springer})\BibitemShut {NoStop}%
	\bibitem [{\citenamefont {{Zhu}}\ \emph {et~al.}(2016)\citenamefont {{Zhu}},
		\citenamefont {{Kueng}}, \citenamefont {{Grassl}},\ and\ \citenamefont
		{{Gross}}}]{ZhuKueGra16}%
	\BibitemOpen
	\bibfield  {author} {\bibinfo {author} {\bibfnamefont {H.}~\bibnamefont
			{{Zhu}}}, \bibinfo {author} {\bibfnamefont {R.}~\bibnamefont {{Kueng}}},
		\bibinfo {author} {\bibfnamefont {M.}~\bibnamefont {{Grassl}}}, \ and\
		\bibinfo {author} {\bibfnamefont {D.}~\bibnamefont {{Gross}}},\ }\href@noop
	{} {\enquote {\bibinfo {title} {{The Clifford group fails gracefully to be a
					unitary 4-design}},}\ } (\bibinfo {year} {2016}),\ \bibinfo {note}
	{arXiv:1609.08172}\BibitemShut {NoStop}%
	\bibitem [{\citenamefont {Folland}(2001)}]{folland2001integrate}%
	\BibitemOpen
	\bibfield  {author} {\bibinfo {author} {\bibfnamefont {G.~B.}\ \bibnamefont
			{Folland}},\ }\bibfield  {title} {\enquote {\bibinfo {title} {How to
				integrate a polynomial over a sphere},}\ }\href@noop {} {\bibfield  {journal}
		{\bibinfo  {journal} {The American Mathematical Monthly}\ }\textbf {\bibinfo
			{volume} {108}},\ \bibinfo {pages} {446--448} (\bibinfo {year}
		{2001})}\BibitemShut {NoStop}%
\end{thebibliography}
%

\end{document}